\documentclass[a4paper,USenglish,cleveref,thm-restate,numberwithinsect,appendix]{lipics-v2021}
\nolinenumbers

\usepackage{style/notation}
\usepackage{style/appendix}
\usepackage{mathpartir}
\usepackage{mathtools}
\usepackage{tikz}
\usetikzlibrary{shapes.symbols,positioning,decorations.markings,arrows.meta}
\usepackage{xspace}

\NewDocumentCommand{\appname}{s}{%
  \IfBooleanTF#1%
    {the extended version of this paper~\cite{TODO}}%
    {the extended paper}%
}

\NewDocumentCommand{\TODO}{m m}%
  {{\bfseries\color{#1}[#2]}}%

\newcommand\numberthis{\addtocounter{equation}{1}\tag{\theequation}}

\Crefname{theorem}{Thm.}{Thms.}
\Crefname{figure}{Fig.}{Figs.}
\Crefname{section}{Sec.}{Secs.}
\Crefname{definition}{Def.}{Defs.}
\Crefname{lemma}{Lem.}{Lems.}
\Crefname{corollary}{Cor.}{Cors.}
\title{Information Flow Control in Cyclic Process Networks}

\author{Bas van den Heuvel}{
  HKA Karlsruhe and University of Freiburg, DE and University of Groningen, NL
}{
  vdheuvel@informatik.uni-freiburg.de
}{
  https://orcid.org/0000-0002-8264-7371
}{Supported in part by the Dutch Research Council (NWO) under project No.\ 016.Vidi.189.046 (Unifying Correctness for Communicating Software).
}
\author{Farzaneh Derakhshan}{
  Illinois Institutie of Technology, USA
}{
  fderakhshan@iit.edu
}{
  https://orcid.org/0000-0002-2156-2606
}{% Funding ack
}
\author{Stephanie Balzer}{
  Carnegie Mellon University, USA
}{
  balzers@cs.cmu.edu
}{
  https://orcid.org/0000-0002-8347-3529
}{Supported in part by the Air Force Office of Scientific Research under award number FA9550-21-1-0385 (Tristan Nguyen, program manager).
Any opinions, findings and conclusions or recommendations expressed here are those of the author(s) and do not necessarily reflect the views of the U.S. Department
of Defense.
}

\authorrunning{Van den Heuvel, Derakhshan, and Balzer}

\Copyright{\textellipsis}

\begin{document}

\maketitle

\begin{abstract}
  Protection of confidential data is an important security consideration of today's applications.
Of particular concern is to guard against unintentional leakage to a (malicious) observer,
who may interact with the program and draw inference from made observations.
Information flow control (IFC) type systems address this concern by statically ruling out such leakage.
This paper contributes an IFC type system for message-passing concurrent programs,
the computational model of choice for many of today's applications such as cloud computing and IoT applications.
Such applications typically either implicitly or explicitly codify protocols according to which message exchange must happen,
and to statically ensure protocol safety, behavioral type systems such as session types can be used.
This paper marries IFC with session typing
and contributes over prior work in the following regards:
(1) support of realistic cyclic process networks as opposed to the restriction to tree-shaped networks,
(2) more permissive, yet entirely secure, IFC control, exploiting cyclic process networks, and
(3) considering deadlocks as another form of side channel, and asserting deadlock-sensitive noninterference (DSNI) for well-typed programs.
To prove DSNI, the paper develops a novel logical relation that accounts for cyclic process networks.
The logical relation is rooted in linear logic, but drops the tree-topology restriction imposed by prior work.

\end{abstract}

% !TeX root = ../main.tex
\section{Introduction}
\label{s:intro}

Many of today's emerging applications and systems
such as cloud computing and IoT applications
are inherently \emph{concurrent} and \emph{message passing}.
Message passing also enjoys popularity in mainstream languages such as Erlang, Go, and Rust.
Similar to functional languages with the $\lambda$-calculus as their theoretical model,
the model of message-passing concurrent languages is the process calculus~\cite{HoareBOOK1985,MilnerBook1980,MilnerBook1989}.
A program in this setting amounts to a number of \emph{processes} connected by \emph{channels},
which compute by exchanging messages along these channels,
rather than by $\beta$-reductions or writing to and reading from shared memory.
Messages may even include channels themselves,
a feature supported in the $\pi$-calculus~\cite{MilnerBook1999,SangiorgiWalkerBook2001}
and referred to as \emph{higher-order} message passing.

Originally untyped~\cite{MilnerBook1999}, the $\pi$-calculus has gradually been enriched with types
to prescribe the kinds of messages that can be exchanged over a channel~\cite{SangiorgiWalkerBook2001}
and to assert correctness properties, such as deadlock freedom and data-race freedom~\cite{KobayashiLICS1997,IgarashiPOPL2001,IgarashiARTICLE2004,conf/concur/Kobayashi06}.
Following in these footsteps,
\emph{session types}~\cite{HondaCONCUR1993,HondaESOP1998} were conceived
to additionally express the \emph{protocols} underlying the exchange.
Session types rely on a \emph{linear} treatment of channels to model the state transitions induced by a protocol,
which was even substantiated by a Curry-Howard correspondence
between the session-typed $\pi$-calculus and
linear logic~\cite{CairesCONCUR2010,WadlerICFP2012,journal/jfp/Wadler14,ToninhoESOP2013,LindleyMorrisESOP2015,CairesARTICLE2016,LindleyMorrisICFP2016}.
Session types based on linear logic enjoy strong properties,
comprising not only race and deadlock freedom but also protocol fidelity.

Security is another correctness consideration arising from today's applications and systems.
One security concern in particular is the protection of confidential information,
by preventing unintentional leakage to a (malicious) observer,
who may interact with the program and draw inference from made observations.
Type systems for \emph{information flow control (IFC)} rule out such leakage by type
checking~\cite{VolpanoARTICLE1996,SmithVolpanoPOPL1998,SabelfeldIEE2003},
given a lattice over security levels and the labeling of observables (e.g., output, locations, channels) with these levels.
Well-typed programs then prevent "flows from high to low" and
guarantee \emph{noninterference}, i.e., that an observer cannot infer any secrets from made observations.
To guarantee noninterference, advanced proof methods
such as logical relations~\cite{PittsStarkHOOTS1998,AhmedPOPL2009,DreyerICFP2010,NeisJFP2011,HurDreyerPOPL2011,ThamsborgBirkedalICFP2011}
and bisimulations~\cite{KoutavasWandPOPL2006,SumiiPierceARTICLE2007,StovringLassenPOPL2007,SangiorgiLICS2007} are used.
If side channels~\cite{SabelfeldIEE2003}, such as the termination channel, are present,
then the literature distinguishes \emph{progress-sensitive} noninterference (PSNI)
from \emph{progress-insensitive} noninterference (PINI),
where the former only equates a divergent program run with another diverging one,
whereas the latter equates a divergent program run with any other run~\cite{HeidinSabelfeldMartkoberdorf2011}.

Whereas the development of IFC type systems has been an active research field for decades
for imperative and predominantly sequential languages,
their exploration in a concurrent, message-passing setting has been more confined to typed process
calculi~\cite{HondaESOP2000,HondaYoshidaPOPL2002,CrafaFMSE2006,HENNESSY20053,KobayashiARTICLE2005,ZDANCEWIC2003,POTTIER2002}
and multiparty session types~\cite{CapecchiCONCUR2010,CapecchiARTICLE2014,CastellaniARTICLE2016,Ciancaglini2016}.
Only recently, IFC has been adopted for session types based on linear logic~\cite{DerakhshanLICS2021,report/BalzerDHY23}.
The resulting type systems exploit the strong guarantees arising from linear logic,
which in particular curtail the network of processes arising at runtime to a tree structure.
However, many real-world application scenarios
are precluded from an insistence on a tree structure, instead requiring support of \emph{cyclic process networks}.
Session type systems~\cite{BalzerICFP2017,BalzerESOP2019,conf/fossacs/DardhaG18,conf/ice/vdHeuvelP21,journal/scico/vdHeuvelP22}
that allow for cyclic process networks increase expressivity while remaining rooted in linear logic.

This paper scales IFC to cyclic process networks
and contributes an IFC type system for an asynchronous $\pi$-calculus with linear session types.
To prove that well-typed processes in the resulting language enjoy noninterference,
we develop a novel logical relation.
Our development was challenged by the possibility of \emph{deadlocks} that can arise in cyclic process networks
and that constitute another form of side channel.
To rule out side-channel attacks due to deadlocks,
we introduce the notion of \emph{deadlock-sensitive noninterference~(DSNI)},
which only equates a deadlocking program with another deadlocking one.
Using our logical relation we  prove that well-typed processes in our language enjoy DSNI (fundamental theorem).

Cyclic process networks also turn out to be beneficial for IFC,
as they permit secure programs that are rejected by existing IFC type systems for linear session types~\cite{DerakhshanLICS2021,report/BalzerDHY23}.
These are programs that exploit the possibility of setting up several channels---rather than just one channel---between two processes,
to separate low-security from high-security communication.

\subparagraph*{Contributions.}
Our contributions are threefold:
\begin{enumerate}
    \item 
        An IFC session type system for an asynchronous $\pi$-calculus with support for cyclic process networks (\Cref{s:IFC}), that satisfies protocol fidelity and communication safety (\Cref{t:subjRed}).
        
    \item
        A logical relation that induces an equivalence between typed processes (\Cref{s:logRel}), defining our notion of DSNI (\Cref{d:DSNIRel}).
        
    \item
        The main result that well typedness implies DSNI (\Cref{t:DSNI} in \Cref{s:DSNI}), following from the fundamental theorem (\Cref{t:fundamental}).

\end{enumerate}

\subparagraph*{Outline.}

In addition to the above contributions, \Cref{s:ideas} gently introduces the key ideas behind the developments in this paper, \Cref{s:rw} discusses related work, and \Cref{s:concl} concludes the paper.
Important proofs are detailed in part; remaining details, proofs, and auxiliary definitions are given in \appref*.

% !TeX root = ../main.tex
\section{Key Ideas}
\label{s:ideas}

In this section, we discuss the key ideas behind our contributions.

\subsection{Cyclic Process Networks Afford Flexible Information Flow Control}
\label{s:ideas:example}

We motivate our developments through a high-level example, focusing on how cyclicity in process networks improves over prior works by increasing the flexibility of information flow control to support more realistic scenarios.

\subparagraph*{Collaborating governments.}

Consider two governments that want to collaborate on scientific and intelligence efforts.
Clearly, the interactions between a government and an intelligence agency is confidential, whereas interactions between a government and the scientific community is not; intelligence may not leak to the scientific community, where there may be spies.

\def\High{$\fIFC{H}$igh\xspace}
\def\Low{$\fIFC{L}$ow\xspace}
We make this more precise by establishing that information can be of \High or \Low confidentiality, and that communication channels can be of \High or \Low security.
Clearly, information of \Low confidentiality can be transmitted over \High-security channels, but not vice versa.
We identify our two governments as $\fProc{X} \in \{\fProc{A},\fProc{B}\}$, each with departments (processes) $\fProc{Gov_X},\fProc{Int_X},\fProc{Science_X}$.
We consider three scenarios, each of which connects these processes to form different process networks. 

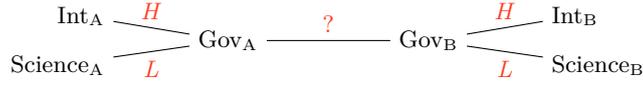
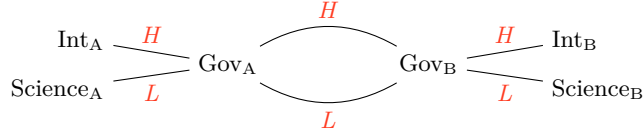
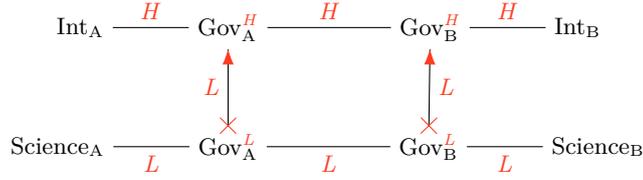
\begin{figure}[t]
    \begin{subfigure}{\textwidth}
        \begin{center}
            {\small\begin{tikzpicture}
                \node (GovA) {$\fProc{Gov_A}$};
                \node [left=of GovA,yshift=1em] (IntA) {$\fProc{Int_A}$};
                \node [left=of GovA,yshift=-1em] (ScienceA) {$\fProc{Science_A}$};
                \node [right=5em of GovA] (GovB) {$\fProc{Gov_B}$};
                \node [right=of GovB,yshift=1em] (IntB) {$\fProc{Int_B}$};
                \node [right=of GovB,yshift=-1em] (ScienceB) {$\fProc{Science_B}$};
                \draw (GovA) to node [above] {$\fIFC{?}$} (GovB);
                \draw (GovA) to node [above] {$\fIFC{H}$} (IntA);
                \draw (GovA) to node [below] {$\fIFC{L}$} (ScienceA);
                \draw (GovB) to node [above] {$\fIFC{H}$} (IntB);
                \draw (GovB) to node [below] {$\fIFC{L}$} (ScienceB);
            \end{tikzpicture}}
        \end{center}
        \caption{Scenario~1: No cyclicity.}\label{f:ideas:example:1}
    \end{subfigure}
    
    \begin{subfigure}{\textwidth}
        \begin{center}
            {\small\begin{tikzpicture}
                \node (GovA) {$\fProc{Gov_A}$};
                \node [left=of GovA,yshift=1em] (IntA) {$\fProc{Int_A}$};
                \node [left=of GovA,yshift=-1em] (ScienceA) {$\fProc{Science_A}$};
                \node [right=5em of GovA] (GovB) {$\fProc{Gov_B}$};
                \node [right=of GovB,yshift=1em] (IntB) {$\fProc{Int_B}$};
                \node [right=of GovB,yshift=-1em] (ScienceB) {$\fProc{Science_B}$};
                \draw [bend left] (GovA) to node [above,red] {$\fIFC{H}$} (GovB);
                \draw [bend right] (GovA) to node [below,red] {$\fIFC{L}$} (GovB);
                \draw (GovA) to node [above,red] {$\fIFC{H}$} (IntA);
                \draw (GovA) to node [below,red] {$\fIFC{L}$} (ScienceA);
                \draw (GovB) to node [above,red] {$\fIFC{H}$} (IntB);
                \draw (GovB) to node [below,red] {$\fIFC{L}$} (ScienceB);
            \end{tikzpicture}}
        \end{center}
        \caption{Scenario~2: Doubly connected governments.}\label{f:ideas:example:2}
    \end{subfigure}
    
    \begin{subfigure}{\textwidth}
        \begin{center}
           {\small \begin{tikzpicture}l
                \tikzset{arrowx/.style={postaction={decorate,decoration={markings,mark=at position 0.1 with {\arrow{Rays[length=3mm,color=RedOrange]};},mark=at position 1 with {\arrow{Latex[line width=0pt,length=2mm,fill=RedOrange]};},},},}}
                \node (GovAH) {$\fProc{Gov_A^{\fIFC{H}}}$};
                \node [below=of GovAH] (GovAL) {$\fProc{Gov_A^{\fIFC{L}}}$};
                \node [left=of GovAH] (IntA) {$\fProc{Int_A}$};
                \node [left=of GovAL] (ScienceA) {$\fProc{Science_A}$};
                \node [right=5em of GovAH] (GovBH) {$\fProc{Gov_B^{\fIFC{H}}}$};
                \node [right=5em of GovAL] (GovBL) {$\fProc{Gov_B^{\fIFC{L}}}$};
                \node [right=of GovBH] (IntB) {$\fProc{Int_B}$};
                \node [right=of GovBL] (ScienceB) {$\fProc{Science_B}$};
                \draw [arrowx] (GovAL) to node [left] {$\fIFC{L}$} (GovAH);
                \draw [arrowx] (GovBL) to node [right] {$\fIFC{L}$} (GovBH);
                \draw (GovAH) to node [above] {$\fIFC{H}$} (GovBH);
                \draw (GovAL) to node [below] {$\fIFC{L}$} (GovBL);
                \draw (GovAH) to node [above] {$\fIFC{H}$} (IntA);
                \draw (GovAL) to node [below] {$\fIFC{L}$} (ScienceA);
                \draw (GovBH) to node [above] {$\fIFC{H}$} (IntB);
                \draw (GovBL) to node [below] {$\fIFC{L}$} (ScienceB);
            \end{tikzpicture}}
        \end{center}
        \caption{Scenario~3: Extended cyclicity for more flexible IFC.}\label{f:ideas:example:3}
    \end{subfigure}   
    \caption{Collaborating governments: three scenarios.}
\end{figure}

\subparagraph*{Scenario~1: No cyclicity.}

In \Cref{f:ideas:example:1}, the governments have only one channel to communicate on.
Lines between departments denote communication channels, and their annotations indicate security levels.
In this scenario, no processes are cyclically connected.

Notice that the channel connecting the two governments does not have a security level assigned.
If we insist on intelligence exchange (i.e., on exchanging \High-confidentiality information), this channel must be of \High security.
However, this inhibits scientific exchange: when a government receives information on a \High-security channel, it cannot guarantee that the information is of \Low confidentiality, so it cannot share it with its scientific department over a \Low-security channel.
Hence, the single channel of communication between the governments is unrealistic.

\subparagraph*{Scenario~2: Doubly-connected governments.}

In \Cref{f:ideas:example:2}, we attempt to remedy this problem by adding a second channel of \Low security between the governments.
Since the governments are connected on two separate channels, they are cyclically connected.

Now the governments can exchange scientific information and share this with their intelligence agencies.
However, it is conceivable that a government makes decisions about which scientific information to share based on intelligence information.
A clever spy may then be able to infer intelligence information from scientific information, \emph{indirectly}.
Hence, once a government receives intelligence information, it should refrain from sharing scientific information.
Clearly, this scenario is still not realistic.

\subparagraph*{Scenario~3: Extended cyclicity for more flexible IFC.}

In \Cref{f:ideas:example:3}, we split our governments into \High- and \Low-confidentiality departments.
The \High-confidentiality departments share intelligence information, and the \Low-confidentiality departments share scientific information.
Crucially, the \Low-confidentiality department can share information with the \High-confidentiality department, but not vice versa.

\subsection{Threats to Noninterference due to Deadlocks}
\label{s:ideas:threats}

When process networks contain cyclic connections, there is a risk of \emph{deadlocked} communication.
For example, consider again the process network in \Cref{f:ideas:example:2}.
Let us refer to the \High-security channel between the governments as $\fProc{h}$, and the \Low-security channel between the governments as $\fProc{l}$.
Suppose the two governments are implemented as follows (in pseudocode):
\[
    \fProc{Gov_A} := \fProc{receive\ on\ h ; send\ on\ l}
    \hfill
    \fProc{Gov_B} := \fProc{receive\ on\ l ; send\ on\ h}
    \hfill
\]
The communication between the governments is deadlocked: each government is waiting to receive from the other, but the corresponding sends are blocked.

In process networks without cyclicity, this kind of deadlock does not occur.
There are ways to prevent them through typing (cf., e.g., \cite{conf/fossacs/DardhaG18,journal/scico/vdHeuvelP22,BalzerESOP2019}),
but the possible occurrence of deadlock introduced by cyclicity is realistic and a possible threat to noninterference.
To see how, consider another pseudocode implementation, where $\fProc{s_A}$ refers to the \Low-security channel between $\fProc{Gov_A}$ and $\fProc{Science_A}$:
 \begin{align}
    \fProc{Gov_A} := \fProc{receive\ x\ on\ h ; if\ x == true\ then\ send\ on\ s_A\ else\ deadlock}
    \label{eq:ideas:DS}
\end{align}
In this scenario, a spy monitoring the information exchanged on $\fProc{s_A}$ is indirectly able to infer \High-confidentiality information: if information is sent on $\fProc{s_A}$, then the spy knows for sure that the value of x is true.
This is why we are after IFC for noninterference that is \emph{deadlock sensitive}.

\subsection{IFC Type System in a Nutshell}
\label{s:ideas:IFC}

In this paper, we implement IFC similarly as in previous works: by enriching a session type system with IFC annotations and requirements.
We build on the session-typed asynchronous variant of the $\pi$-calculus of Van den Heuvel and Pérez~\cite{conf/ice/vdHeuvelP21,journal/scico/vdHeuvelP22} (stripped from the ``priority'' mechanisms that rule out deadlock).

As anticipated in \Cref{s:ideas:example}, channels are appointed \emph{maximum-secrecy levels} (secrecy levels, for short), indicating the maximum secrecy of messages that can be sent on channels securely. For example, in \Cref{f:ideas:example:2}, the channel between $\fProc{Int_A}$ and $\fProc{Gov_A}$ has a secrecy level of \High, while the channel between $\fProc{Science_A}$ and $\fProc{Gov_A}$ has a secrecy level of \Low. This indicates that a spy with low-level security clearance can observe the messages sent on channels of secrecy level \Low but not \High.
A partial order on these secrecy levels forms a \emph{secrecy lattice}; for example, $\fIFC{L \lleq H}$ (\Low is lower than \High).

As processes receive messages on channels, they learn ``secrets'', possibly influencing the information sent in future messages (referred to in the literature as \emph{flow sensitivity}~\cite{SabelfeldIEE2003}).
Key in our IFC is thus that it is forbidden to send messages on a channel if the level of secrecy learned so far exceeds the secrecy level of the channel.
To ensure this, our type system assigns to each process a \emph{running secrecy} that increases as the process receives higher-secrecy-level information.
As processes evolve, the secrecy levels of their channels do not change, whereas their running secrecies do.

For example, in \Cref{f:ideas:example:2}, suppose $\fProc{Gov_A}$ starts with a \Low running secrecy, thus being able to send messages to both $\fProc{Int_A}$ and $\fProc{Science_A}$.
After receiving a message from $\fProc{Int_A}$, the running secrecy of $\fProc{Gov_A}$ becomes \High: the secrecy level of the channel between them is \High.
Hence, after this message, $\fProc{Gov_A}$ can no longer send messages to $\fProc{Science_A}$.

Finally, we need to address how our IFC handles deadlock sensitivity, as introduced in \Cref{s:ideas:threats}.
It turns out that it is sufficient to rely on running secrecies and their dynamics as described above.
To see how, consider again the implementation of $\fProc{Gov_A}$ in~\eqref{eq:ideas:DS}.
Assuming it starts with \Low running secrecy, the process receives on a \High-security channel, so its running secrecy becomes \High.
Our IFC then disallows the process from sending on the \Low security channel.
Hence, this example would be considered ill typed in our type system.

\subsection{Logical Relation for DSNI in a Nutshell}
\label{s:ideas:logRel}

Let us be more precise in what we mean by DSNI.
A process may have a number of ``unconnected'' channels.
By connecting these channels to other processes, we create a \emph{context} in which to run the process.
We refer to the channels connecting the process to its context as the \emph{interface}. For example, in \Cref{f:ideas:example:2}, $\fProc{Gov_A}$ can be considered as a standalone process, with the rest of the processes being a potential context for it. The interface is then the four channels connecting $\fProc{Gov_A}$ to the other processes.

With DSNI, we assume the existence of an ``attacker'', a more precise definition of the ``spy'' mentioned in \Cref{s:ideas:example}.
This attacker knows the specification of our process, and has the ability to observe messages from and to the process over \emph{observable channels}: channels in the interface that have secrecy levels up to a given secrecy level~$\fIFC{\xi}$.
Moreover, the attacker cannot measure time but can observe the relative order in which messages are sent through different channels.
By running our process in different contexts and observing the messages on observable channels, the attacker may be able to use its knowledge of the process' specification to infer information about messages on unobservable channels.
As such, noninterference means that the attacker is not able to do so; in our case, we are after DSNI, because we do not want the attacker to infer information from deadlocks either.

In this paper, we define DSNI as an \emph{equivalence} between the behavior on observable channels of the same process in different contexts.
This equivalence is defined by means of a \emph{logical relation}.
The relation scrutinizes messages from and to the process on observable channels, and ``ignores'' messages on unobservable channels.
Our main result is that well typedness implies DSNI (\Cref{t:DSNI}).

\subsection{Technical Challenges}
\label{s:ideas:challenges}

The subsequent sections first introduce our process language and then develop an IFC type system for that language and state and prove DSNI using a logical relation.  These sections are naturally quite technical.  To bridge the divide, we briefly survey here the main challenges our development had to overcome.

\subparagraph*{Asynchronous communication}

Our process language is an \emph{asynchronous} $\pi$-calculus with linear session types, based on Van den Heuvel and Pérez's Asynchronous Priority-based Classical Processes (APCP)~\cite{conf/ice/vdHeuvelP21,journal/scico/vdHeuvelP22}, but without recursion and ``priority'' mechanisms (which prevents deadlocks).
As in the asynchronous $\pi$-calculus, outputs in our calculus do not have any continuations but are atomic processes composed in parallel with other processes.
To model session sequencing, a process must adopt a \emph{continuation-passing} style, in the sense that an output not only comprises a message but also a continuation channel.

When continuation channels are part of messages exchanged over an observable channel in the interface between a process and its context,
the question comes up whether the continuation channel becomes observable as well.
The natural impulse might be to consider them observable too.  For sure this is the right choice in linear session-typed process calculi that confine process networks to trees \cite{DerakhshanLICS2021,report/BalzerDHY23}, guaranteeing that the continuation channel sent as part of the message resides within the sending process itself.
However, due to the possibility of cycles in our setting, a continuation channel sent as part of a message may actually reside within the context outside the sending process.  As a result, the logical relation has to consider the binding structure of the process and the context when determining observability of continuation channels.
We detail this case analysis in \Cref{s:logRel} when we introduce the logical relation, with a pictorial illustration in \Cref{f:valRel:diags}.

\subparagraph*{Observable deadlocks}

Deadlock-sensitive noninterference (DSNI) provides a very strong notion of noninterference in that it equates a deadlocking process only with another deadlocking one (as opposed to an arbitrary one).  As a result, it prevents leakage through deadlocks, a side channel similar to the termination channel.
DSNI is asserted by the definition of the logical relation and challenges the proof of \Cref{t:DSNI}, stating that well typedness implies DSNI.
\Cref{t:DSNI} is proved a generalized \emph{fundamental theorem} (\Cref{t:fundamental}), which asserts that all executions of a process, if well typed, are related by the logical relation, up to the secrecy level~$\fIFC{\xi}$ of the observer.
Because this theorem relates two different processes, but with the same observable behavior (where deadlocks are observable), the proof must maintain a tight correspondence between the two processes.
This correspondence is achieved by employing the notion of \emph{relevant nodes} (\Cref{d:relNode}), which are the parts of a process that can have observable outcomes either directly (by sending a message over the interface) or indirectly (by initiating a chain of messages ending with an observable one), and asserting that the relevant nodes of both processes are indistinguishable (up to structural congruence).
Our notion of relevant nodes is inspired by Derakhshan et al.~\cite{DerakhshanLICS2021}, but accounts for cyclic process networks.

\subparagraph*{Structural congruence and alpha equivalence}

Our logical relation makes use of structural congruence to single out the action in a process producing an observable message.
Because structural congruence permits alpha renaming, process relatedness must account for alpha-equivalence classes.
As usual, proofs require a careful treatment of alpha renaming, which additionally becomes more nuanced by the existence of binders for observable names in contexts.
This treatment becomes especially apparent in the so-called \emph{catch-up} lemma (\Cref{l:catchUp}), a lemma used in the proof of the fundamental theorem to assert that two observably equivalent processes can ``catch up'' on each other's unobservable reductions.

% !TeX root = ../main.tex
\section{Linear Session Types for Information Flow Control}
\label{s:IFC}

In this section, we define our information flow control (IFC) type system.
We first introduce our process language (an asynchronous $\pi$-calculus) along with a linear session-type system in \Cref{s:IFC:procLang}.
Then, we enrich the type system with IFC in \Cref{s:IFC:IFC}.
In \Cref{s:IFC:typePres}, we prove that well-typed processes enjoy communication safety and protocol fidelity as corollaries of a type-preservation result.
As we will see in \Cref{s:DSNI}, well typedness in the resulting IFC type system implies noninterference.

\subsection{Process Language: Syntax, Semantics, and Types}
\label{s:IFC:procLang}

Our process language is an asynchronous $\pi$-calculus, where parallel subprocesses communicate on connected channels.
To be precise, we adapt the non-recursive fragment of Van den Heuvel and Pérez's Asynchronous Priority-based Classical Processes (APCP)~\cite{conf/ice/vdHeuvelP21,journal/scico/vdHeuvelP22} by removing their ``priority'' mechanisms that prevent deadlock and adding our IFC.

\subparagraph*{Syntax.}

The syntactic elements of our language are typeset in a \textcolor{black}{black and non-italic} font.
In our language, channels have two distinct endpoints, denoted $\fProc{a},\fProc{b},\fProc{c},\ldots,\fProc{x},\fProc{y},\fProc{z}$ and further referred to as \emph{names}.
By design, all names are used linearly, meaning that they are used for a communication exactly once.

\begin{definition}[Syntax]
    \label{d:procSyntax}
    \emph{Processes} $\fProc{P},\fProc{Q},\fProc{R},\ldots$ are defined by the following syntax:
    \[
        \fProc{P},\fProc{Q},\fProc{R},\ldots ::= \fProc{0} \sepr \fProc{( P \| Q )} \sepr \fProc{\nu{xy} P} \sepr \fProc{\pClose x[]} \sepr \fProc{\pWait x() ; P} \sepr \fProc{\pSel x[b]<j} \sepr \fProc{\pBra x(z)>\{i:P_i\}_{i \in I}} \sepr \fProc{\pSend x[a,b]} \sepr \fProc{\pRecv x(y,z) ; P}
    \]
\end{definition}

We write $\fProc{P \pSubst{ x/y }}$ to denote the capture-avoiding substitution of $\fProc{y}$ for $\fProc{x}$ in $\fProc{P}$.
Process $\fProc{0}$ denotes inaction.
In $\fProc{( P \| Q )}$, processes $\fProc{P}$ and $\fProc{Q}$ run in parallel; we often omit the parentheses.
Restriction $\fProc{\nu{xy} P}$ binds $\fProc{x}$ and $\fProc{y}$ in $\fProc{P}$ to form a channel, enabling communication.

Process $\fProc{\pClose x[]}$ closes the channel to which $\fProc{x}$ belongs, and $\fProc{\pWait x() ; P}$ waits for the channel to close before continuing as $\fProc{P}$.
Selection $\fProc{\pSel x[b]<j}$ sends the label $\fProc{j}$ over $\fProc{x}$ along with a name $\fProc{b}$; we refer to $\fProc{b}$ as the selection's \emph{continuation}, as it provides a means to continue communicating after the selection.
Branch $\fProc{\pBra x(z)>\{ i : P_i \}_{i \in I}}$ waits to receive on $\fProc{x}$ a label $\fProc{j} \in \fProc{I}$ along with a continuation $\fProc{b}$ before continuing as $\fProc{P_j \pSubst{ b/z }}$; this binds $\fProc{z}$ in each $\fProc{P_i}$.
Send $\fProc{\pSend x[a,b]}$ sends names $\fProc{a}$ and $\fProc{b}$ over $\fProc{x}$; we typically refer to $\fProc{a}$ and $\fProc{b}$ as the send's payload and continuation, respectively, but there is no technical distinction between them.
Receive $\fProc{\pRecv x(y,z) ; P}$ waits to receive on $\fProc{x}$ two names $\fProc{a}$ and $\fProc{b}$ before continuing as $\fProc{P \pSubst{ a/y , b/z }}$; this binds $\fProc{y}$ and $\fProc{z}$ in $\fProc{P}$.
All names in a process are free unless bound as described above; we write $\fn(\fProc{P})$ to denote the set of free names of $\fProc{P}$.

\begin{example}
    \label{x:syntax}
    To illustrate process syntax, we further develop the example introduced in \Cref{s:ideas:example}.
    We develop two simple accounts of $\fProc{Gov_A}$: one where information flow is secure, and one where it is not.

    In the first scenario, $\fProc{Gov_A^{\fIFC{L}}}$ passes a research outcome ($\fProc{oc}$) to $\fProc{Gov_A^{\fIFC{H}}}$, which determines a command for $\fProc{Int_A}$:
    \begin{align*}
        \fProc{Gov_A^{\fIFC{L}}} &:= \fProc{\nu{a_{\fIFC{H}}^{1'} a_{\fIFC{H}}^1} ( \pSel a_{\fIFC{H}}[a_{\fIFC{H}}^{1'}]<{oc_2} \| \pClose a_{\fIFC{H}}^1[] )}
        \\
        \fProc{Gov_A^{\fIFC{H}}} &:= \fProc{\pBra a_{\fIFC{L}}(a_{\fIFC{L}}^{1})>\left\{ \begin{array}{@{}l@{}}
            \fProc{oc_1 : \nu{a_I^{1'} a_I^1} ( \pSel a_I[a_I^{1'}]<act \| \pWait a_{\fIFC{L}}^1() ; \pClose a_I^1[] ) ,} \\
            \fProc{oc_2 : \nu{a_I^{1'} a_I^1} ( \pSel a_I[a_I^{1'}]<wait \| \pWait a_{\fIFC{L}}^1() ; \pClose a_I^1[] )}
        \end{array} \right\}}
        \\
        \fProc{Int_A} &:= \fProc{\pBra i_A(i_A^1)>\{ act : \pWait i_A^1() ; 0 , wait : \pWait i_A^1() ; 0 \}}
        \\
        \fProc{A_{secure}} &:= \fProc{\nu{a_{\fIFC{H}} a_{\fIFC{L}}} \nu{a_I i_A} ( Gov_A^{\fIFC{L}} \| Gov_A^{\fIFC{H}} \| Int_A )}
    \end{align*}

    In the second scenario, $\fProc{Gov_A^{\fIFC{H}}}$ receives intelligence ($\fProc{int}$) from $\fProc{Int_A}$ and shares the information ($\fProc{inf}$) with $\fProc{Gov_A^{\fIFC{L}}}$:
    \begin{align*}
        \fProc{Int_A} &:= \fProc{\nu{i_A^{1'} i_A^1} ( \pSel i_A[i_A^1]<int_1 \| \pClose i_A^1[] )}
        \\
        \fProc{Gov_A^{\fIFC{H}}} &:= \fProc{\pBra a_I(a_I^1)>\left\{ \begin{array}{@{}l@{}}
            \fProc{int_1 : \nu{a_{\fIFC{L}}^{1'} a_{\fIFC{L}}^1} ( \pSel a_{\fIFC{L}}[a_{\fIFC{L}}^{1'}]<inf_1 \| \pWait a_I^1() ; \pClose a_{\fIFC{L}}^1[] ) ,} \\
            \fProc{int_2 : \nu{a_{\fIFC{L}}^{1'} a_{\fIFC{L}}^1} ( \pSel a_{\fIFC{L}}[a_{\fIFC{L}}^{1'}]<inf_2 \| \pWait a_I^1() ; \pClose a_{\fIFC{L}}^1[] )}
        \end{array} \right\}}
        \\
        \fProc{Gov_A^{\fIFC{L}}} &:= \fProc{\pBra a_{\fIFC{H}}(a_{\fIFC{H}}^1)>\{ inf_1 : \pWait a_{\fIFC{H}}^1() ; 0 , inf_2 : \pWait a_{\fIFC{H}}^1() ; 0 \}}
        \\
        \fProc{A_{insecure}} &:= \fProc{\nu{a_{\fIFC{H}} a_{\fIFC{L}}} \nu{a_I i_A} ( Gov_A^{\fIFC{L}} \| Gov_A^{\fIFC{H}} \| Int_A )}
    \end{align*}
\end{example}

Variants of the $\pi$-calculus often include the forwarder process $\fProc{\pFwd x y}$ which forwards any communications between $\fProc{x}$ and $\fProc{y}$ by fusing $\fProc{x}$ and $\fProc{y}$.
Here, we choose to omit forwarders for a smoother definition of our logical relation; they can be added as syntactic sugar using \emph{identity expansion} (cf.\ \appref[d:fwd]).

\subparagraph*{Semantics.}

The dynamics of our language is defined in terms of a reduction semantics, where each step represents the synchronization of complementary communications on the two endpoints of a channel.
As usual, reduction relies on structural congruence, which restructures processes without affecting channel connections and the order of communications.

\begin{figure}[t]
    {\small    Structural congruence ($\fProc{P} \sc \fProc{Q}$):
   \begin{mathpar}
        \begin{bussproof}[sc-alpha]
            \bussAssume{
                \fProc{P} \alpheq \fProc{Q}
            }
            \bussUn{
                \fProc{P} \sc \fProc{Q}
            }
        \end{bussproof}
        \and
        \begin{bussproof}[sc-par-nil]
            \bussAx{
                \fProc{P \| 0} \sc \fProc{P}
            }
        \end{bussproof}
        \and
        \begin{bussproof}[sc-par-symm]
            \bussAx{
                \fProc{P \| Q} \sc \fProc{Q \| P}
            }
        \end{bussproof}
        \and
        \begin{bussproof}[sc-par-assoc]
            \bussAx{
                \fProc{( P \| Q ) \| R} \sc \fProc{P \| ( Q \| R )}
            }
        \end{bussproof}
        \and
        \begin{bussproof}[sc-res-symm]
            \bussAx{
                \fProc{\nu{xy} P} \sc \fProc{\nu{yx} P}
            }
        \end{bussproof}
        \and
        \begin{bussproof}[sc-res-assoc]
            \bussAx{
                \fProc{\nu{xy} \nu{zw} P} \sc \fProc{\nu{zw} \nu{xy} P}
            }
        \end{bussproof}
        \and
        \begin{bussproof}[sc-res-comm]
            \bussAssume{
                \fProc{x},\fProc{y} \notin \fProc{\fn(Q)}
            }
            \bussUn{
                \fProc{\nu{xy} ( P \| Q )} \sc \fProc{\nu{xy} P \| Q}
            }
        \end{bussproof}
    \end{mathpar}

    Reduction ($\fProc{P} \redd \fProc{Q}$):
    \begin{mathpar}
        \begin{bussproof}[red-close-wait]
            \bussAx{
                \fProc{\nu{xy} ( \pClose x[] \| \pWait y() ; P )} \redd \fProc{P}
            }
        \end{bussproof}
        \and
        \begin{bussproof}[red-sel-bra]
            \bussAssume{
                \fProc{j} \in \fProc{I}
            }
            \bussUn{
                \fProc{\nu{xy} ( \pSel x[b]<j \| \pBra y(w)>\{ i : Q_i \}_{i \in I} )} \redd \fProc{Q_j \pSubst{ b/w }}
            }
        \end{bussproof}
        \and
        \begin{bussproof}[red-send-recv]
            \bussAx{
                \fProc{\nu{xy} ( \pSend x[a,b] \| \pRecv y(z,w) ; Q )} \redd \fProc{Q \pSubst{ a/z,b/w }}
            }
        \end{bussproof}
        \and
        \begin{bussproof}[red-sc]
            \bussAssume{
                \fProc{P} \sc \fProc{P'}
            }
            \bussAssume{
                \fProc{P'} \redd \fProc{Q'}
            }
            \bussAssume{
                \fProc{Q'} \sc \fProc{Q}
            }
            \bussTern{
                \fProc{P} \redd \fProc{Q}
            }
        \end{bussproof}
        \and
        \begin{bussproof}[red-par]
            \bussAssume{
                \fProc{P} \redd \fProc{P'}
            }
            \bussUn{
                \fProc{P \| Q} \redd \fProc{P' \| Q}
            }
        \end{bussproof}
        \and
        \begin{bussproof}[red-res]
            \bussAssume{
                \fProc{P} \redd \fProc{P'}
            }
            \bussUn{
                \fProc{\nu{xy} P} \redd \fProc{\nu{xy} P'}
            }
        \end{bussproof}
    \end{mathpar}}

    \caption{Structural congruence (top) and reduction (bottom); cf.\ \Cref{d:procSCRedd}.}\label{f:procSCRedd}
\end{figure}

\begin{definition}[Reduction Semantics]
    \label{d:procSCRedd}
    \emph{Structural congruence} is the least congruence on the syntax of processes (i.e., closed under arbitrary process contexts), denoted $\fProc{P} \sc \fProc{Q}$, induced by the axioms in \Cref{f:procSCRedd} (top).

    \emph{Reduction} is a binary relation on processes, denoted $\fProc{P} \redd \fProc{Q}$, defined by the rules in \Cref{f:procSCRedd} (bottom).
    We write $\fProc{P} \nredd$ to denote that there is no $\fProc{Q}$ such that $\fProc{P} \redd \fProc{Q}$.
\end{definition}

Rule~\ruleLabel{sc-alpha} allows alpha conversion, i.e., renaming bound names.
Rule~\ruleLabel{sc-par-nil} defines $\fProc{0}$ as the unit of parallel composition, and Rules~\ruleLabel{sc-par-symm} and~\ruleLabel{sc-par-assoc} define parallel composition as symmetric and associative, respectively.
Rules~\ruleLabel{sc-res-symm} and~\ruleLabel{sc-res-assoc} define symmetry and associativity of restriction, respectively.
Rule~\ruleLabel{sc-res-comm} defines commutativity of restriction, as long as this does not capture or free any names; this is often referred to as \emph{scope extrusion}.

Rules~\ruleLabel{red-close-wait}, \ruleLabel{red-sel-bra}, and~\ruleLabel{red-send-recv} define synchronizations of complementary communications on names connected by restriction; these rules formalize the behavior described below \Cref{d:procSyntax}.
Rules~\ruleLabel{red-sc}, \ruleLabel{red-par}, and~\ruleLabel{red-res} close reduction under structural congruence, parallel composition, and restriction, respectively.

\begin{example}
    We illustrate process semantics on $\fProc{A_{secure}}$ defined in \Cref{x:syntax}.
    We have
    \begin{align*}
        \fProc{A_{secure}} &= \fProc{\nu{a_{\fIFC{H}} a_{\fIFC{L}}} \nu{a_I i_A} ( Gov_A^{\fIFC{L}} \| Gov_A^{\fIFC{H}} \| Int_A )}
        \\
        &\sc \fProc{\nu{a_I i_A} ( \nu{a_{\fIFC{H}}^{1'} a_{\fIFC{H}}^1} ( \nu{a_{\fIFC{H}} a_{\fIFC{L}}} ( \pSel a_{\fIFC{H}}[a_{\fIFC{H}}^{1'}]<oc_2 \| \pBra a_{\fIFC{L}}(a_{\fIFC{L}}^1)>\{ \ldots \} ) \| \pClose a_{\fIFC{H}}^1[] ) \| Int_A )}
        \\
        &\redd \fProc{\nu{a_I i_A} ( \nu{a_{\fIFC{H}}^{1'} a_{\fIFC{H}}^1} ( \nu{a_I^{1'} a_I^1} ( \pSel a_I[a_I^{1'}]<wait \| \pWait a_{\fIFC{H}}^{1'}() ; \pClose a_I^1[] ) \| \pClose a_{\fIFC{H}}^1[] ) \| Int_A )},
    \end{align*}
    from where asynchronous communication enables further communication between $\fProc{a_{\fIFC{H}}^{1'}}$ and $\fProc{a_{\fIFC{H}}^1}$ or between $\fProc{a_I}$ and $\fProc{i_A}$; for example,
    \[
        {} \redd \fProc{\nu{a_I i_A} ( \nu{a_I^{1'} a_I^1} ( \pSel a_I[a_I^{1'}]<wait \| \pClose a_I^1[] ) \| Int_A )}.
    \]
\end{example}

\subparagraph*{Types.}

We use linear session types to ``tame'' our processes.
The system we use is derived from classical linear logic, so types are expressed as linear-logic propositions\footnote{
    This choice is usually motivated as it comes with deadlock freedom, but we have two different reasons: (1)~it allows for direct compatibility with session-type systems for deadlock freedom, and (2)~a logical basis gives us a very clean and well-behaved linear session type system, which we can carefully extend to serve our goals (here, guaranteeing noninterference by typing).
}; they are typeset in a \textcolor{RoyalBlue}{\sffamily blue and sans-serif} font.

\begin{definition}[Types]
    \label{d:typeSyntax}
    \emph{Types} $\fType{A},\fType{B},\fType{C},\ldots$ are defined by the following syntax:
    \[
        \fType{A},\fType{B},\fType{C},\ldots ::= \fType{1} \sepr \fType{\bot} \sepr \fType{\oplus \{ i : A \}_{i \in I}} \sepr \fType{\with \{ i : A \}_{i \in I}} \sepr \fType{A \tensor B} \sepr \fType{A \parr B}
    \]

    \emph{Duality} is a unary operation on types, denoted $\fType{\dual{A}}$, defined as follows:
    \begin{align*}
        \fType{\dual{1}} &:= \fType{\bot}
        &
        \fType{\dual{\oplus \{ i : A_i \}_{i \in I}}} &:= \fType{\with \{ i : \dual{A_i} \}_{i \in I}}
        &
        \fType{\dual{A \tensor B}} &:= \fType{\dual{A} \parr \dual{B}}
        \\
        \fType{\dual{\bot}} &:= \fType{1}
        &
        \fType{\dual{\with \{ i : A_i \}_{i \in I}}} &:= \fType{\oplus \{ i : \dual{A_i} \}_{i \in I}}
        &
        \fType{\dual{A \parr B}} &:= \fType{\dual{A} \tensor \dual{B}}
    \end{align*}
\end{definition}

Type~$\fType{1}$ is associated with names that close channels, and~$\fType{\bot}$ with names that wait for channels to close.
Types~$\fType{\oplus \{ i : A_i \}_{i \in I}}$ and $\fType{\with \{ i : A_i \}_{i \in I}}$ are associated with names that make and expect labeled selections, respectively; given $\fProc{j} \in \fProc{I}$, $\fType{A_j}$ is the type of the continuation after $\fProc{j}$ has been selected/received.
Types~$\fType{A \tensor B}$ and $\fType{A \parr B}$ are associated with names that send and receive, respectively; $\fType{A}$ and $\fType{B}$ are the types of the payload and continuation afterwards.

Duality is a key component of session types, as it defines precisely what is meant by complementary behavior; for example, $\fType{\dual{1}} = \fType{\bot}$ is complementary to $\fType{1}$.
Clearly, duality is an involution (i.e., $\fType{\dual{(\dual{A})}} = \fType{A}$).

Our type system is defined as a sequent calculus.
In the following, ignore the annotations in \textcolor{RedOrange}{\textit{red and italic}}; these annotations are for IFC, explained in \Cref{s:IFC:IFC}.

\begin{figure}[t]
    {\small
    \begin{mathpar}
        \begin{bussproof}[typ-inact]
            \bussAx{
                \fType{\fIFC{\Omega} \vdash \fIFC{\fProc{0} \at d} :: \emptyset}
            }
        \end{bussproof}
        \and
        \begin{bussproof}[typ-par]
            \bussAssume{
                \fIFC{\Omega \Vdash d \lleq d'_1 \lcap d'_2}
            }
            \bussAssume{
                \fType{\fIFC{\Omega} \vdash \fIFC{\fProc{P} \at d'_1} :: \Gamma}
            }
            \bussAssume{
                \fType{\fIFC{\Omega} \vdash \fIFC{\fProc{Q} \at d'_2} :: \Delta}
            }
            \bussTern{
                \fType{\fIFC{\Omega} \vdash \fIFC{\fProc{P \| Q} \at d} :: \Gamma , \Delta}
            }
        \end{bussproof}
        \and
        \begin{bussproof}[typ-res]
            \bussAssume{
                \fType{\fIFC{\Omega} \vdash \fIFC{\fProc{P} \at d} :: \Gamma , \fProc{x}:A\fIFC{[c]}, \fProc{y}:\dual{A}\fIFC{[c]}}
            }
            \bussUn{
                \fType{\fIFC{\Omega} \vdash \fIFC{\fProc{\nu{xy} P} \at d} :: \Gamma}
            }
        \end{bussproof}
        \and
        \begin{bussproof}[typ-close]
            \bussAssume{
                \fIFC{\Omega \Vdash d \lleq c}
            }
            \bussUn{
                \fType{\fIFC{\Omega} \vdash \fIFC{\fProc{\pClose x[]} \at d} :: \fProc{x}:1\fIFC{[c]}}
            }
        \end{bussproof}
        \and
        \begin{bussproof}[typ-wait]
            \bussAssume{
                \fIFC{\Omega \Vdash d' = d \lcup c}
            }
            \bussAssume{
                \fType{\fIFC{\Omega} \vdash \fIFC{\fProc{P} \at d'} :: \Gamma}
            }
            \bussBin{
                \fType{\fIFC{\Omega} \vdash \fIFC{\fProc{\pWait x() ; P} \at d} :: \Gamma , \fProc{x}:\bot\fIFC{[c]}}
            }
        \end{bussproof}
        \and
        \begin{bussproof}[typ-sel]
            \bussAssume{
                \fIFC{\Omega \Vdash d \lleq c}
            }
            \bussAssume{
                \fProc{j} \in \fProc{I}
            }
            \bussBin{
                \fType{\fIFC{\Omega} \vdash \fIFC{\fProc{\pSel x[b]<j} \at d} :: \fProc{x}:\oplus \{ i : A_i \}_{i \in I}\fIFC{[c]} , \fProc{b}:\dual{A_j}\fIFC{[c]}}
            }
        \end{bussproof}
        \and
        \begin{bussproof}[typ-bra]
            \bussAssume{
                \fIFC{\Omega \Vdash d' = d \lcup c}
            }
            \bussAssume{
                \forall \fProc{i} \in \fProc{I}.~ \fType{\fIFC{\Omega} \vdash \fIFC{\fProc{P_i} \at d'} :: \Gamma , \fProc{z}:A_i\fIFC{[c]}}
            }
            \bussBin{
                \fType{\fIFC{\Omega} \vdash \fIFC{\fProc{\pBra x(z)>\{i:P_i\}_{i \in I}} \at d} :: \Gamma , \fProc{x}:\& \{ i : A_i \}_{i \in I}\fIFC{[c]}}
            }
        \end{bussproof}
        \and
        \begin{bussproof}[typ-send]
            \bussAssume{
                \fIFC{\Omega \Vdash d \lleq c}
            }
            \bussUn{
                \fType{\fIFC{\Omega} \vdash \fIFC{\fProc{\pSend x[a,b]} \at d} :: \fProc{x}:A \tensor B\fIFC{[c]} , \fProc{a}:\dual{A}\fIFC{[c]} , \fProc{b}:\dual{B}\fIFC{[c]}}
            }
        \end{bussproof}
        \and
        \begin{bussproof}[typ-recv]
            \bussAssume{
                \fIFC{\Omega \Vdash d' = d \lcup c}
            }
            \bussAssume{
                \fType{\fIFC{\Omega} \vdash \fIFC{\fProc{P} \at d'} :: \Gamma , \fProc{y}:A\fIFC{[c]} , \fProc{z}:B\fIFC{[c]}}
            }
            \bussBin{
                \fType{\fIFC{\Omega} \vdash \fIFC{\fProc{\pRecv x(y,z) ; P} \at d} :: \Gamma , \fProc{x}:A \parr B\fIFC{[c]}}
            }
        \end{bussproof}
    \end{mathpar}
    }
    \caption{Typing rules; cf.\ \Cref{d:typeSys}.}\label{f:typeSys}
\end{figure}

\begin{definition}[Type System]
    \label{d:typeSys}
    \emph{Typing contexts} $\fType{\Gamma},\fType{\Delta},\ldots$ are defined by the following syntax:
    \[
        \fType{\Gamma},\fType{\Delta},\ldots ::= \fType{\emptyset} \sepr \fType{\Gamma, \fProc{x}:{A}\fIFC{[c]}}
    \]

    \emph{Typing judgments} are denoted $\fType{\fIFC{\Omega} \vdash \fIFC{\fProc{P} \at d} :: \Gamma}$.
    They are derived using the rules in \Cref{f:typeSys}.
    % When $\fIFC{\Omega}$ is clear from the context, we often omit it from typing judgments.

\end{definition}

Typing contexts are thus sets of types assigned to names; the type system allows implicitly reordering these assignments in typing contexts.
Whenever we write $\fType{\Gamma , \Delta}$, we assume that the sets of names appearing in $\fType{\Gamma}$ and $\fType{\Delta}$ are disjoint.

Rule~\ruleLabel{typ-inact} types inaction under empty context.
Rule~\ruleLabel{typ-par} types parallel composition by splitting the typing context into disjoint parts, one for each parallel process.
Rule~\ruleLabel{typ-res} types restriction by requiring the connected names to be dually typed.
Rule~\ruleLabel{typ-close} types a close with only its subject in the context.
Dually, Rule~\ruleLabel{typ-wait} types a wait by removing its subject from the context of the continuation.
Rule~\ruleLabel{typ-sel} types a selection; note that the continuation is typed dually to the continuation type of the selection itself, as this name will be received by a corresponding branch and used for further communications there.
Dually, Rule~\ruleLabel{typ-bra} types a branch; it requires every continuation to be typed with the same context besides the type of the continuation name.
Rule~\ruleLabel{typ-send} types a send; the payload and continuation are typed dually, similar to the continuation in Rule~\ruleLabel{typ-sel}.
Dually, Rule~\ruleLabel{typ-recv} types a receive.

\begin{example}
    \label{x:typing}
    To illustrate process typing, we type the secure variant of $\fProc{Gov_A^{\fIFC{H}}}$ introduced in \Cref{x:syntax} as follows.
    We omit the \textcolor{RedOrange}{\textit{red and italic}} IFC annotations entirely, as well as the typing of the $\fProc{oc_2}$ branch which is analogous to the $\fProc{oc_1}$ branch.
    \[
        \begin{bussproof}
            \bussAx[\ruleLabel{typ-sel}]{
                \fType{
                    \begin{array}[t]{@{}l@{}}
                        \vdash
                        \fProc{\pSel a_I[a_I^{1'}]<act}
                        \\
                        ::
                        \fProc{a_I} : \oplus \{ act : 1 , wait : 1 \}
                        ,
                        \fProc{a_I^{1'}} : \bot
                    \end{array}
                }
            }
            \bussAx[\ruleLabel{typ-close}]{
                \fType{
                    \vdash
                    \fProc{\pClose a_I^1[]}
                    ::
                    \fProc{a_I^1} : 1
                }
            }
            \bussUn[\ruleLabel{typ-wait}]{
                \fType{
                    \begin{array}[t]{@{}l@{}}
                        \vdash
                        \fProc{\pWait a_{\fIFC{L}}^1() ; \pClose a_I^1[]}
                        \\
                        ::
                        \fProc{a_{\fIFC{L}}^1} : \bot
                        ,
                        \fProc{a_I^1} : 1
                    \end{array}
                }
            }
            \bussBin[\ruleLabel{typ-par}]{
                \fType{
                    \begin{array}[t]{@{}l@{}}
                        \vdash
                        \fProc{\pSel a_I[a_I^{1'}]<act \| \pWait a_{\fIFC{L}}^1() ; \pClose a_I^1[]}
                        \\
                        ::
                        \fProc{a_{\fIFC{L}}^1} : \bot
                        ,
                        \fProc{a_I} : \oplus \{ act : 1 , wait : 1 \}
                        ,
                        \fProc{a_I^{1'}} : \bot
                        ,
                        \fProc{a_I^1} : 1
                    \end{array}
                }
            }
            \bussUn[\ruleLabel{typ-res}]{
                \fType{
                    \begin{array}[t]{@{}l@{}}
                        \vdash
                        \fProc{\nu{a_I^{1'} a_I^1} ( \pSel a_I[a_I^{1'}]<act \| \pWait a_{\fIFC{L}}^1() ; \pClose a_I^1[] )}
                        \\
                        ::
                        \fProc{a_{\fIFC{L}}^1} : \bot
                        ,
                        \fProc{a_I} : \oplus \{ act : 1 , wait : 1 \}
                    \end{array}
                }
            }
            \bussAssume{
                \vdots
            }
            \bussBin[\ruleLabel{typ-bra}]{
                \fType{
                    \vdash
                    \fProc{\pBra a_{\fIFC{L}}(a_{\fIFC{L}}^{1})>\left\{ \begin{array}{@{}l@{}}
                        \fProc{oc_1 : \nu{a_I^{1'} a_I^1} ( \pSel a_I[a_I^{1'}]<act \| \pWait a_{\fIFC{L}}^1() ; \pClose a_I^1[] ) ,} \\
                        \fProc{oc_2 : \nu{a_I^{1'} a_I^1} ( \pSel a_I[a_I^{1'}]<wait \| \pWait a_{\fIFC{L}}^1() ; \pClose a_I^1[] )}
                    \end{array} \right\}}
                    :: \begin{array}[t]{@{}l@{}}
                        \fProc{a_{\fIFC{L}}} : \& \{ oc_1 : \bot , oc_2 : \bot \}
                        , \\
                        \fProc{a_I} : \oplus \{ act : 1 , wait : 1 \}
                    \end{array}
                }
            }
        \end{bussproof}
    \]
\end{example}

\subsection{Information Flow Control}
\label{s:IFC:IFC}

We now enrich the type system presented thus far with IFC, such that well typedness guarantees noninterference (\Cref{s:DSNI}).
That is, we introduce and explain the annotations in \textcolor{RedOrange}{\textit{red and italic}} in \Cref{f:typeSys}, and formalize the intuititions given in \Cref{s:ideas:IFC}.

We use $\fIFC{c},\fIFC{d},\ldots$ to denote secrecy levels.
The relation between secrecy levels is denoted~$\fIFC{c \lleq d}$ ($\fIFC{c}$ is at most as secret as $\fIFC{d}$), forming a lattice~$\fIFC{\Omega}$.
We write $\fIFC{\Omega \Vdash \phi}$ to denote that the relation~$\fIFC{\phi}$ between secrecies holds within $\fIFC{\Omega}$.
%If $\fIFC{\Omega}$ is obvious from the context, we often simply write $\fIFC{\phi}$.
The least upper bound (join) and greatest lower bound (meet) are denoted $\fIFC{c \lcup d}$ and $\fIFC{c \lcap d}$, respectively.

Every name is assigned a secrecy level, denoted in typing contexts using square brackets after the name's type, as in $\fType{\fProc{x}:A\fIFC{[c]}}$.
To remember the level of secrecy of the messages received by a process, we annotate the process in typing judgments with a \emph{running secrecy}, denoted~$\fIFC{\fProc{P} \at d}$.
Given the running secrecy~$\fIFC{d}$ of the process before an input and the secrecy level~$\fIFC{c}$ of the input's subject, the input updates the running secrecy to the join $\fIFC{d \lcup c}$.
When a process then performs an output, to make sure that the name is secured for handling the secrecy level of the outgoing message, our IFC requires that its running secrecy is not higher than that of the output's subject name.

We make these intuitions precise by discussing the IFC annotations on each typing rule in \Cref{f:typeSys}.
Since Rule~\ruleLabel{typ-inact} does not involve communication, no secrecy checks are necessary.
Rule~\ruleLabel{typ-close} types an output, so it requires that the running secrecy $\fIFC{d}$ of the close is at most the secrecy level $\fIFC{c}$ of the closed name ($\fIFC{d \lleq c}$): the information received so far (the running secrecy) is not more secret than the closed name.
Rules~\ruleLabel{typ-sel} and~\ruleLabel{typ-send} also type outputs, so their checks are similar.
Rule~\ruleLabel{typ-wait} types an input, so it sets the running secrecy $\fIFC{d'}$ of the continuation of the wait to the least upper bound of the running secrecy $\fIFC{d}$ before the wait and the secrecy level~$\fIFC{c}$ of the name of the wait ($\fIFC{d' = d \lcup c}$; i.e., to the least secrecy level that is at least as high as both involved secrecies).
Rules~\ruleLabel{typ-bra} and~\ruleLabel{typ-recv} also type inputs, so they update running secrecies similarly.
Rule~\ruleLabel{typ-par} combines the running secrecies $\fIFC{d'_1}$ and $\fIFC{d'_2}$ of the parallel processes by taking a secrecy level~$\fIFC{d}$ that is at most the greatest lower bound of $\fIFC{d'_1}$ and $\fIFC{d'_2}$ ($\fIFC{d \lleq d'_1 \lcap d'_2}$); this way, the parallel composition has a running secrecy of at most the least common secrecy level of its components, ensuring that each process in the composition has at least as much information as the parallel composition itself.
Rule~\ruleLabel{typ-res} requires that the secrecies of the connected names coincide, ensuring that secrecy checks are consistent on both names of the created channel.
Note that channels in typing contexts may have different secrecy levels.
However, typing rules enforce that the channels in the same session (e.g., a send and its continuation) have the same secrecy levels.
For example, Rule~\ruleLabel{typ-sel} ensures that channel $\fProc{x}$ and its continuation $\fProc{b}$ are both of the same secrecy level $\fIFC{c}$.

\begin{example}
    \label{x:ifc}
    To illustrate IFC in our type system, we consider again the typing of $\fProc{Gov_A^{\fIFC{H}}}$ from \Cref{x:syntax,x:typing}.
    We repeat the typing derivation and include IFC annotations, but omit processes and types to save space.
    Let $\fIFC{\Omega}$ be the lattice with the only relation $\fIFC{L \llt H}$.
    \[
        \begin{bussproof}
            \bussAssume{
                \fIFC{L \lleq H}
            }
            \bussUn[\ruleLabel{typ-sel}]{
                \fType{
                    \fIFC{\Omega}
                    \vdash
                    \fIFC{{} \at L}
                    ::
                    \fProc{a_I} : \fIFC{[H]}
                    ,
                    \fProc{a_I^{1'}} : \fIFC{[H]}
                }
            }
            \bussAssume{
                \fIFC{L \lleq H}
            }
            \bussUn[\ruleLabel{typ-close}]{
                \fType{
                    \fIFC{\Omega}
                    \vdash
                    \fIFC{{} \at L \lcup L = L}
                    ::
                    \fProc{a_I^1} : \fIFC{[H]}
                }
            }
            \bussUn[\ruleLabel{typ-wait}]{
                \fType{
                    \fIFC{\Omega}
                    \vdash
                    \fIFC{{} \at L}
                    ::
                    \fProc{a_{\fIFC{L}}^1} : \fIFC{[L]}
                    ,
                    \fProc{a_I^1} : \fIFC{[H]}
                }
            }
            \bussBin[\ruleLabel{typ-par}]{
                \fType{
                    \fIFC{\Omega}
                    \vdash
                    \fIFC{{} \at L}
                    ::
                    \fProc{a_{\fIFC{L}}^1} : \fIFC{[L]}
                    ,
                    \fProc{a_I} : \fIFC{[H]}
                    ,
                    \fProc{a_I^{1'}} : \fIFC{[H]}
                    ,
                    \fProc{a_I^1} : \fIFC{[H]}
                }
            }
            \bussUn[\ruleLabel{typ-res}]{
                \fType{
                    \fIFC{\Omega}
                    \vdash
                    \fIFC{{} \at L \lcup L = L}
                    ::
                    \fProc{a_{\fIFC{L}}^1} : \fIFC{[L]}
                    ,
                    \fProc{a_I} : \fIFC{[H]}
                }
            }
            \bussAssume{
                \vdots
            }
            \bussBin[\ruleLabel{typ-bra}]{
                \fType{
                    \fIFC{\Omega}
                    \vdash
                    \fIFC{{} \at L}
                    ::
                    \fProc{a_{\fIFC{L}}} : \fIFC{[L]}
                    ,
                    \fProc{a_I} : \fIFC{[H]}
                }
            }
        \end{bussproof}
    \]
    Hence, $\fProc{Gov_A^{\fIFC{H}}}$ is considered secure in our type system, and as we will show in \Cref{s:DSNI} this means that noninterference holds for this process.

    However, the initial assignment of maximum secrecies to endpoints is chosen by the user.
    Well typedness and thus noninterference depends on this initial choice.
    To illustrate, reconsider the derivation above but now swapping the initial maximum secrecies:
    \[
        \begin{bussproof}
            \bussAssume{
                \fIFC{H \not\lleq L}
            }
            \bussUn[\ruleLabel{typ-sel}]{
                \fType{
                    \fIFC{\Omega}
                    \vdash
                    \fIFC{{} \at H}
                    ::
                    \fProc{a_I} : \fIFC{[L]}
                    ,
                    \fProc{a_I^{1'}} : \fIFC{[L]}
                }
            }
            \bussAssume{
                \fIFC{H \not\lleq L}
            }
            \bussUn[\ruleLabel{typ-close}]{
                \fType{
                    \fIFC{\Omega}
                    \vdash
                    \fIFC{{} \at H \lcup H = H}
                    ::
                    \fProc{a_I^1} : \fIFC{[L]}
                }
            }
            \bussUn[\ruleLabel{typ-wait}]{
                \fType{
                    \fIFC{\Omega}
                    \vdash
                    \fIFC{{} \at H}
                    ::
                    \fProc{a_{\fIFC{L}}^1} : \fIFC{[H]}
                    ,
                    \fProc{a_I^1} : \fIFC{[L]}
                }
            }
            \bussBin[\ruleLabel{typ-par}]{
                \fType{
                    \fIFC{\Omega}
                    \vdash
                    \fIFC{{} \at H}
                    ::
                    \fProc{a_{\fIFC{L}}^1} : \fIFC{[H]}
                    ,
                    \fProc{a_I} : \fIFC{[L]}
                    ,
                    \fProc{a_I^{1'}} : \fIFC{[L]}
                    ,
                    \fProc{a_I^1} : \fIFC{[L]}
                }
            }
            \bussUn[\ruleLabel{typ-res}]{
                \fType{
                    \fIFC{\Omega}
                    \vdash
                    \fIFC{{} \at L \lcup H = H}
                    ::
                    \fProc{a_{\fIFC{L}}^1} : \fIFC{[H]}
                    ,
                    \fProc{a_I} : \fIFC{[L]}
                }
            }
            \bussAssume{
                \vdots
            }
            \bussBin[\ruleLabel{typ-bra}]{
                \fType{
                    \fIFC{\Omega}
                    \vdash
                    \fIFC{{} \at L}
                    ::
                    \fProc{a_{\fIFC{L}}} : \fIFC{[H]}
                    ,
                    \fProc{a_I} : \fIFC{[L]}
                }
            }
        \end{bussproof}
    \]
    This time, $\fProc{Gov_A^{\fIFC{H}}}$ is not well typed: the IFC requirements of Rules~\ruleLabel{typ-sel} and~\ruleLabel{typ-close} do not hold.
\end{example}

\subsection{Type Preservation}
\label{s:IFC:typePres}

Our type system guarantees by well typedness the usual correctness properties: \emph{session fidelity} and \emph{communication safety}.
The former states that a process correctly implements the session types assigned to its names, and the latter that no communication mismatches take place (such as simultaneous outputs on both names of a channel).

Both these properties follow directly from \emph{type preservation}: well typedness is preserved across structural congruences (subject congruence; \Cref{t:subjCong}) and reduction (subject reduction; \Cref{t:subjRed}).
These results rely on two lemmas:
\begin{itemize}

    \item
        \Cref{l:notFnImpliesNotInDom} states that names that are not free in a process are not assigned in the typing of the process.

    \item
        \Cref{l:substitution} states that substitution in a process is reflected in its typing.

\end{itemize}

\begin{lemma}
    \label{l:notFnImpliesNotInDom}
    Given $\fType{\fIFC{\Omega} \vdash \fIFC{\fProc{P} \at d} :: \Gamma}$, if $\fProc{x} \notin \fn(\fProc{P})$, then $\fProc{x} \notin \dom(\fType{\Gamma})$.
\end{lemma}

\begin{restatable}[Substitution]{lemma}{lSubstitution}
    \label{l:substitution}
    Given $\fType{\fIFC{\Omega} \vdash \fIFC{\fProc{P} \at d} :: \Gamma , \fProc{x}:A\fIFC{[c]}}$, we have
    \[
        \fType{\fIFC{\Omega} \vdash \fIFC{\fProc{P \{ y/x \}} \at d} :: \Gamma , \fProc{y}:A\fIFC{[c]}}.
    \]
\end{restatable}

\begin{restatable}[Subject Congruence]{theorem}{tSubjCong}
    \label{t:subjCong}
    If $\fType{\fIFC{\Omega} \vdash \fIFC{\fProc{P} \at d} :: \Gamma}$ and $\fProc{P} \sc \fProc{Q}$ for some $\fProc{Q}$, then $\fType{\fIFC{\Omega} \vdash \fIFC{\fProc{Q} \at d} :: \Gamma}$.
\end{restatable}

% plscheck
\begin{proof}
    By induction on the derivation of $\fProc{P} \sc \fProc{Q}$.
    The inductive cases correspond to closure under arbitrary process contexts in \Cref{d:procSyntax}; these cases follow from the IH straightforwardly.
    The base cases correspond to the seven rules in \Cref{f:procSCRedd} (top).
    In each case, we apply inversion on the typing of $\fProc{P}$ to derive the typing of $\fProc{Q}$, and vice versa, with straightforward reasoning about running secrecies.
    The only interesting case is Rule~\ruleLabel{sc-par-assoc} ($\fProc{( P \| Q ) \| R} \sc \fProc{P \| ( Q \| R )}$), where we derive the running secrecy of $\fProc{Q \| R}$ from that of $\fProc{P \| Q}$, and vice versa.
    The full proof is in \appref[proof:t:subjCong].
\end{proof}

The following theorem states that (i)~reduction preserves the well typedness of processes, and (ii)~the running secrecy of processes may either stay the same or increase during reduction.
This implies that a process never forgets the secrets it has learned, but it may learn more secrets as it reduces.

\begin{restatable}[Subject Reduction]{theorem}{tSubjRed}
    \label{t:subjRed}
    If $\fType{\fIFC{\Omega} \vdash \fIFC{\fProc{P} \at d} :: \Gamma}$ and $\fProc{P} \redd \fProc{Q}$ for some $\fProc{Q}$, then $\fType{\fIFC{\Omega} \vdash \fIFC{\fProc{Q} \at d'} :: \Gamma}$ for some $\fIFC{d'}$ such that $\fIFC{\Omega \Vdash d \lleq d'}$.
\end{restatable}

\begin{proof}
    By induction on the derivation of $\fProc{P} \redd \fProc{Q}$.
    The cases correspond to the reduction rules in \Cref{f:procSCRedd} (bottom).
    In each case, we apply inversion on the typing of $\fProc{P}$ to derive the typing of $\fProc{Q}$.
    The full proof is in \appref[proof:t:subjRed]; here, we show the interesting case of Rule~\ruleLabel{red-send-recv}: $\fProc{\nu{xy} ( \pSend x[a,b] \| \pRecv y(z,w) ; P )} \redd \fProc{P \pSubst{a/z,b/w}}$.
    Given
    \begin{align}
        & \fIFC{\Omega \Vdash d \lleq d'_1 \lcap d'_2},
        \label{eq:subjRed:sendRecv1}
        \\
        & \begin{bussproof}
            \bussAssume{
                \fIFC{\Omega \Vdash d'_1 \lleq c}
            }
            \bussUn[\ruleLabel{typ-send}]{
                \fType{\fIFC{\Omega} \vdash \fIFC{\fProc{\pSend x[a,b]} \at d'_1} :: \fProc{x}:A \tensor B\fIFC{[c]} , \fProc{a}:\dual{A}\fIFC{[c]} , \fProc{b}:\dual{B}\fIFC{[c]}}
            }
        \end{bussproof},
        \label{eq:subjRed:sendRecv2}
        \\
        & \fIFC{\Omega \Vdash d''_2 = d'_2 \lcup c},
        \label{eq:subjRed:sendRecv3}
    \end{align}
    we have
    \begin{mathpar}
        \begin{bussproof}
            \bussAssume{
                \eqref{eq:subjRed:sendRecv1}
            }
            \bussAssume{
                \eqref{eq:subjRed:sendRecv2}
            }
            \bussAssume{
                \eqref{eq:subjRed:sendRecv3}
            }
            \bussAssume{
                \fType{\fIFC{\Omega} \vdash \fIFC{\fProc{P} \at d''_2} :: \Gamma , \fProc{z}:\dual{A}\fIFC{[c]} , \fProc{w}:\dual{B}\fIFC{[c]}}
            }
            \bussBin[\ruleLabel{typ-recv}]{
                \fType{\fIFC{\Omega} \vdash \fIFC{\fProc{\pRecv y(z,w) ; P} \at d'_2} :: \Gamma , \fProc{y}:\dual{A} \parr \dual{B}\fIFC{[c]}}
            }
            \bussTern[\ruleLabel{typ-par}]{
                \fType{\fIFC{\Omega} \vdash \fIFC{\fProc{\pSend x[a,b] \| \pRecv y(z,w) ; P} \at d} :: \Gamma , \fProc{x}:A \tensor B\fIFC{[c]} , \fProc{y}:\dual{A} \parr \dual{B}\fIFC{[c]} , \fProc{a}:\dual{A}\fIFC{[c]} , \fProc{b}:\dual{B}\fIFC{[c]}}
            }
            \bussUn[\ruleLabel{typ-res}]{
                \fType{\fIFC{\Omega} \vdash \fIFC{\fProc{\nu{xy} ( \pSend x[a,b] \| \pRecv y(z,w) ; P )} \at d} :: \Gamma , \fProc{a}:\dual{A}\fIFC{[c]} , \fProc{b}:\dual{B}\fIFC{[c]}}
            }
        \end{bussproof}
        \and
        \Rightarrow
        \and
        \begin{bussproof}
            \bussAssume{
                \text{\Cref{l:substitution} twice}
            }
            \bussUn{
                \fType{\fIFC{\Omega} \vdash \fIFC{\fProc{P \pSubst{a/z,b/w}} \at d''_2} :: \Gamma , \fProc{a}:\dual{A}\fIFC{[c]} , \fProc{b}:\dual{B}\fIFC{[c]}}
            }
        \end{bussproof}
    \end{mathpar}
    By assumption and by definition, $\fIFC{\Omega \Vdash d''_2 \lgeq d'_2}$.
    Also, by definition, $\fIFC{\Omega \Vdash d'_2 \lgeq d'_1 \lcap d'_2}$, so, by assumption, $\fIFC{\Omega \Vdash d'_2 \lgeq d}$.
    Hence, $\fIFC{\Omega \Vdash d''_2 \lgeq d}$.
\end{proof}

\subparagraph*{Liveness / Progress.}

Liveness / Progress properties specify the conditions under which processes can reduce.
The progress property of APCP states that reduction takes place for a syntactic notion of ``live'' processes~\cite{journal/scico/vdHeuvelP22}.
Since this result does not rely on APCP's priority mechanisms, it applies to our process language as well.

% !TeX root = ../main.tex
\section{Logical Relation}
\label{s:logRel}

This section defines an equivalence on typed processes up to ``observable messages'' (\Cref{d:DSNIRel}) that we will use to state and prove DSNI in \Cref{s:DSNI}.
We first give some preliminary definitions in \Cref{s:logRel:prelim}, before defining the logical relation that induces this equivalence in \Cref{s:logRel:theRel}.

\subsection{Preliminary Definitions}
\label{s:logRel:prelim}

As anticipated in \Cref{s:ideas:logRel}, we are interested in the behavior of a process when it runs in different contexts.
That is, we want to connect all the free names of the process in arbitrary ways.
To this end, we define evaluation contexts: processes with a hole inside which a process may reduce (so under parallel composition and restriction).
Evaluation contexts are typeset using an \textcolor{YellowOrange}{\ttfamily orange and monospaced} font.

\begin{definition}[Evaluation Context]
    \label{d:evalCtx}
    \emph{Evaluation contexts} ($\fCtx{E}$) are defined as follows:
    \[
        \fCtx{E} ::= \fCtx{\hole} \sepr \fCtx{E \| P} \sepr \fCtx{\nu{xy} E}
    \]
    We write $\fCtx{E[\fProc{P}]}$ to denote the process obtained by replacing the hole $\fCtx{\hole}$ in $\fCtx{E}$ by $\fProc{P}$.

    Any definitions on processes before and after this definition are lifted to evaluation contexts, without assigning any meaning to the hole.
    The exception is that alpha renaming does not apply to names that are bound by restriction but not free inside the scope of the restriction.
\end{definition}

\begin{example}
    \label{x:ctx}
    The following is an evaluation context:
    \begin{align*}
        \fCtx{E} := \fCtx{\nu{uw} \nu{xy} \nu{zv} ( \pWait x() ; \pClose u[] \| \pClose z[] \| \hole )}
        \label{eq:ctx}
    \end{align*}
    Both $\fProc{u}$ and $\fProc{w}$ are bound in $\fCtx{E}$.
    Since $\fProc{u}$ appears free within the scope of the restriction as the subject of a close, alpha renaming applies: $\fCtx{E} \alpheq \fCtx{\nu{aw} \nu{xy} \nu{zv} ( \pWait x() ; \pClose a[] \| \pClose z[] \| \hole )}$.
    However, the same does not hold for $\fProc{w}$: $\fCtx{E} \not\alpheq \fCtx{\nu{ua} \nu{xy} \nu{zv} ( \pWait x() ; \pClose u[] \| \pClose z[] \| \hole )}$.
\end{example}

We refer to the names that connect the process and its context as the \emph{interface}.
Our logical relation focuses on messages between context and process, i.e., messages that must pass through the interface.
The following definition identifies outputs in process and context that are not blocked by prefixes. In particular, $\aon(\fProc{P})$ is the set of names \emph{along} which $\fProc{P}$ is ready to output, and $\acon(\fCtx{E})$ is the set of names \emph{to} which the context is ready to output.

\begin{definition}[Active Interface Names]
    \label{d:ain}
    We define the set of \emph{active output names} of $\fProc{P}$, denoted $\aon(\fProc{P})$, as the subjects of non-blocked outputs in $\fProc{P}$:
    \begin{align*}
        \aon(\fProc{0}) &:= \emptyset
        &
        \aon(\fProc{P \| Q}) &:= \aon(\fProc{P}) \cup \aon(\fProc{Q})
        &
        \aon(\fProc{\nu{xy} P}) &:= \aon(\fProc{P}) \setminus \{\fProc{x},\fProc{y}\}
        \\
        \aon(\fProc{\pClose x[]}) &:= \{\fProc{x}\}
        &
        \aon(\fProc{\pSend x[a,b]}) &:= \{\fProc{x}\}
        &
        \aon(\fProc{\pSel x[b]<j}) &:= \{\fProc{x}\}
        \\
        \aon(\fProc{\pWait x() ; P}) &:= \emptyset
        &
        \aon(\fProc{\pRecv x(y,z) ; P}) &:= \emptyset
        &
        \mathllap{
            \aon(\fProc{\pBra x(z)>\{i:P_i\}_{i \in I}})
        } &:= \emptyset
    \end{align*}

    We define the set of \emph{active context output names} of $\fCtx{E}$, denoted $\acon(\fCtx{E})$, as the names in the interface of $\fCtx{E}$ that are connected to active output names of $\fCtx{E}$ through restriction:
    \[
        \acon(\fCtx{E}) :=
        \{ \fProc{x} \mid \exists \fProc{y},\fCtx{E'} .~ \big( \fCtx{E} \sc \fCtx{\nu{xy} E'} \wedge \fProc{x} \notin \fn(\fCtx{E'}) \wedge \fProc{y} \in \aon(\fCtx{E'}) \big) \}
    \]

    We define the set of \emph{active interface names} of $\fCtx{E}$ and $\fProc{P}$ as the union of the active context output names of $\fCtx{E}$ and the active output names of $\fProc{P}$:
    \[
        \ain(\fCtx{E} , \fProc{P}) := \acon(\fCtx{E}) \cup \aon(\fProc{P})
    \]
\end{definition}

\begin{example}
    We illustrate the active interface names between $\fProc{P} := \fProc{\pClose y[] \| \pWait w() ; \pWait v() ; 0}$ and $\fCtx{E}$ from \Cref{x:ctx}.
    It is easy to see that $\aon(\fProc{P}) = \{\fProc{y}\}$.
    To determine $\acon(\fCtx{E})$ we search for names in $\fCtx{E}$ that are bound by restriction to names used for output, but not used themselves.
    That is, we look for names in the interface between the context and the containing process, on which the containing process can expect to receive an output from the context.
    For example, in $\fCtx{E}$, name $\fProc{v}$ is bound to $\fProc{z}$ which is used for an output, while $\fProc{v}$ itself is not used (it appears in the interface).
    Since there are no further such names, we have $\acon(\fCtx{E}) = \{\fProc{v}\}$.
    As such, $\ain(\fCtx{E},\fProc{P}) = \{\fProc{y},\fProc{v}\}$.
\end{example}

The interface is where an attacker (cf.\ \Cref{s:ideas:logRel}) may observe the behavior of our process.
In this, we assume that the attacker can only observe messages up to a certain secrecy level~$\fIFC{\xi}$.
As such, our relation is only interested in the behavior of the process on \emph{observable} channels.
To this end, we define a projection on typing contexts to filter out unobservable channels.
Also, we define when a process in context is well typed with respect to a given typing context of observable channels.

\begin{definition}[Projection and Networks]
    \label{d:projNetw}
    Given a secrecy lattice $\fIFC{\Omega}$, a secrecy level $\fIFC{\xi} \in \dom(\fIFC{\Omega})$, and a typing context $\fType{\Gamma}$, we define the \emph{projection} $\fType{\Gamma \proj_{\fIFC{\Omega}} \fIFC{\xi}}$ as follows:
    \begin{align*}
        \fType{(\Gamma , \fProc{x}:A\fIFC{[c]}) \proj_{\fIFC{\Omega}} \fIFC{\xi}} &:= \begin{cases}
            \fType{(\Gamma \proj_{\fIFC{\Omega}} \fIFC{\xi}) , \fProc{x}:A\fIFC{[c]}} & \text{if $\fIFC{\Omega \Vdash c \lleq \xi}$}
            \\
            \fType{\Gamma \proj_{\fIFC{\Omega}} \fIFC{\xi}} & \text{if $\fIFC{\Omega \Vdash c \not\lleq \xi}$}
        \end{cases}
        &
        \fType{\emptyset \proj_{\fIFC{\Omega}} \fIFC{\xi}} &:= \fType{\emptyset}
    \end{align*}
    We often omit $\fIFC{\Omega}$ when it is clear from the context.

    We say $\fCtx{E}$ and $\fProc{P}$ \emph{form a network with interface $\fType{\Gamma}$ observable up to $\fIFC{\xi}$ under $\fIFC{\Omega}$}, denoted $(\fCtx{E},\fProc{P}) \in \netw{\fIFC{\Omega};\fIFC{\xi}}(\fType{\Gamma})$ if and only if there are $\fIFC{d},\fIFC{d'},\fType{\Gamma'}$ such that $\fType{\Gamma} = \fType{\Gamma' \proj_{\fIFC{\Omega}} \fIFC{\xi}}$, $\fType{\fIFC{\Omega} \vdash \fIFC{\fProc{P} \at d'} :: \Gamma'}$, and $\fType{\fIFC{\Omega} \vdash \fIFC{\fCtx{E[\fProc{P}]} \at d} :: \emptyset}$.
    By abuse of notation, we write $(\fCtx{E_1},\fProc{P_1} ; \fCtx{E_2},\fProc{P_2}) \in \netw{\fIFC{\Omega};\fIFC{\xi}}(\fType{\Gamma})$ to denote $(\fCtx{E_1},\fProc{P_1}) \in \netw{\fIFC{\Omega};\fIFC{\xi}}(\fType{\Gamma})$ and $(\fCtx{E_2},\fProc{P_2}) \in \netw{\fIFC{\Omega};\fIFC{\xi}}(\fType{\Gamma})$.
\end{definition}

\begin{example}
    We anticipate illustrating noninterference on the secure running example introduced in \Cref{x:syntax} on a $\fIFC{L}$ow secrecy channel, by considering the projection of its typing context and an evaluation context to form a network.
    Recall the typing and IFC annotations from \Cref{x:typing,x:ifc}:
    \[
        \fType{\vdash \fIFC{\fProc{Gov_A^{\fIFC{H}}} \at L} :: \Gamma' = \fProc{a_{\fIFC{L}}} : \& \{ oc_1 : \bot , oc_2 : \bot \}\fIFC{[L]} , \fProc{a_I} : \oplus \{ act : 1 , wait : 1 \}\fIFC{[H]}}.
    \]
    We have $\fType{\Gamma} := \fType{\Gamma' \proj \fIFC{L}} = \fType{\fProc{a_{\fIFC{L}}} : \& \{ oc_1 : \bot , oc_2 : \bot \}\fIFC{[L]}}$.
    Let $\fCtx{E} := \fCtx{\nu{a_{\fIFC{H}} a_{\fIFC{L}}} \nu{a_I i_A} ( Gov_A^{\fIFC{L}} \| \hole \| Int_A )}$.
    It is straightforward to confirm that
    % \[
    $
        \fType{\vdash \fIFC{\fCtx{E[\fProc{Gov_A^{\fIFC{H}}}]} \at L} :: \emptyset}
    $.
    % \]
    Hence, $(\fCtx{E},\fProc{Gov_A^{\fIFC{H}}}) \in \netw{\fIFC{L}}(\fType{\Gamma})$.
\end{example}

In our relation, we want to exhaust reductions on unobservable channels, after which we scrutinize behavior on observable names in the interface.
To this end, we define \emph{unobservable reductions}, which entail reductions internal to the process or the context, but also communications between process and context on unobservable interface channels.

\begin{definition}[Unobservable Reduction]
    \label{d:uRedd}
    We define \emph{unobservable reduction} as
    \[
        \fCtx{E},\fProc{P} \uRedd{\fIFC{\Omega};\fIFC{\xi};\fType{\Gamma}} \fCtx{E'},\fProc{P'}
    \]
    if and only if $(\fCtx{E},\fProc{P}) \in \netw{\fIFC{\Omega};\fIFC{\xi}}(\fType{\Gamma})$, $\fCtx{E[\fProc{P}]} \redd \fCtx{E'[\fProc{P'}]}$ and $(\fCtx{E'},\fProc{P'}) \in \netw{\fIFC{\Omega};\fIFC{\xi}}(\fType{\Gamma})$.
    We write $\uReddQ{\fIFC{\Omega};\fIFC{\xi};\fType{\Gamma}}$ (resp.\ $\uRedd*{\fIFC{\Omega};\fIFC{\xi};\fType{\Gamma}}$) for the reflexive (resp.\ reflexive transitive) closure of $\uRedd{\fIFC{\Omega};\fIFC{\xi};\fType{\Gamma}}$, and $\fCtx{E},\fProc{P} \nuRedd{\fIFC{\Omega};\fIFC{\xi};\fType{\Gamma}}$ to denote that there are no $\fCtx{E'},\fProc{P'}$ such that $\fCtx{E},\fProc{P} \uRedd{\fIFC{\Omega};\fIFC{\xi};\fType{\Gamma}} \fCtx{E'},\fProc{P'}$.
\end{definition}

Our relation often requires ``zooming in'' on specific parts of processes.
To this end, we define notions to deal with atomic parts of processes.

\begin{definition}[Nodes and Normal Forms]
    \label{d:nodesNf}
    Given a process $\fProc{P}$, we say $\fProc{P}$ is a \emph{node} if $\fProc{P} \not\sc \fProc{Q \| R}$, $\fProc{P} \not\sc \fProc{\nu{xy} Q}$, and $\fProc{P} \not\sc \fProc{0}$.

    We say a process $\fProc{Q} = \fProc{\nu{x_i y_i}_{i \in I} \prod_{j \in J} P_j}$ is in \emph{normal form} if, for every $\fProc{j} \in \fProc{J}$, $\fProc{P_j}$ is a node,
    and $\fProc{Q}$ is a \emph{normal form of $\fProc{P}$} if $\fProc{P} \sc \fProc{Q}$.
    Normal forms are closed under structural congruence: every process induces an equivalence class of structurally congruent normal forms.

    Given a process in normal form $\fProc{Q} = \fProc{\nu{x_i y_i}_{i \in I} \prod_{j \in J} P_j}$, we define $\nodes(\fProc{Q}) := \{ \fProc{P_j} \mid \fProc{j} \in \fProc{J} \}$ and $\binders(\fProc{Q}) := \big\{ \{\fProc{x_i},\fProc{y_i}\} \mid \fProc{i} \in \fProc{I} \big\}$.
    By abuse of notation, given a process $\fProc{P}$ not necessarily in normal form, we write $\nodes(\fProc{P})$ to denote $\nodes(\fProc{Q})$ for an arbitrary normal form $\fProc{Q}$ of $\fProc{P}$.
\end{definition}

\noindent
For example, let
\begin{align*}
    \hspace{-1ex}
    \fProc{P} &:= \fProc{\nu{uw} \big( \nu{xy} ( \pWait z() ; \pClose x[] \| \pWait y() ; \pClose u[] ) \| \pWait w() ; 0 \big) \| 0}
    &
    \fProc{Q} &:= \fProc{\nu{uw} \nu{xy} ( \pWait z() ; \pClose x[] \| \pWait y() ; \pClose u[] \| \pWait w() ; 0 )}.
\end{align*}
Then $\fProc{Q}$ is a normal form of $\fProc{P}$, with $\nodes(\fProc{Q}) = \{\fProc{\pWait z() ; \pClose x[]} , \fProc{\pWait y() ; \pClose u[]} , \fProc{\pWait w() ; 0}\}$ and $\binders(\fProc{Q}) = \{ \{\fProc{u},\fProc{w}\} , \{\fProc{x},\fProc{y}\} \}$.

\subsection{The Relation}
\label{s:logRel:theRel}

Having presented all its ingredients, we now introduce our logical relation.
As usual, the relation consists of two parts: a \emph{term interpretation} and a \emph{value interpretation}, defined by mutual multiset induction on the interfaces of processes.
The term interpretation is the main part of the relation, and is responsible for calling on the value interpretation when a message is ready to be communicated across the observable interface, as well as ensuring deadlock sensitivity of our noninterference result.
The value interpretation zooms in on the interface, and ensures that the two runs of the process behave identically on observable messages that are to be communicated across the interface.

\begin{figure}[t]
  {\small  \begin{align}
        &
        (\fCtx{E_1[\fProc{P_1}]} ; \fCtx{E_2[\fProc{P_2}]}) \in \termrel{\fIFC{\Omega}}{\fIFC{\xi}}{\fType{\Gamma}}
        \iff
        \label{eq:logRel:term1}
        \\ & \wedge
        (\fCtx{E_1},\fProc{P_1} ; \fCtx{E_2},\fProc{P_2}) \in \netw{\fIFC{\Omega};\fIFC{\xi}}(\fType{\Gamma})
        \label{eq:logRel:term2}
        \\ & \wedge
        \forall \fCtx{E'_1},\fProc{P'_1}.~
        \fCtx{E_1},\fProc{P_1} \uRedd*{\fIFC{\Omega};\fIFC{\xi};\fType{\Gamma}} \fCtx{E'_1},\fProc{P'_1} \nuRedd{\fIFC{\Omega};\fIFC{\xi};\fType{\Gamma}}
        \label{eq:logRel:term3}
        \\ & \hphantom{{}\wedge{}} \implies
        \exists \fCtx{E'_2},\fProc{P'_2}.~
        \begin{array}[t]{@{}l@{}}
            \fCtx{E_2},\fProc{P_2} \uRedd*{\fIFC{\Omega};\fIFC{\xi};\fType{\Gamma}} \fCtx{E'_2},\fProc{P'_2} \nuRedd{\fIFC{\Omega};\fIFC{\xi};\fType{\Gamma}}
            \\ \wedge
            \forall \fProc{x} \in \big( \ain(\fCtx{E'_1},\fProc{P'_1}) \cup \ain(\fCtx{E'_2},\fProc{P'_2}) \big) \cap \dom(\fType{\Gamma}) .~ (\fCtx{E'_1},\fProc{P'_1};\fCtx{E'_2},\fProc{P'_2}) \in \mathrlap{
                \valrel{\fIFC{\Omega}}{\fIFC{\xi}}{\fType{\Gamma}}{\fProc{x}}
            }
            \\ \wedge
            \aon(\fProc{P'_1}) \cap \dom(\fType{\Gamma}) = \aon(\fProc{P'_2}) \cap \dom(\fType{\Gamma})
        \end{array}
        \label{eq:logRel:term4}
    \end{align}}

    \caption{Term relation.}\label{f:termRel}
\end{figure}

Let us start by presenting our term interpretation, denoted $\termrel{\fIFC{\Omega}}{\fIFC{\xi}}{\fType{\Gamma}}$, in \Cref{f:termRel}.
It relates pairs of processes, given a secrecy lattice $\fIFC{\Omega}$, a secrecy level $\fIFC{\xi} \in \dom(\fIFC{\Omega})$, and an interface~$\fType{\Gamma}$.
We break down its definition part by part.
Part~\eqref{eq:logRel:term1} implicitly requires each process to be separable into a context $\fCtx{E_i}$ and a process $\fProc{P_i}$.
Part~\eqref{eq:logRel:term2} then requires the interface between $\fCtx{E_i}$ and $\fProc{P_i}$ to correspond to $\fType{\Gamma}$ up to observability $\fIFC{\xi}$ (cf.\ \Cref{d:projNetw}).
Part~\eqref{eq:logRel:term3} exhausts unobservable reductions for $\fCtx{E_1},\fProc{P_1}$ (cf.\ \Cref{d:uRedd}) in every way possible, resulting in $\fCtx{E'_1},\fProc{P'_1}$.
Part~\eqref{eq:logRel:term4} first requires $\fCtx{E_2},\fProc{P_2}$ to ``catch up'' through exhaustive unobservable reductions, resulting in $\fCtx{E'_2},\fProc{P'_2}$.
Then, Part~\eqref{eq:logRel:term4} invokes the value interpretation, presented next, to scrutinize any messages that are ready to be transferred across the observable part of the interface of either $\fCtx{E'_i},\fProc{P'_i}$ (cf.\ \Cref{d:ain}).
Finally, Part~\eqref{eq:logRel:term4} ensures deadlock sensitivity by requiring the observable messages of $\fProc{P'_1}$ and $\fProc{P'_2}$ to coincide.
In particular, if $\fProc{P'_1}$ does not produce any observable messages along a name in the interface due to a deadlock imposed by a secret, Part~\eqref{eq:logRel:term4} guarantees that $\fProc{P'_2}$ does not produce any observable messages on that name either.

\begin{figure}[t]
   {\small \begin{align*}
        \fType{1}\quad
        &
        (\fCtx{E_1}, \fProc{P_1} ; \fCtx{E_2}, \fProc{P_2}) \in \valrel{\fIFC{\Omega}}{\fIFC{\xi}}{\fType{\Gamma, \fProc{x}:1\fIFC{[c]}}}{\fProc{x}}
        \iff \big(
            (\fCtx{E_1}, \fProc{P_1} ; \fCtx{E_2}, \fProc{P_2}) \in \netw{\fIFC{\Omega};\fIFC{\xi}}(\fType{\Gamma, \fProc{x}:1\fIFC{[c]}})
            \\ & \wedge
            \fProc{P_1} \sc \fProc{\pClose x[] \| P'_1}
            \wedge
            \fProc{P_2} \sc \fProc{\pClose x[] \| P'_2}
            \wedge
            (\fCtx{E_1 \big[ \pClose x[] \| \fProc{P'_1} \big]} ; \fCtx{E_2 \big[ \pClose x[] \| \fProc{P'_2} \big]}) \in \termrel{\fIFC{\Omega}}{\fIFC{\xi}}{\fType{\Gamma}}
        \big)
        \\[2pt] \hline \displaybreak[1] \\[-12pt]
        \fType{\oplus}\quad
        &
        (\fCtx{E_1},\fProc{P_1} ; \fCtx{E_2},\fProc{P_2}) \in \valrel{\fIFC{\Omega}}{\fIFC{\xi}}{\fType{\Gamma,\fProc{x}:\oplus \{ i : A_i \}_{i \in I}\fIFC{[c]}}}{\fProc{x}}
        \iff
        \numberthis\label{eq:logRel:val:oplus1}
        \\ & \hphantom{{}\wedge{}}
        (\fCtx{E_1}, \fProc{P_1} ; \fCtx{E_2}, \fProc{P_2}) \in \netw{\fIFC{\Omega};\fIFC{\xi}}(\fType{\Gamma, \fProc{x}:\oplus \{ i : A_i \}_{i \in I}\fIFC{[c]}})
        \numberthis\label{eq:logRel:val:oplus2}
        \\ & \wedge
        \exists \fProc{j} \in \fProc{I}.~ \fProc{\pSel x[b_1]<j} \in \nodes(\fProc{P_1}) \wedge \fProc{\pSel x[b_2]<j} \in \nodes(\fProc{P_2})
        \numberthis\label{eq:logRel:val:oplus3}
        \\ & \wedge
        \begin{array}[t]{@{}l@{}l@{}}
            \fProc{b_1} \in \dom(\fType{\Gamma}) \implies {}
            & \fProc{b_1} = \fProc{b_2}
            \wedge \fProc{P_1} \sc \fProc{\pSel x[b_1]<j \| P'_1}
            \wedge \fProc{P_2} \sc \fProc{\pSel x[b_2]<j \| P'_2}
            \\
            & {}\wedge{}
            (\fCtx{E_1 \big[ \pSel x[b_1]<j \| \fProc{P'_1} \big]} ; \fCtx{E_2 \big[ \pSel x[b_2]<j \| \fProc{P'_2} \big]}) \in \termrel{\fIFC{\Omega}}{\fIFC{\xi}}{\fType{\Gamma \setminus \fProc{b_1}}}
        \end{array}
        \numberthis\label{eq:logRel:val:oplus4}
        \\ & \wedge
        \fProc{b_1} \notin \dom(\fType{\Gamma}) \implies \big(
            \fProc{b_2} \notin \dom(\fType{\Gamma})
            \numberthis\label{eq:logRel:val:oplus5}
            \\ & \hphantom{{}\wedge{}} \begin{array}[t]{@{}l@{}}
                {}\wedge
                \fProc{P_1} \sc \fProc{\nu{b_1b'} ( \pSel x[b_1]<j \| P'_1 )}
                \wedge \fProc{P_2} \sc \fProc{\nu{b_2b'} ( \pSel x[b_2]<j \| P'_2 )}
                \\
                {}\wedge
                (\fCtx{E_1 \big[ \nu{b_1b'} ( \pSel x[b_1]<j \| \fProc{P'_1} ) \big]} ; \fCtx{E_2 \big[ \nu{b_2b'} ( \pSel x[b_2]<j \| \fProc{P'_2} ) \big]}) \in \termrel{\fIFC{\Omega}}{\fIFC{\xi}}{\fType{\Gamma, \fProc{b'}:A_j\fIFC{[c]}}}
            \big)
        \end{array}
        \\[2pt] \hline \displaybreak[1] \\[-12pt]
        \fType{\tensor}\quad
        &
        (\fCtx{E_1}, \fProc{P_1} ; \fCtx{E_2}, \fProc{P_2}) \in \valrel{\fIFC{\Omega}}{\fIFC{\xi}}{\fType{\Gamma, \fProc{x}:A \tensor B\fIFC{[c]}}}{\fProc{x}}
        \iff
        \\ & \hphantom{{}\wedge{}}
        (\fCtx{E_1}, \fProc{P_1} ; \fCtx{E_2}, \fProc{P_2}) \in \netw{\fIFC{\Omega};\fIFC{\xi}}(\fType{\Gamma, \fProc{x}:A \tensor B\fIFC{[c]}})
        \\ & \wedge
        \fProc{\pSend x[a_1,b_1]} \in \nodes(\fProc{P_1}) \wedge \fProc{\pSend x[a_2,b_2]} \in \nodes(\fProc{P_2})
        \\ & \wedge
        \begin{array}[t]{@{}l@{}l@{}}
            \fProc{a_1},\fProc{b_1} \in \dom(\fType{\Gamma}) \implies {}
            & \fProc{a_1} = \fProc{b_1} \wedge \fProc{a_2} = \fProc{b_2}
            \wedge \fProc{P_1} \sc \fProc{\pSend x[a_1,b_1] \| P'_1}
            \wedge \fProc{P_2} \sc \fProc{\pSend x[a_2,b_2] \| P'_2}
            \\ & {}\wedge{}
            (\fCtx{E_1 \big[ \pSend x[a,b] \| \fProc{P'_1} \big]} ; \fCtx{E_2 \big[ \pSend x[a,b] \| \fProc{P'_2} \big]}) \in \termrel{\fIFC{\Omega}}{\fIFC{\xi}}{\fType{\Gamma \setminus \fProc{a},\fProc{b}}}
        \end{array}
        \\ & \wedge
        (\fProc{a_1} \in \dom(\fType{\Gamma}) \wedge \fProc{b_1} \notin \dom(\fType{\Gamma})) \implies \big(
            \fProc{a_1} = \fProc{a_2} \wedge \fProc{b_2} \notin \dom(\fType{\Gamma})
            \\ & \hphantom{{}\wedge{}} \begin{array}[t]{@{}l@{}}
                {}\wedge
                \fProc{P_1} \sc \fProc{\nu{b_1b'} ( \pSend x[a_1,b_1] \| P'_1 )}
                \wedge
                \fProc{P_2} \sc \fProc{\nu{b_2b'} ( \pSend x[a_2,b_2] \| P'_2 )}
                \\
                {}\wedge
                (\fCtx{E_1 \big[ \nu{b_1b'} ( \pSend x[a_1,b_1] \| \fProc{P'_1} ) \big]} ; \fCtx{E_2 \big[ \nu{b_2b'} ( \pSend x[a_2,b_2] \| \fProc{P'_2} ) \big]}) \in \termrel{\fIFC{\Omega}}{\fIFC{\xi}}{\fType{\Gamma, \fProc{b'}:B\fIFC{[c]} \setminus \fProc{a}}}
            \big)
        \end{array}
        \\ & \wedge
        (\fProc{a_1} \notin \dom(\fType{\Gamma}) \wedge \fProc{b_1} \in \dom(\fType{\Gamma})) \implies \big(
            \fProc{a_2} \notin \dom(\fType{\Gamma}) \wedge \fProc{b_1} = \fProc{b_2}
            \\ & \hphantom{{}\wedge{}}
            \begin{array}[t]{@{}l@{}}
                {}\wedge
                \fProc{P_1} \sc \fProc{\nu{a_1a'} ( \pSend x[a_1,b_1] \| P'_1 )}
                \wedge
                \fProc{P_2} \sc \fProc{\nu{a_2a'} ( \pSend x[a_2,b_2] \| P'_2 )}
                \\
                {}\wedge
                (\fCtx{E_1 \big[ \nu{a_1a'} ( \pSend x[a_1,b_1] \| \fProc{P'_1} ) \big]} ; \fCtx{E_2 \big[ \nu{a_2a'} ( \pSend x[a_2,b_2] \| \fProc{P'_2} ) \big]}) \in \termrel{\fIFC{\Omega}}{\fIFC{\xi}}{\fType{\Gamma, \fProc{a'}:C\fIFC{[c]} \setminus \fProc{b}}}
            \big)
        \end{array}
        \\ & \wedge
        \fProc{a_1},\fProc{b_1} \notin \dom(\fType{\Gamma}) \implies \big(
            \fProc{a_2},\fProc{b_2} \notin \dom(\fType{\Gamma})
            \\ & \hphantom{{}\wedge{}}
            \begin{array}[t]{@{}l@{}}
                {}\wedge
                \fProc{P_1} \sc \fProc{\nu{a_1a'} \nu{b_1b'} ( \pSend x[a_1,b_1] \| P'_1 )}
                \wedge
                \fProc{P_2} \sc \fProc{\nu{a_2a'} \nu{b_2b'} ( \pSend x[a_2,b_2] \| P'_2 )}
                \\
                {}\wedge
                ( \begin{array}[t]{@{}l@{}}
                    \fCtx{E_1 \big[ \nu{a_1a'}\nu{b_1b'} ( \pSend x[a_1,b_1] \| \fProc{P'_1} ) \big]} ;
                    \\
                    \fCtx{E_2 \big[ \nu{a_2a'}\nu{b_2b'} ( \pSend x[a_2,b_2] \| \fProc{P'_2} ) \big]}
                ) \in \termrel{\fIFC{\Omega}}{\fIFC{\xi}}{\fType{\Gamma, \fProc{a'}:A\fIFC{[c]}, \fProc{b'}:B\fIFC{[c]}}}
                \big)
            \end{array}
        \end{array}
    \end{align*}
}
    \caption{
        Value interpretation, output cases ($\fType{1},\fType{\oplus},\fType{\tensor}$).
    }\label{f:valRel:output}
\end{figure}

\begin{figure}[t]
   {\small \begin{align*}
        \fType{\bot}\quad
        &
        (\fCtx{E_1}, \fProc{P_1} ; \fCtx{E_2}, \fProc{P_2}) \in \valrel{\fIFC{\Omega}}{\fIFC{\xi}}{\fType{\Gamma, \fProc{x}:\bot\fIFC{[c]}}}{\fProc{x}}
        \iff \big(
            (\fCtx{E_1}, \fProc{P_1} ; \fCtx{E_2}, \fProc{P_2}) \in \netw{\fIFC{\Omega};\fIFC{\xi}}(\fType{\Gamma, \fProc{x}:\bot\fIFC{[c]}})
            \\ & \wedge \big(
                \fCtx{E_1} \sc \fCtx{\nu{yx} ( \pClose y[] \| E'_1 )}
                \wedge
                \fCtx{E_2} \sc \fCtx{\nu{yx} ( \pClose y[] \| E'_2 )}
            \big)
            \\ & \hphantom{\wedge} \implies
            (\fCtx{E'_1 \big[ \fProc{\nu{yx} ( \pClose y[] \| P_1 )} \big]} ; \fCtx{E'_2 \big[ \fProc{\nu{yx} ( \pClose y[] \| P_2 )} \big]}) \in \termrel{\fIFC{\Omega}}{\fIFC{\xi}}{\fType{\Gamma}}
        \big)
        \\[2pt] \hline \displaybreak[1] \\[-12pt]
        \fType{\with}\quad
        &
        (\fCtx{E_1}, \fProc{P_1} ; \fCtx{E_2}, \fProc{P_2}) \in \valrel{\fIFC{\Omega}}{\fIFC{\xi}}{\fType{\Gamma, \fProc{x}:\with \{ i : A_i \}_{i \in I}\fIFC{[c]}}}{\fProc{x}}
        \iff
        \numberthis\label{eq:logRel:val:with1}
        \\ & \hphantom{{}\wedge{}}
        (\fCtx{E_1}, \fProc{P_1} ; \fCtx{E_2}, \fProc{P_2}) \in \netw{\fIFC{\Omega};\fIFC{\xi}}(\fType{\Gamma, \fProc{x}:\with \{ i : A_i \}_{i \in I}\fIFC{[c]}})
        \numberthis\label{eq:logRel:val:with2}
        \\ & \wedge \big( \exists \fProc{j} \in \fProc{I}.~
            \fCtx{E_1} \sc \fCtx{\nu{bb'} \nu{yx} ( \pSel y[b]<j \| E'_1 )}
            \wedge
            \fCtx{E_2} \sc \fCtx{\nu{bb'} \nu{yx} ( \pSel y[b]<j \| E'_2 )}
        \big)
        \numberthis\label{eq:logRel:val:with3}
        \\ & \hphantom{\wedge} \implies
        ( \begin{array}[t]{@{}l@{}}
            \fCtx{\nu{bb'} E'_1 \big[ \fProc{\nu{yx} ( \pSel y[b]<j \| P_1 )} \big]} ;
            \\
            \fCtx{\nu{bb'} E'_2 \big[ \fProc{\nu{yx} ( \pSel y[b]<j \| P_2 )} \big]}) \in \termrel{\fIFC{\Omega}}{\fIFC{\xi}}{\fType{\Gamma , \fProc{b}:A_j\fIFC{[c]}}}
        \end{array}
        \numberthis\label{eq:logRel:val:with4}
        \\[2pt] \hline \displaybreak[1] \\[-12pt]
        \fType{\parr}\quad
        &
        (\fCtx{E_1}, \fProc{P_1} ; \fCtx{E_2}, \fProc{P_2}) \in \valrel{\fIFC{\Omega}}{\fIFC{\xi}}{\fType{\Gamma, \fProc{x}:A \parr B\fIFC{[c]}}}{\fProc{x}}
        \iff
        \\ & \hphantom{{}\wedge{}}
        (\fCtx{E_1}, \fProc{P_1} ; \fCtx{E_2}, \fProc{P_2}) \in \netw{\fIFC{\Omega};\fIFC{\xi}}(\fType{\Gamma, \fProc{x}:A \parr B\fIFC{[c]}})
        \\ & \wedge \big(
            \fCtx{E_1} \sc \fCtx{\nu{aa'} \nu{bb'} \nu{yx} ( \pSend y[a,b] \| E'_1 )}
            \wedge
            \fCtx{E_2} \sc \fCtx{\nu{aa'} \nu{bb'} \nu{yx} ( \pSend y[a,b] \| E'_2 )}
        \big)
        \\ & \hphantom{{}\wedge{}} \implies
        ( \begin{array}[t]{@{}l@{}}
            \fCtx{\nu{aa'} \nu{bb'} E'_1 \big[ \fProc{\nu{yx} ( \pSend y[a,b] \| P_1 )} \big]} ;
            \\
            \fCtx{\nu{aa'} \nu{bb'} E'_2 \big[ \fProc{\nu{yx} ( \pSend y[a,b] \| P_2 )} \big]}) \in \termrel{\fIFC{\Omega}}{\fIFC{\xi}}{\fType{\Gamma , \fProc{a}:A\fIFC{[c]} , \fProc{b}:B\fIFC{[c]}}}
        \end{array}
    \end{align*}}

    \caption{Value interpretation, input cases ($\fType{\bot},\fType{\with},\fType{\parr}$).}\label{f:valRel:input}
\end{figure}

We present our value interpretation, denoted $\valrel{\fIFC{\Omega}}{\fIFC{\xi}}{\fType{\Gamma}}{\fProc{x}}$, in \Cref{f:valRel:output,f:valRel:input}.
It relates pairs of context-process tuples $(\fCtx{E_1},\fProc{P_1};\fCtx{E_2},\fProc{P_2})$, given a secrecy lattice $\fIFC{\Omega}$, a secrecy level $\fIFC{\xi} \in \dom(\fIFC{\Omega})$, an (observable) interface $\fType{\Gamma}$, and a name $\fProc{x} \in \dom(\fType{\Gamma})$.
The relation is defined by cases on the type $\fType{A}$ assigned to $\fProc{x}$ in $\fType{\Gamma}$.
If $\fType{A}$ is output-like ($\fType{1},\fType{\oplus},\fType{\tensor}$; \Cref{f:valRel:output}), the relation looks for a corresponding output on $\fProc{x}$ in the processes to (observably) move across the interface into the contexts\footnote{
    In the rest of the paper, we often write $\fType{\Gamma \setminus \fProc{x}}$ to denote $\fType{\Gamma'}$ given $\fType{\Gamma} = \fType{\Gamma', \fProc{x}:A\fIFC{[c]}}$.
}; if $\fType{A}$ is input-like ($\fType{\bot},\fType{\with},\fType{\parr}$; \Cref{f:valRel:input}), the relation looks for a corresponding output on a name $\fProc{y}$ connected by restriction to $\fProc{x}$ in the contexts to (observably) move across the interface into the processes.
We detail the representative cases where $\fType{A} \in \{\fType{\oplus},\fType{\with}\}$.

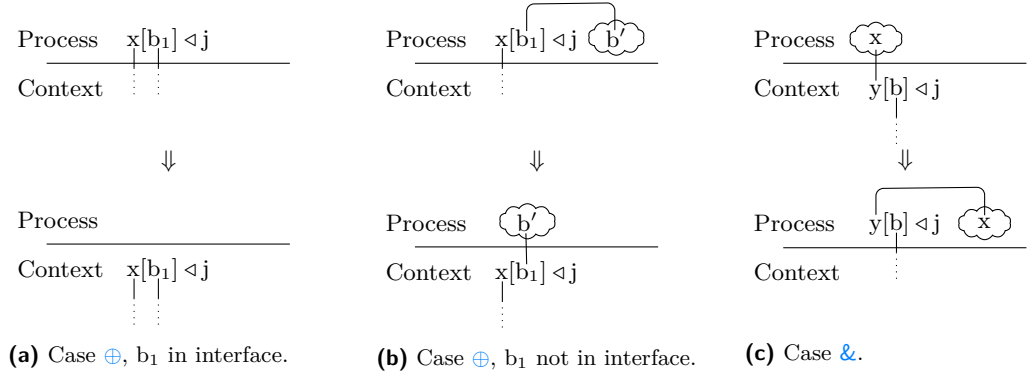
\begin{figure}[t]
   {\small
    \newlength{\xName} \newlength{\bName} \newlength{\bpName} \newlength{\yName}
    \begin{subfigure}[t]{.3\textwidth}
        \begin{tikzpicture}[font=\small,baseline={(0,0)}]
            \draw (5mm, 0) -- (\textwidth-5mm, 0);
            \node[anchor=south west] at (0, 1mm) {Process};
            \node[anchor=north west] at (0, -1mm) {Context};
            \node (sel) [anchor=south] at (.5\textwidth, .4mm) {$\fProc{\pSel x[b_1]<j}$};
            \setlength{\xName}{-4.5mm}
            \draw ([shift={(\xName,1.6mm)}]sel.south) -- ([shift={(\xName,-1mm)}]sel.south);
            \draw[dotted] ([shift={(\xName,-1mm)}]sel.south) -- ([shift={(\xName,-5mm)}]sel.south);
            \setlength{\bName}{-1.3mm}
            \draw ([shift={(\bName,1.6mm)}]sel.south) -- ([shift={(\bName,-1mm)}]sel.south);
            \draw[dotted] ([shift={(\bName,-1mm)}]sel.south) -- ([shift={(\bName,-5mm)}]sel.south);
        \end{tikzpicture}
        \vspace{3mm}
        \begin{center} $\Downarrow$ \end{center}
        \vspace{1mm}
        \begin{tikzpicture}[font=\small,baseline={(0,0)}]
            \draw (5mm, 0) -- (\textwidth-5mm, 0);
            \node[anchor=south west] at (0, 1mm) {Process};
            \node[anchor=north west] at (0, -1mm) {Context};
            \node (sel) [anchor=north] at (.5\textwidth, -.7mm) {$\fProc{\pSel x[b_1]<j}$};
            \setlength{\xName}{-4.5mm}
            \draw ([shift={(\xName,1.6mm)}]sel.south) -- ([shift={(\xName,-1mm)}]sel.south);
            \draw[dotted] ([shift={(\xName,-1mm)}]sel.south) -- ([shift={(\xName,-5mm)}]sel.south);
            \setlength{\bName}{-1.3mm}
            \draw ([shift={(\bName,1.6mm)}]sel.south) -- ([shift={(\bName,-1mm)}]sel.south);
            \draw[dotted] ([shift={(\bName,-1mm)}]sel.south) -- ([shift={(\bName,-5mm)}]sel.south);
        \end{tikzpicture}
        \caption{Case~$\fType{\oplus}$, $\fProc{b_1}$ in interface.}\label{f:valRel:oplusIn}
    \end{subfigure}\hfill
    \begin{subfigure}[t]{.3\textwidth}
        \begin{tikzpicture}[font=\small,baseline={(0,0)}]
            \draw (5mm, 0) -- (\textwidth-5mm, 0);
            \node[anchor=south west] at (0, 1mm) {Process};
            \node[anchor=north west] at (0, -1mm) {Context};
            \node (sel) [anchor=south] at (.5\textwidth, .4mm) {$\fProc{\pSel x[b_1]<j}$};
            \node[cloud, draw, cloud puffs=10, cloud puff arc=120, aspect=2, inner sep=0mm] (bp) [anchor=south] at (.75\textwidth, 1mm) {$\fProc{b'}$};
            \setlength{\xName}{-4.5mm}
            \draw ([shift={(\xName,1.6mm)}]sel.south) -- ([shift={(\xName,-1mm)}]sel.south);
            \draw[dotted] ([shift={(\xName,-1mm)}]sel.south) -- ([shift={(\xName,-5mm)}]sel.south);
            \setlength{\bName}{-1.3mm}
            \setlength{\bpName}{-.2mm}
            \draw ([shift={(\bName,-1.6mm)}]sel.north) -- ([shift={(\bName,1mm)}]sel.north) to[out=90,in=180] ([shift={(\bName+1mm,2mm)}]sel.north) -- ([shift={(\bpName-1mm,2.2mm)}]bp.north) to[out=0,in=90] ([shift={(\bpName,1.2mm)}]bp.north) -- ([shift={(\bpName,-1.4mm)}]bp.north);
        \end{tikzpicture}
        \vspace{3mm}
        \begin{center} $\Downarrow$ \end{center}
        \vspace{1mm}
        \begin{tikzpicture}[font=\small,baseline={(0,0)}]
            \draw (5mm, 0) -- (\textwidth-5mm, 0);
            \node[anchor=south west] at (0, 1mm) {Process};
            \node[anchor=north west] at (0, -1mm) {Context};
            \node (sel) [anchor=north] at (.5\textwidth, -.7mm) {$\fProc{\pSel x[b_1]<j}$};
            \setlength{\bName}{-1.3mm}
            \node[cloud, draw, cloud puffs=10, cloud puff arc=120, aspect=2, inner sep=0mm] (bp) [anchor=south] at (.5\textwidth+\bName+.1mm, 1mm) {$\fProc{b'}$};
            \setlength{\xName}{-4.5mm}
            \draw ([shift={(\xName,1.6mm)}]sel.south) -- ([shift={(\xName,-1mm)}]sel.south);
            \draw[dotted] ([shift={(\xName,-1mm)}]sel.south) -- ([shift={(\xName,-5mm)}]sel.south);
            \setlength{\bpName}{-.2mm}
            \draw ([shift={(\bName,-1.6mm)}]sel.north) -- ([shift={(\bpName,.8mm)}]bp.south);
        \end{tikzpicture}
        \caption{Case~$\fType{\oplus}$, $\fProc{b_1}$ not in interface.}\label{f:valRel:oplusNotIn}
    \end{subfigure}\hfill
    \begin{subfigure}[t]{.3\textwidth}
        \begin{tikzpicture}[font=\small,baseline={(0,0)}]
            \draw (5mm, 0) -- (\textwidth-5mm, 0);
            \node[anchor=south west] at (0, 1mm) {Process};
            \node[anchor=north west] at (0, -1mm) {Context};
            \node (sel) [anchor=north] at (.5\textwidth, -.7mm) {$\fProc{\pSel y[b]<j}$};
            \setlength{\yName}{-3.8mm}
            \setlength{\xName}{0mm}
            \node[cloud, draw, cloud puffs=10, cloud puff arc=120, aspect=2, inner sep=.5mm] (x) [anchor=south] at (.5\textwidth+\yName+.0mm, 1mm) {$\fProc{x}$};
            \draw ([shift={(\yName,-1.6mm)}]sel.north) -- ([shift={(\xName,.8mm)}]x.south);
            \setlength{\bName}{-1.1mm}
            \draw ([shift={(\bName,1.6mm)}]sel.south) -- ([shift={(\bName,-1mm)}]sel.south);
            \draw[dotted] ([shift={(\bName,-1mm)}]sel.south) -- ([shift={(\bName,-5mm)}]sel.south);
        \end{tikzpicture}
        \vspace{-3mm}
        \begin{center} $\Downarrow$ \end{center}
        \vspace{-1mm}
        \begin{tikzpicture}[font=\small,baseline={(0,0)}]
            \draw (5mm, 0) -- (\textwidth-5mm, 0);
            \node[anchor=south west] at (0, 1mm) {Process};
            \node[anchor=north west] at (0, -1mm) {Context};
            \node (sel) [anchor=south] at (.5\textwidth, .4mm) {$\fProc{\pSel y[b]<j}$};
            \node[cloud, draw, cloud puffs=10, cloud puff arc=120, aspect=2, inner sep=.5mm] (x) [anchor=south] at (.75\textwidth, 1mm) {$\fProc{x}$};
            \setlength{\yName}{-3.8mm}
            \setlength{\xName}{0mm}
            \setlength{\bName}{-1.1mm}
            \draw ([shift={(\bName,1.6mm)}]sel.south) -- ([shift={(\bName,-1mm)}]sel.south);
            \draw[dotted] ([shift={(\bName,-1mm)}]sel.south) -- ([shift={(\bName,-5mm)}]sel.south);
            \draw ([shift={(\yName,-1.5mm)}]sel.north) -- ([shift={(\yName,1mm)}]sel.north) to[out=90,in=180] ([shift={(\yName+1mm,2mm)}]sel.north) -- ([shift={(\xName-1mm,2.5mm)}]x.north) to[out=0,in=90] ([shift={(\xName,1.5mm)}]x.north) -- ([shift={(\xName,-1.2mm)}]x.north);
        \end{tikzpicture}
        \vspace{5mm}
        \caption{Case~$\fType{\with}$.}\label{f:valRel:with}
    \end{subfigure}
    }
    \caption{
        Illustrations of the value interpretation on selections: the selection is moved to/from the process, influencing name connections through the interface.
        Names in clouds represent parts of the process where the name is used.
    }\label{f:valRel:diags}
\end{figure}

When $\fType{A} = \fType{\oplus \{i:A_i\}_{i \in I}}$~\eqref{eq:logRel:val:oplus1}, we first check well typedness as usual~\eqref{eq:logRel:val:oplus2} (cf.\ \Cref{d:projNetw}).
We then assert that both $\fProc{P_1}$ and $\fProc{P_2}$ have ready a selection on $\fProc{x}$, both on the same label $\fProc{j} \in \fProc{I}$~\eqref{eq:logRel:val:oplus3}.
We find that the selections carry continuations $\fProc{b_1}$ and $\fProc{b_2}$, respectively.
Since we intend to move the selections across the interface into the contexts, we need to inspect where these $\fProc{b_i}$ are bound: in the context or in the process.
\begin{itemize}

    \item
        If $\fProc{b_1}$ appears in the interface~\eqref{eq:logRel:val:oplus4}, it is bound in $\fCtx{E_1}$.
        We then assert that $\fProc{b_1}$ and $\fProc{b_2}$ actually represent the same name, and thus that $\fProc{b_2}$ is bound in $\fCtx{E_2}$.
        Next, we use structural congruence (\Cref{d:procSCRedd}) to separate the selections from the rest of the processes.
        The case ends with a call on the term interpretation, where the selections have been moved into the contexts.
        Note that here we remove $\fProc{b_1}$ ($= \fProc{b_2}$) from the interface, as the processes have relinquished control over this name to their respective contexts: we no longer need to monitor behavior on $\fProc{b_1}$.
        \Cref{f:valRel:oplusIn} illustrates this case.

    \item
        If $\fProc{b_1}$ does not appear in the interface~\eqref{eq:logRel:val:oplus5}, it is bound in $\fProc{P_1}$.
        We first assert that $\fProc{b_2}$ also does not appear in the interface, and thus is bound in $\fProc{P_2}$.
        We then use structural congruence to identify the names to which each $\fProc{b_i}$ is bound---since they are both bound, we conveniently apply alpha conversion and use $\fProc{b'}$ in both cases---, and to separate the selections from the rest of the processes.
        Finally, we call on the term interpretation, where the selections along with the binders $\fProc{\nu{b_ib'}}$ are moved into the contexts.
        Here, we add $\fProc{b'}$ to the interface, as it must be used in the remainder of the processes, and thus must be monitored.
        \Cref{f:valRel:oplusNotIn} shows this case.

\end{itemize}

When $\fType{A} = \fType{\with \{ i : A_i \}_{i \in I}}$~\eqref{eq:logRel:val:with1}, the processes are expecting a selection from the contexts.
We again start with the usual well typedness check~\eqref{eq:logRel:val:with2} (cf.\ \Cref{d:projNetw}).
The purpose of our relation is to compare runs of the same process in different contexts, and so we cannot make assertions about the readiness of the contexts to make the required selection, or that these are selections of the same label.
We therefore proceed only under the condition that indeed the contexts are both ready to select the same label~\eqref{eq:logRel:val:with3}.
This condition uses structural congruence to identify the names in the contexts to which $\fProc{x}$ is connected, conveniently referred to as $\fProc{y}$ in both contexts.
It also identifies the continuations~$\fProc{b}$ of the selections and the names $\fProc{b'}$ to which they are connected, as well as the remainder of the contexts.
It then calls on the term evaluation~\eqref{eq:logRel:val:with4}, where the selections along with the restrictions binding $\fProc{x}$ to $\fProc{y}$ are moved into the processes.
As such, $\fProc{x}$ is no longer in the interface, but now the continuations of the selections are: we add $\fProc{b}$ to the interface.
\Cref{f:valRel:with} illustrates this case.

Finally, we use our logical relation to define \emph{equivalence up to observable messages}.
We say two processes are equivalent up to secrecy level $\fIFC{\xi}$ if they agree on their observable interface and they are related by the logical relation when placed inside any two arbitrary evaluation contexts (cf.\ \Cref{d:evalCtx}).
This ensures that, regardless of the context in which the processes run, they will behave the same with respect to the observable interface.

\begin{definition}[Equivalence up to Observable Messages]
    \label{d:DSNIRel}
    The relation
    \[
        ( \fType{\fIFC{\Omega} \vdash \fIFC{\fProc{P_1} \at d_1} :: \Gamma_1} ) \equiv_{\fIFC{\xi}} ( \fType{\fIFC{\Omega} \vdash \fIFC{\fProc{P_2} \at d_2} :: \Gamma_2} )
    \]
    holds if and only if
    $\fType{\Gamma_1 \proj \fIFC{\xi}} = \fType{\Gamma_2 \proj \fIFC{\xi}} = \fType{\Gamma}$,
    and
    for every $\fCtx{E_1},\fCtx{E_2}$ such that $\fType{\fIFC{\Omega} \vdash \fIFC{\fCtx{E_1[\fProc{P_1}]} \at d_1} :: \emptyset}$ and $\fType{\fIFC{\Omega} \vdash \fIFC{\fCtx{E_2[\fProc{P_2}]} \at d_2} :: \emptyset}$, $(\fCtx{E_1[\fProc{P_1}]};\fCtx{E_2[\fProc{P_2}]}) \in \termrel{\fIFC{\Omega}}{\fIFC{\xi}}{\fType{\Gamma}}$ and $(\fCtx{E_2[\fProc{P_2}]};\fCtx{E_1[\fProc{P_1}]}) \in \termrel{\fIFC{\Omega}}{\fIFC{\xi}}{\fType{\Gamma}}$.
\end{definition}

\begin{example}
    \label{x:equiv}
    Consider again the secure variant of $\fProc{Gov_A^{\fIFC{H}}}$ from \Cref{x:syntax}.
    Anticipating noninterference, it is straightforward to check that the continuations of the initial branch are equivalent up to observable messages:
    \begin{align*}
        & ( \fType{\vdash \fIFC{\fProc{\nu{a_I^{1'} a_I^1} ( \pSel a_I[a_I^{1'}]<act \| \pWait a_{\fIFC{L}}^1() ; \pClose a_I^1[] )} \at L} :: \fProc{a_I} : \oplus \{ act : 1 , wait : 1 \}\fIFC{[H]} , \fProc{a_{\fIFC{L}}^1} : \bot\fIFC{[L]}} )
        \\
        \equiv_{\fIFC{L}}
        & ( \fType{\vdash \fIFC{\fProc{\nu{a_I^{1'} a_I^1} ( \pSel a_I[a_I^{1'}]<wait \| \pWait a_{\fIFC{L}}^1() ; \pClose a_I^1[] )} \at L} :: \fProc{a_I} : \oplus \{ act : 1 , wait : 1 \}\fIFC{[H]} , \fProc{a_{\fIFC{L}}^1} : \bot\fIFC{[L]}} )
    \end{align*}
    The crucial part is that the different selections on $\fProc{a_I}$ are unobservable.

    On the other hand, consider also the insecure variant of $\fProc{Gov_A^{\fIFC{H}}}$ from \Cref{x:syntax}.
    Even though their typing contexts are equal (and, hence, so are the projections onto $\fIFC{L}$), the continuations of the initial branch are \emph{not} equivalent up to observable messages:
    \begin{align*}
        & ( \fType{\vdash \fIFC{\fProc{\nu{a_{\fIFC{L}}^{1'} a_{\fIFC{L}}^1} ( \pSel a_{\fIFC{L}}[a_{\fIFC{L}}^{1'}]<inf_1 \| \pWait a_I^1() ; \pClose a_{\fIFC{L}}^1[] )}} :: \fProc{a_{\fIFC{L}}} : \oplus \{ inf_1 : 1 , inf_2 : 1 \}\fIFC{[L]} , \fProc{a_I^1} : \bot\fIFC{[H]}} )
        \\
        \not\equiv_{\fIFC{L}}
        & ( \fType{\vdash \fIFC{\fProc{\nu{a_{\fIFC{L}}^{1'} a_{\fIFC{L}}^1} ( \pSel a_{\fIFC{L}}[a_{\fIFC{L}}^{1'}]<inf_2 \| \pWait a_I^1() ; \pClose a_{\fIFC{L}}^1[] )}} :: \fProc{a_{\fIFC{L}}} : \oplus \{ inf_1 : 1 , inf_2 : 1 \}\fIFC{[L]} , \fProc{a_I^1} : \bot\fIFC{[H]}} )
    \end{align*}
    Here, the different selections on $\fProc{a_{\fIFC{L}}}$ \emph{are} observable.
\end{example}

\section{Deadlock-Sensitive Noninterference (DSNI)}
\label{s:DSNI}

Our main result is that the observable behavior (up to a given secrecy level $\fIFC{\xi}$) of any well-typed process is the same when placed in different contexts.
We formalize this using our logical relation (\Cref{d:DSNIRel}):

\begin{restatable}[DSNI]{theorem}{tDSNI}
    \label{t:DSNI}
    For all secrecy lattices $\fIFC{\Omega}$, secrecy levels $\fIFC{\xi} \in \dom(\fIFC{\Omega})$ and processes $\fType{\fIFC{\Omega} \vdash \fIFC{\fProc{P} \at d} :: \Gamma}$, we have
    \(
        ( \fType{\fIFC{\Omega} \vdash \fIFC{\fProc{P} \at d} :: \Gamma} ) \equiv_{\fIFC{\xi}} ( \fType{\fIFC{\Omega} \vdash \fIFC{\fProc{P} \at d} :: \Gamma} ).
    \)
\end{restatable}

\begin{example}
    Following up on \Cref{x:equiv}, we can conclude that DSNI holds for the secure variant of $\fProc{Gov_A^{\fIFC{H}}}$, but not for the insecure variant.
\end{example}

To prove this main result, we prove a more general result (the \emph{fundamental theorem}; \Cref{t:fundamental}) that relates two processes through \Cref{d:DSNIRel} given that they are \emph{observably equivalent}.
We first define precisely what we mean with observable equivalence before presenting and proving our fundamental theorem in \Cref{s:DSNI:fundamental}.

\subsection{Observable Equivalence}
\label{s:DSNI:obsEq}

Towards defining observable equivalence, we want to identify the nodes (cf.\ \Cref{d:nodesNf}) of processes that can contribute to messages on observable names in the interface, referred to as \emph{relevant nodes}.
Nodes with running secrecy $\fIFC{\not\lleq \xi}$ obviously cannot influence observable interface names.
However, nodes with running secrecy $\fIFC{\lleq \xi}$ are not necessarily capable of influencing observable interface names either.
In particular, two types of nodes with running secrecy~$\fIFC{\lleq \xi}$ cannot influence the observable interface:
\begin{itemize}

    \item
        Nodes that input on unobservable names increase their running secrecy after the input, such that they no longer influence observable interface names.

    \item
        Nodes that output on unobservable names can only influence nodes that input on unobservable names (and thus cannot influence observable interface names indirectly via the receiving node).

\end{itemize}

\noindent
The following notion of \emph{quasi-running secrecy} anticipates these scenarios by assigning a secrecy level to a process based on the influence of its foremost prefix corresponding to the subsequent input/output.
It is defined as the join of the current running secrecy of the process and the secrecy level of the name on which the next input/output occurs.
If either of the two levels is unobservable, the quasi-running secrecy will be unobservable.
In such cases, we know that the foremost prefix of the process cannot influence the observable interface.

\begin{definition}[Quasi-running Secrecy]
    \label{d:quasi}
    Given a node typed $\fType{\fIFC{\Omega} \vdash \fIFC{\fProc{P} \at d} :: \Gamma}$, we define the \emph{quasi-running secrecy} of $\fProc{P}$, denoted $\fIFC{\quasi(\fType{\fIFC{\Omega} \vdash \fIFC{\fProc{P} \at d} :: \Gamma})}$ as follows:
    \[
        \fIFC{\quasi(\fType{\fIFC{\Omega} \vdash \fIFC{\fProc{P} \at d} :: \Gamma})} := \begin{cases}
            \fIFC{d \lcup c} & \text{if $\fProc{P} = \fProc{\pClose x[]}$ and $\fType{\fProc{x}:1\fIFC{[c]}} \in \fType{\Gamma}$}
            \\
            \fIFC{d \lcup c} & \text{if $\fProc{P} = \fProc{\pWait x() ; P'}$ and $\fType{\fProc{x}:\bot\fIFC{[c]}} \in \fType{\Gamma}$}
            \\
            \fIFC{d \lcup c} & \text{if $\fProc{P} = \fProc{\pSel x[b]<j}$ and $\fType{\fProc{x}:\oplus \{ i : A_i \}_{i \in I}\fIFC{[c]}} \in \fType{\Gamma}$}
            \\
            \fIFC{d \lcup c} & \text{if $\fProc{P} = \fProc{\pBra x(z)>\{ i : P_i \}_{i \in I}}$ and $\fType{\fProc{x}:\with \{ i : A_i \}_{i \in I}\fIFC{[c]}} \in \fType{\Gamma}$}
            \\
            \fIFC{d \lcup c} & \text{if $\fProc{P} = \fProc{\pSend x[a,b]}$ and $\fType{\fProc{x}:A \tensor B\fIFC{[c]}} \in \fType{\Gamma}$}
            \\
            \fIFC{d \lcup c} & \text{if $\fProc{P} = \fProc{\pRecv x(y,z) ; P'}$ and $\fType{\fProc{x}:A \parr B\fIFC{[c]}} \in \fType{\Gamma}$}
        \end{cases}
    \]
\end{definition}

To compute which nodes of a process are relevant, we start with nodes that have connections to the interface (through free names).
We then look at nodes that are connected to these relevant nodes through restrictions.
However, not all connections imply a possible influence on the observable interface.
Consider a node $\fProc{\pSend x[a,b]}$ that is connected to a relevant node on~$\fProc{a}$: the node does not define behavior on $\fProc{a}$ but merely outputs the name, and so~$\fProc{a}$ cannot influence the observable interface through this name.
For example, in $\fProc{\nu{xy} \nu{au} ( \pSend x[a,b] \| \pRecv y(z,w) ; \pWait z() ; \ldots \| \pClose u[] )}$, the name $\fProc{a}$ is not used for communication until it has been received on $\fProc{y}$; hence, the close on $\fProc{u}$ is not considered relevant even if the send on $\fProc{x}$ were relevant.
We make this precise by defining \emph{free communication names}: free names that are used as the subjects of unblocked prefixes.

\begin{definition}[Free Communication Names]
    \label{d:fcn}
    We define the \emph{free communication names} of~$\fProc{P}$, denoted $\fcn(\fProc{P})$, as follows:
    \begin{align*}
        \fcn(\fProc{0}) &:= \emptyset
        \\
        \fcn(\fProc{P \| Q}) &:= \fcn(\fProc{P}) \cup \fcn(\fProc{Q})
        &
        \fcn(\fProc{\nu{xy} P}) &:= \fcn(\fProc{P}) \setminus \{\fProc{x},\fProc{y}\}
        \\
        \fcn(\fProc{\pClose x[]}) &:= \{\fProc{x}\}
        &
        \fcn(\fProc{\pWait x() ; P}) &:= \{\fProc{x}\} \cup \fcn(\fProc{P})
        \\
        \fcn(\fProc{\pSend x[a,b]}) &:= \{\fProc{x}\}
        &
        \fcn(\fProc{\pRecv x(y,z) ; P}) &:= \{\fProc{x}\} \cup \fcn(\fProc{P}) \setminus \{\fProc{y},\fProc{z}\}
        \\
        \fcn(\fProc{\pSel x[b]<j}) &:= \{\fProc{x}\}
        &
        \fcn(\fProc{\pBra x(z)>\{ i: P_i \}_{i \in I}}) &:= \{\fProc{x}\} \cup \bigcup_{\fProc{i} \in \fProc{I}} \fcn(\fProc{P_i}) \setminus \{\fProc{z}\}
    \end{align*}
\end{definition}

We now have all the ingredients to determine the relevant nodes of a process.
We define the set of relevant nodes of a process inductively by following chains of nodes connected through restriction (of which there are finitely many).
We start with nodes connected to the interface directly, and add them if their quasi-running secrecy is $\fIFC{\lleq \xi}$.
We then keep adding nodes that are connected to already relevant nodes on observable channels (names with secrecy level $\fIFC{\lleq \xi}$) with quasi-running secrecy $\fIFC{\lleq \xi}$.

\begin{definition}[Relevant Nodes and Binders, and Relevant Form]
    \label{d:relNode}
    Suppose given a process in normal form $\fProc{P}$ typed $\fType{\fIFC{\Omega} \vdash \fIFC{\fProc{P} \at d} :: \Gamma}$.
    Suppose every node $\fProc{Q} \in \nodes(\fProc{P})$ is typed $\fType{\fIFC{\Omega} \vdash \fIFC{\fProc{Q} \at d_{\fProc{Q}}} :: \Gamma_{\fProc{Q}}}$.
    Given a secrecy level $\fIFC{\xi} \in \dom(\fIFC{\Omega})$, we define the set of \emph{relevant nodes} of $\fProc{P}$, denoted $\N(\fProc{P})$, by induction on the size of $\binders(\fProc{P})$ as follows ($\N(\fProc{P}) := \N_{|\binders(\fProc{P})|}(\fProc{P})$):
    \begin{align*}
        \N_0(\fProc{P}) &:= \{ \fProc{Q} \in \nodes(\fProc{P}) \mid \exists \fProc{z} \in \fcn(\fProc{Q}) .~ \fProc{z} \in \dom(\fType{\Gamma \proj \fIFC{\xi}}) \wedge \fIFC{\quasi(\fType{\fIFC{\Omega} \vdash \fIFC{\fProc{Q} \at d_{\fProc{Q}}} :: \Gamma_{\fProc{Q}}}) \lleq \xi} \}
        \\
        \N_{n+1}(\fProc{P}) &:= \N_n(\fProc{P}) \cup \left\{ \fProc{Q} \in \nodes(\fProc{P}) ~~\middle|~~
            \begin{array}{@{}l@{}}
                \exists \fProc{z} \in \fcn(\fProc{Q}) .~ \exists \fProc{Q'} \in \N_n(\fProc{P}) .~ \exists \fType{\fProc{w}:A_{\fProc{w}}\fIFC{[c]}} \in \fType{\Gamma_{\fProc{Q'}}} .
                \\
                \quad ( \fIFC{\Omega \Vdash c \lleq \xi} \wedge \{\fProc{z},\fProc{w}\} \in \binders(\fProc{P}) )
                \\
                {} \wedge \fIFC{\quasi(\fType{\fIFC{\Omega} \vdash \fIFC{\fProc{Q} \at d_{\fProc{Q}}} :: \Gamma_{\fProc{Q}}}) \lleq \xi}
            \end{array}
        \right\}
        \\
        & \forall 0 \leq n < |\binders(\fProc{P})|
    \end{align*}
    We also define the set of \emph{relevant binders} of $\fProc{P}$, denoted $\B(\fProc{P})$, as the subset of $\binders(\fProc{P})$ used in the inductive step of the definition of $\N(\fProc{P})$.
    We then define the \emph{relevant form} of a process in normal form $\fProc{P}$, denoted $\fProc{P \proj \fIFC{\xi}}$, as $\fProc{\nu{xy}_{\{x,y\} \in \B(P)} \prod_{Q \in \N(P)} Q}$.
\end{definition}

Processes are then observably equivalent if their relevant nodes and relevant binders are indistinguishable (up to structural congruence).

\begin{definition}[Observable Equivalence]
    \label{d:obsEq}
    We say that two processes $\fProc{P}$ and $\fProc{P'}$ are \emph{observably equivalent}, denoted $\fProc{P} \obseq_{\fIFC{\xi}} \fProc{P'}$, if and only if there are normal forms $\fProc{Q},\fProc{Q'}$ of $\fProc{P},\fProc{P'}$ respectively such that $\fProc{Q \proj \fIFC{\xi}} \sc \fProc{Q' \proj \fIFC{\xi}}$.
\end{definition}

\subsection{The Fundamental Theorem}
\label{s:DSNI:fundamental}

We now state and prove our fundamental theorem, from which DSNI (\Cref{t:DSNI}) follows.

\begin{restatable}[Fundamental Theorem]{theorem}{tFundamental}
    \label{t:fundamental}
    For all secrecy lattices $\fIFC{\Omega}$, secrecy levels $\fIFC{\xi} \in \dom(\fIFC{\Omega})$ and processes $\fType{\fIFC{\Omega} \vdash \fIFC{\fProc{P_1} \at d_1} :: \Gamma_1}$ and $\fType{\fIFC{\Omega} \vdash \fIFC{\fProc{P_2} \at d_2} :: \Gamma_2}$ with $\fProc{P_1} \obseq_{\fIFC{\xi}} \fProc{P_2}$ and $\fType{\Gamma_1 \proj \fIFC{\xi}} = \fType{\Gamma_2 \proj \fIFC{\xi}}$, we have $( \fType{\fIFC{\Omega} \vdash \fIFC{\fProc{P_1} \at d_1} :: \Gamma_1} ) \equiv_{\fIFC{\xi}} ( \fType{\fIFC{\Omega} \vdash \fIFC{\fProc{P_2} \at d_2} :: \Gamma_2} )$.
\end{restatable}

\noindent
We first give several auxiliary results and definitions, before proving \Cref{t:fundamental} on \Cpageref{proof:t:fundamentalMain}:
\begin{itemize}

    \item
        \Cref{l:redCtx} splits a process that reduces into an evaluation context (\Cref{d:evalCtx}) containing the source of the reduction originating from one of the reduction axioms in \Cref{f:procSCRedd} (bottom).

    \item
        \Cref{l:uRedd} splits unobservable reduction (\Cref{d:uRedd}) into one of three cases: reduction internal in the context, reduction internal in the process, and communication between context and process on unobservable names.

    \item
        \Cref{l:catchUp} asserts that two observably equivalent (\Cref{d:obsEq}) processes can ``catch up'' on each other's unobservable reductions (\Cref{d:uRedd}).
        That is, if one process reduces unobservably, then the other process can do zero or one unobservable reductions such that the resulting processes are again observably equivalent.

    \item
        \Cref{d:weight} defines a \emph{weight} on types and typing context, which we use for induction in the proof of \Cref{t:fundamental} on \Cpageref{proof:t:fundamentalMain}.

\end{itemize}
\Cref{l:redCtx,l:uRedd} are proven in \appref[proof:l:redCtx,proof:l:uRedd].

\begin{restatable}{lemma}{lRedCtx}
    \label{l:redCtx}
    Suppose given a process typed $\fType{\fIFC{\Omega} \vdash \fIFC{\fProc{P} \at d} :: \Gamma}$.
    If $\fProc{P} \redd \fProc{P'}$, then there exists an $\fCtx{E}$ for which either of the following holds:
    \begin{enumerate}

        \item\label{i:redCtx:closeWait}
            $\fProc{P} \sc \fCtx{E[\fProc{\nu{xy} ( \pClose x[] \| \pWait y() ; Q )}]}$ and $\fProc{P'} \sc \fCtx{E[\fProc{Q}]}$;

        \item\label{i:redCtx:sendRecv}
            $\fProc{P} \sc \fCtx{E[\fProc{\nu{xy} ( \pSend x[a,b] \| \pRecv y(z,w) ; Q )}]}$ and $\fProc{P'} \sc \fCtx{E[\fProc{Q \pSubst{ a/z,b/w }}]}$;

        \item\label{i:redCtx:selBra}
            $\fProc{P} \sc \fCtx{E[\fProc{\nu{xy} ( \pSel x[b]<j \| \pBra y(w)>\{ i : Q_i \}_{i \in I} )}]}$ for $\fProc{j} \in \fProc{I}$ and $\fProc{P'} \sc \fCtx{E[\fProc{Q_j \pSubst{ b/w }}]}$.

    \end{enumerate}
\end{restatable}

\begin{restatable}{lemma}{luRedd}
    \label{l:uRedd}
    Suppose $(\fCtx{E},\fProc{P}) \in \netw{\fIFC{\Omega};\fIFC{\xi}}(\fType{\Gamma})$ and $\fCtx{E},\fProc{P} \uRedd{\fIFC{\Omega};\fIFC{\xi};\fType{\Gamma}} \fCtx{E'},\fProc{P'}$.
    Then %either of the following holds:
    % \begin{itemize}
    % 
        % \item
            $\fCtx{E} \redd \fCtx{E'}$ and $\fProc{P} = \fProc{P'}$%;
            , or
        %
        % \item
            $\fProc{P} \redd \fProc{P'}$ and $\fCtx{E} = \fCtx{E'}$%;
            , or
        %
        % \item
            \Cref{l:redCtx} applies, on names not in $\fType{\Gamma}$.
    %
    % \end{itemize}
\end{restatable}

\begin{restatable}[Catch Up]{lemma}{lcatchUp}
    \label{l:catchUp}
    Suppose $(\fCtx{E_1},\fProc{P_1} ; \fCtx{E_2},\fProc{P_2}) \in \netw{\fIFC{\Omega};\fIFC{\xi}}(\fType{\Gamma})$ such that $\fProc{P_1} \obseq_{\fIFC{\xi}} \fProc{P_2}$.
    If $\fCtx{E_1},\fProc{P_1} \uRedd{\fIFC{\Omega};\fIFC{\xi};\fType{\Gamma}} \fCtx{E'_1},\fProc{P'_1}$, then there exists $\fProc{P'_2}$ such that $\fCtx{E_2},\fProc{P_2} \uReddQ{\fIFC{\Omega};\fIFC{\xi};\fType{\Gamma}} \fCtx{E_2},\fProc{P'_2}$ and $\fProc{P'_1} \obseq_{\fIFC{\xi}} \fProc{P'_2}$.
\end{restatable}

\begin{proof}
    For a smoother proof, we consider a normal form $\fProc{Q_1}$ of $\fProc{P_1}$, and obtain from $\fProc{Q_1}$ a normal form $\fProc{Q_2}$ of $\fProc{P_2}$ such that $\fProc{Q_1 \proj \fIFC{\xi}} = \fProc{Q_2 \proj \fIFC{\xi}}$.
    By \Cref{d:nodesNf}, the thesis follows by proving the thesis for these $\fProc{Q_1},\fProc{Q_2}$.

    By \Cref{l:uRedd}, we can distinguish three cases from which $\fCtx{E_1},\fProc{Q_1} \uRedd{\fIFC{\Omega};\fIFC{\xi};\fType{\Gamma}} \fCtx{E'_1},\fProc{Q'_1}$ follows.
    \begin{itemize}

        \item
            \textbf{(Internal in context: $\fCtx{E_1} \redd \fCtx{E'_1}$ and $\fProc{Q_1} = \fProc{Q'_1}$)}
            The thesis holds directly with $\fProc{Q'_2} := \fProc{Q_2}$.

        \item
            \textbf{(Internal in process: $\fProc{Q_1} \redd \fProc{Q'_1}$ and $\fCtx{E_1} = \fCtx{E'_1}$)}
            By \Cref{l:redCtx}, $\fProc{Q_1}$'s reduction is due to one of three possible synchronizations inside some evaluation context.
            Note that \Cref{l:redCtx} may give us processes that are alpha variant to $\fProc{Q_1}$ and $\fProc{Q'_1}$; in the following we implicitly apply further alpha renaming to match the names in $\fProc{Q_1}$ and $\fProc{Q'_1}$.
            % plscheck
            For space considerations, we sketch only the \textbf{(Close-Wait)} case; the other two cases are analogous.
            Full details are in \appref[proof:l:catchUp].

            We have $\fProc{Q_1} \sc \fCtx{F_1[\fProc{\nu{xy} ( \pClose x[] \| \pWait y() ; R )}]} \redd \fCtx{F_1[ \fProc{R} ]} \sc \fProc{Q'_1}$.
            The analysis depends on whether the close on $\fProc{x}$ is a relevant node of $\fProc{Q_1}$ or not.

            If not, we derive that the wait on $\fProc{y}$ is also not relevant.
            It follows by well typedness that the continuation $\fProc{R}$ will neither add relevant nodes nor influence relevancy of other nodes, so $\fProc{Q_1 \proj \fIFC{\xi}} = \fProc{Q'_1 \proj \fIFC{\xi}}$ and the thesis follows with $\fProc{Q'_2} := \fProc{Q_2}$.

            If the close is indeed a relevant node of $\fProc{Q_1}$, we derive that the wait on $\fProc{y}$ and the binder between $\fProc{x}$ and $\fProc{y}$ are also relevant.
            By assumption, they are then also relevant in $\fProc{Q_2}$, so we can derive a similar reduction to $\fProc{Q'_2}$.

            It remains to show that $\fProc{Q'_1 \proj \fIFC{\xi}} \sc \fProc{Q'_2 \proj \fIFC{\xi}}$, which boils down to showing that these processes have coinciding sets of relevant nodes and binders.
            Both directions of these set inclusions are analogous, so we focus on one: from $\fProc{Q'_1}$ to $\fProc{Q'_2}$.
            The analysis is by induction on the construction of the sets of relevant nodes and binders.
            In each case, we consider the appearance of the node: in $\fProc{R}$ or in $\fCtx{F_1}$.
            In both cases, a thorough analysis of how the node was included as a relevant node---through a path of relevant binders and nodes that were added before---reveals an analogous relevant node in $\fProc{Q'_2}$.

        \item
            \textbf{(Communication between context and process on names not in $\fType{\Gamma}$)}
            By definition, the secrecy levels of the involved names are incomparable to $\fIFC{\xi}$.
            Therefore, none of the nodes involved are relevant or influence relevancy of any other nodes, so $\fProc{Q_1 \proj \fIFC{\xi}} = \fProc{Q'_1 \proj \fIFC{\xi}}$ and the thesis holds with $\fProc{Q'_2} := \fProc{Q_2}$.
            \qedhere

    \end{itemize}
\end{proof}

\begin{definition}[Weight]
    \label{d:weight}
    The weight of a type $\fType{A}$, denoted $\w(\fType{A})$, is defined as follows:
    \begin{align*}
        \w(\fType{1}) &:= 1
        &
        \w(\fType{A \tensor B}) &:= \w(\fType{A}) + \w(\fType{B}) + 1
        &
        \w(\fType{\oplus \{ i:A_i \}_{i \in I}}) &:= \max_{i \in I}(\w(\fType{A_i})) + 1
        \\
        \w(\fType{\bot}) &:= 1
        &
        \w(\fType{A \parr B}) &:= \w(\fType{A}) + \w(\fType{B}) + 1
        &
        \w(\fType{\with \{ i:A_i \}_{i \in I}}) &:= \max_{i \in I}(\w(\fType{A_i})) + 1
    \end{align*}
    The weight of a typing context $\w(\fType{\Gamma})$ is the sum of the weights of its types.
\end{definition}

\tFundamental*

\begin{proof}
    \label{proof:t:fundamentalMain}
    Let $\fType{\Gamma} := \fType{\Gamma_1 \proj \fIFC{\xi}} = \fType{\Gamma_2 \proj \fIFC{\xi}}$.
    Take any $\fCtx{E_1},\fCtx{E_2}$ such that $\fType{\fIFC{\Omega} \vdash \fIFC{\fCtx{E_1[\fProc{P_1}]} \at d'_1} :: \emptyset}$ and $\fType{\fIFC{\Omega} \vdash \fIFC{\fCtx{E_2[\fProc{P_2}]} \at d'_2} :: \emptyset}$.
    We need to show that $(\fCtx{E_1[\fProc{P_1}]};\fCtx{E_2[\fProc{P_2}]}) \in \termrel{\fIFC{\Omega}}{\fIFC{\xi}}{\fType{\Gamma}}$, which we do by induction on $\w(\fType{\Gamma})$.

    The first condition is that $(\fCtx{E_1},\fProc{P_1} ; \fCtx{E_2},\fProc{P_2}) \in \netw{\fIFC{\Omega};\fIFC{\xi}}(\fType{\Gamma})$; this holds by assumption.

    Next, take any $\fCtx{E'_1},\fProc{P'_1}$ such that $\fCtx{E_1},\fProc{P_1} \uRedd*{\fIFC{\Omega};\fIFC{\xi};\fType{\Gamma}} \fCtx{E'_1},\fProc{P'_1} \nuRedd{\fIFC{\Omega};\fIFC{\xi};\fType{\Gamma}}$.
    A straightforward induction on the length of these unobservable reductions shows that, by \Cref{d:uRedd} and \Cref{l:catchUp}, there are $\fCtx{E'_2},\fProc{P'_2}$ such that $\fCtx{E_2},\fProc{P_2} \uRedd*{\fIFC{\Omega};\fIFC{\xi};\fType{\Gamma}} \fCtx{E'_2},\fProc{P'_2} \nuRedd{\fIFC{\Omega};\fIFC{\xi};\fType{\Gamma}}$, $(\fCtx{E'_1},\fProc{P'_1} ; \fCtx{E'_2},\fProc{P'_2}) \in \netw{\fIFC{\Omega};\fIFC{\xi}}(\fType{\Gamma})$, and $\fProc{P'_1} \obseq_{\fIFC{\xi}} \fProc{P'_2}$.

    Now, we need to show that, for every $\fProc{x} \in \big( \ain(\fCtx{E'_1},\fProc{P'_1}) \cup \ain(\fCtx{E'_2},\fProc{P'_2}) \big) \cap \dom(\fType{\Gamma})$,
    \[
        (\fCtx{E'_1},\fProc{P'_1} ; \fCtx{E'_2},\fProc{P'_2}) \in \valrel{\fIFC{\Omega}}{\fIFC{\xi}}{\fType{\Gamma}}{\fProc{x}}.
    \]
    Take any such $\fProc{x}$.
    Either $\fProc{x} \in \ain(\fCtx{E'_1},\fProc{P'_1})$ or $\fProc{x} \in \ain(\fCtx{E'_2},\fProc{P'_2})$; w.l.o.g., assume the former.
    The rest of the analysis depends on the type of $\fProc{x}$ in $\fType{\Gamma}$.

    First, we discuss the output-like cases ($\fType{1},\fType{\oplus},\fType{\tensor}$).
    In each case, by well typedness, $\fProc{x}$ is the subject of an output-like prefix in $\fProc{P'_1}$.
    Since $\fProc{x} \in \ain(\fCtx{E'_1},\fProc{P'_1})$, this prefix is unguarded.
    Since $\fProc{x} \in \dom(\fType{\Gamma}) = \dom(\fType{\Gamma_1 \proj \fIFC{\xi}})$, the node in which the prefix appears is relevant in $\fProc{P'_1}$.
    Therefore, since $\fProc{P'_1} \obseq_{\fIFC{\xi}} \fProc{P'_2}$, there is also a relevant node in $\fProc{P'_2}$ where this prefix appears unguarded.
    
    For space considerations, we only detail the case where $\fProc{x}$ has type $\fType{\oplus \{ i : A_i \}\fIFC{[c]}}$; the other cases are discussed in \appref[proof:t:fundamental].
    There exists $\fProc{j} \in \fProc{I}$ such that $\fProc{\pSel x[b_1]<j} \in \nodes(\fProc{P'_1})$ and $\fProc{\pSel x[b_2]<j} \in \nodes(\fProc{P'_1})$.
    The analysis depends on whether $\fProc{b_1} \in \dom(\fType{\Gamma})$ or not.
    \begin{itemize}

        \item
            \textbf{($\fProc{b_1} \in \dom(\fType{\Gamma})$)}
            By well typedness, $\fProc{b_1} \in \fn(\fProc{P'_1}) \cap \fn(\fProc{P'_2})$.
            Since $\fProc{P'_1} \obseq_{\fIFC{\xi}} \fProc{P'_2}$, then $\fProc{b_1} = \fProc{b_2}$.
            Hence, $\fProc{P'_1} \sc \fProc{\pSel x[b_1]<j \| P''_1}$ and $\fProc{P'_2} \sc \fProc{\pSel x[b_2]<j \| P''_2}$.
            Similar to the case above, and since $\w(\fType{\Gamma \setminus \fProc{x}}) < \w(\fType{\Gamma})$, it follows from the IH that $(\fCtx{E'_1\big[ \pSel x[b_1]<j \| \fProc{P''_1} \big]} ; \fCtx{E'_2\big[ \pSel x[b_2]<j \| \fProc{P''_2} \big]}) \in \termrel{\fIFC{\Omega}}{\fIFC{\xi}}{\fType{\Gamma \setminus \fProc{x}}}$.
            This proves that $(\fCtx{E'_1},\fProc{P'_1} ; \fCtx{E'_2},\fProc{P'_2}) \in \valrel{\fIFC{\Omega}}{\fIFC{\xi}}{\fType{\Gamma}}{\fProc{x}}$.

        \item
            \sloppy
            \textbf{($\fProc{b_1} \notin \dom(\fType{\Gamma})$)}
            By well typedness, $\fProc{P'_1} \sc \fProc{\nu{b_1b'} ( \pSel x[b_1]<j \| P''_1 )}$.
            The selection on $\fProc{x}$ is a relevant node of $\fProc{P'_1}$.
            Since $\fProc{P'_1} \obseq_{\fIFC{\xi}} \fProc{P'_2}$, it is also a relevant node of $\fProc{P'_2}$.
            Moreover, $\fProc{b_2} \notin \fn(\fProc{P'_2})$: otherwise, $\fProc{b_2} = \fProc{b_1}$, and then $\fProc{b_1} \in \fn(\fProc{P'_1})$.
            Hence, $\fProc{P'_2} \sc \fProc{\nu{b_2b'} ( \pSel x[b_2]<j \| P''_2 )}$.
            Clearly, $\fType{\fIFC{\Omega} \vdash \fIFC{\fProc{P''_1} \at d''_1} :: \Gamma_1 \setminus \fProc{x} , \fProc{b'}}$ and $\fType{\fIFC{\Omega} \vdash \fIFC{\fProc{P''_2} \at d''_2} :: \Gamma_2 \setminus \fProc{x} , \fProc{b'}}$, and $\fType{\fIFC{\Omega} \vdash \fIFC{\fCtx{E'_1\big[\nu{b_1b'} ( \pSel x[b_1]<j \| \fProc{P''_1} )\big]} \at d'''_1} :: \emptyset}$ and $\fType{\fIFC{\Omega} \vdash \fIFC{\fCtx{E'_2\big[\nu{b_2b'} ( \pSel x[b_2]<j \| \fProc{P''_2} )\big]} \at d'''_2} :: \emptyset}$.
            Again, since $\fProc{P'_1} \obseq_{\fIFC{\xi}} \fProc{P'_2}$, the chain of nodes and binders that are relevant in $\fProc{P'_1}$ through the binder $\fProc{\nu{b_1b'}}$ has an equivalent such chain in $\fProc{P'_2}$ through $\fProc{\nu{b_2b'}}$ and the selection on $\fProc{b_2}$.
            Hence, the effect on relevant nodes and binders by removing the binder and the selection on $\fProc{x}$ is the same on $\fProc{P''_1}$ as it is on $\fProc{P''_2}$: $\fProc{P''_1} \obseq_{\fIFC{\xi}} \fProc{P''_2}$.
            Clearly, $\fType{\Gamma_1 \setminus \fProc{x} , \fProc{b'} \proj \fIFC{\xi}} = \fType{\Gamma_2 \setminus \fProc{x} , \fProc{b'}} = \fType{\Gamma \setminus \fProc{x} , \fProc{b'}}$.
            Also, $\w(\fType{A_j}) < \w(\fType{\oplus \{ A_i \}_{i \in I}})$, so $\w(\fType{\Gamma \setminus \fProc{x} , \fProc{b'}}) < \w(\fType{\Gamma})$.
            It then follows from the IH that $(\fCtx{E'_1\big[\nu{b_1b'} ( \pSel x[b_1]<j \| \fProc{P''_1} )\big]} ; \fCtx{E'_2\big[\nu{b_2b'} ( \pSel x[b_2]<j \| \fProc{P''_2} )\big]} ) \in \termrel{\fIFC{\Omega}}{\fIFC{\xi}}{\fType{\Gamma \setminus \fProc{x} , \fProc{b'}}}$.
            This proves that $(\fCtx{E'_1},\fProc{P'_1} ; \fCtx{E'_2},\fProc{P'_2}) \in \valrel{\fIFC{\Omega}}{\fIFC{\xi}}{\fType{\Gamma}}{\fProc{x}}$.

    \end{itemize}

    % \textbf{($\fProc{x}$ has type $\fType{1\fIFC{[c]}}$)}
    % Then $\fProc{P'_1} \sc \fProc{\pClose x[] \| P''_1}$ and $\fProc{P'_2} \sc \fProc{\pClose x[] \| P''_2}$.
    % Clearly, $\fType{\fIFC{\Omega} \vdash \fIFC{\fProc{P''_1} \at d''_1} :: \Gamma_1 \setminus \fProc{x}}$ and $\fType{\fIFC{\Omega} \vdash \fIFC{\fProc{P''_2} \at d''_2} :: \Gamma_2 \setminus \fProc{x}}$, and $\fType{\fIFC{\Omega} \vdash \fIFC{\fCtx{E'_1 \big[ \pClose x[] \| \fProc{P''_1} \big]} \at d'''_1} :: \emptyset}$ and $\fType{\fIFC{\Omega} \vdash \fIFC{\fCtx{E'_2 \big[ \pClose x[] \| \fProc{P''_2} \big]} \at d'''_2} :: \emptyset}$.
    % Also clearly, $\fProc{P''_1} \obseq_{\fIFC{\xi}} \fProc{P''_2}$ and $\fType{\Gamma_1 \setminus \fProc{x} \proj \fIFC{\xi}} = \fType{\Gamma_2 \setminus \fProc{x} \proj \fIFC{\xi}} = \fType{\Gamma \setminus \fProc{x}}$.
    % Hence, since $\w(\fType{\Gamma \setminus \fProc{x}}) < \w(\fType{\Gamma})$, it follows from the IH that $(\fCtx{E'_1\big[ \pClose x[] \| \fProc{P''_1} \big]} ; \fCtx{E'_2\big[ \pClose x[] \| \fProc{P''_2} \big]}) \in \termrel{\fIFC{\Omega}}{\fIFC{\xi}}{\fType{\Gamma \setminus \fProc{x}}}$.
    % This proves that $(\fCtx{E'_1},\fProc{P'_1} ; \fCtx{E'_2},\fProc{P'_2}) \in \valrel{\fIFC{\Omega}}{\fIFC{\xi}}{\fType{\Gamma}}{\fProc{x}}$.

    Next, we discuss the negative cases ($\fType{\bot},\fType{\with},\fType{\parr}$).
    In each case, by well typedness, $\fProc{x}$ is the subject of an input-like prefix in $\fProc{P'_1}$.
    The context $\fCtx{E'_1}$ binds $\fProc{x}$ to some $\fProc{y}$ by restriction, and $\fCtx{E'_1}$ contains a complementary output-like prefix on $\fProc{y}$.
    Following similar reasoning, the same holds for $\fCtx{E'_2}$.
    Since $\fProc{x} \in \ain(\fCtx{E'_1},\fProc{P'_1})$, this output-like prefix appears unguarded in $\fCtx{E'_1}$.
    To prove the thesis, we assume that this prefix also appears unguarded in $\fCtx{E'_2}$.
    
    For space considerations, we only detail the case where $\fProc{x}$ has type $\fType{\bot\fIFC{[c]}}$; the other cases require additional care in handling continuation endpoints and are discussed in \appref[proof:t:fundamental].
    We have $\fCtx{E'_1} \sc \fCtx{\nu{yx}( \pClose y[] \| E''_1 )}$ and $\fCtx{E'_2} \sc \fCtx{\nu{yx}( \pClose y[] \| E''_2 )}$.
    Let $\fProc{P''_1} := \fProc{\nu{yx}( \pClose y[] \| P'_1 )}$ and $\fProc{P''_2} := \fProc{\nu{yx}( \pClose y[] \| P'_2 )}$.
    Clearly, $\fType{\fIFC{\Omega} \vdash \fIFC{\fProc{P''_1} \at d''_1} :: \Gamma_1 \setminus \fProc{x}}$ and $\fType{\fIFC{\Omega} \vdash \fIFC{\fProc{P''_2} \at d''_2} :: \Gamma_2 \setminus \fProc{x}}$, and $\fType{\fIFC{\Omega} \vdash \fIFC{\fCtx{E''_1[\fProc{P''_1}]} \at d'''_1} :: \emptyset}$ and $\fType{\fIFC{\Omega} \vdash \fIFC{\fCtx{E''_2[\fProc{P''_2}]} \at d'''_2} :: \emptyset}$.

    Let $\fProc{Q_1},\fProc{Q_2}$ denote the nodes of $\fProc{P'_1},\fProc{P'_2}$, respectively, in which $\fProc{x}$ appears.
    To prove that $\fProc{P''_1} \obseq_{\fIFC{\xi}} \fProc{P''_2}$, it suffices to show that $\fProc{Q_1}$ and any related binders are relevant in $\fProc{P''_1}$ if and only if $\fProc{Q_2}$ and any related binders are relevant in $\fProc{P''_2}$; any connected nodes/binders follow similar reasoning.
    We detail only the left-to-right direction; the other direction is analogous.
    Suppose $\fProc{Q_1}$ is relevant in $\fProc{P''_1}$.
    Then $\fIFC{\quasi(\fProc{Q_1}) \lleq \xi}$, and thus $\fProc{Q_1}$ is also relevant in $\fProc{P'_1}$ through $\fProc{x}$ in the interface.
    Then also $\fProc{Q_2}$ is relevant in $\fProc{P'_2}$, where $\fProc{Q_1} \sc \fProc{Q_2}$ and $\fIFC{\quasi(\fProc{Q_2}) \lleq \xi}$.
    The analysis depends on how $\fProc{Q_1}$ is relevant in $\fProc{P''_1}$: (i)~through the interface, or (ii)~through a restriction with another relevant node.
    In case~(i), it follows straightforwardly that $\fProc{Q_2}$ is also relevant in $\fProc{P'_2}$.
    In case~(ii), the connected node is also relevant in $\fProc{P'_1}$, and hence there is a related node that is also relevant in $\fProc{P'_2}$.
    Since the two processes agree on observable channels, the channel responsible for including $\fProc{Q_1}$ as a relevant node of $\fProc{P''_1}$ is also bound in $\fProc{P''_2}$.
    Then we can conclude that $\fProc{Q_2}$ is a relevant node of $\fProc{P''_2}$.

    Since $\w(\fType{\Gamma \setminus \fProc{x}}) < \w(\fType{\Gamma})$, it then follows from the IH that $(\fCtx{E''_1[\fProc{P''_1}]} ; \fCtx{E''_2[\fProc{P''_2}]}) \in \termrel{\fIFC{\Omega}}{\fIFC{\xi}}{\fType{\Gamma \setminus \fProc{x}}}$.
    This proves that $(\fCtx{E'_1},\fProc{P'_1} ; \fCtx{E'_2},\fProc{P'_2}) \in \valrel{\fIFC{\Omega}}{\fIFC{\xi}}{\fType{\Gamma}}{\fProc{x}}$.

    Finally, we show that $\aon(\fProc{P'_1}) \cap \dom(\fType{\Gamma}) = \aon(\fProc{P'_2}) \cap \dom(\fType{\Gamma})$.
    To prove this set equality, we take any $\fProc{x} \in \aon(\fProc{P'_1}) \cap \dom(\fType{\Gamma})$ and prove that $\fProc{x} \in \aon(\fProc{P'_2}) \cap \dom(\fType{\Gamma})$; the other direction is analogous.
    Clearly, $\fProc{x}$ is the subject of an output-like prefix in $\fProc{P'_1}$.
    Since $\fProc{x} \in \dom(\fType{\Gamma})$, this output-like prefix must appear unguarded in a node in $\fProc{P'_1}$.
    If the quasi-running secrecy of this node is observable, this node is relevant in $\fProc{P'_1}$.
    Since $\fProc{P'_1} \obseq_{\fIFC{\xi}} \fProc{P'_2}$, $\fProc{P'_2}$ must also have a relevant node in which the output-like prefix appears unguarded.
    Otherwise, the node is not relevant in $\fProc{P'_1}$, and hence the node in which the output-like prefix appears in $\fProc{P'_2}$ is also not relevant in $\fProc{P'_2}$.
    Hence, $\fProc{x} \in \aon(\fProc{P'_2}) \cap \dom(\fType{\Gamma})$.
\end{proof}

\subparagraph*{DSNI and deadlock freedom.}

As mentioned in \Cref{s:IFC:procLang}, our process language is based on the finite fragment of APCP with priority mechanisms removed.
By enriching our process language with APCP's priority mechanisms, we restrict well typedness to deadlock-free processes.
As such, our results remain relevant if we only consider deadlock-free processes.
% As such, our results all remain relevant if we were to only consider deadlock free processes.

\section{Related Work}
\label{s:rw}

\subparagraph*{Logical relations for session types.}

Existing logical relations for session types are primarily unary, focusing on proving termination~\cite{PerezESOP2012, PerezARTICLE2014,DeYoungFSCD2020}.
Binary logical relations have been contributed for proving parametricity~\cite{CairesESOP2013} and noninterference~\cite{DerakhshanLICS2021,report/BalzerDHY23}.
All of these logical relations are developed for intuitionistic linear session types, where process networks form trees and, as a result, neither permit cyclic networks nor deadlocks.
Whereas our logical relation has its foundations in linear session types,
it differs in that it is based on classical linear logic and allows for cycles and deadlocks.

Our work is most similar to prior work by Derakhshan et al.~\cite{DerakhshanLICS2021} and Balzer et al.~\cite{report/BalzerDHY23} on a binary logical relation, in which the authors develop a flow-sensitive IFC type system and use the logical relation to prove noninterference.
Our IFC type system is also flow sensitive, but it is designed for an adaptation of APCP (the background of which we discuss separately) that allows for cyclic process networks and deadlocks.
Our logical relation resembles the one by the authors in that it employs an interface of names along which observations can be made.  In contrast to Derakhshan et al.~\cite{DerakhshanLICS2021} and Balzer et al.~\cite{report/BalzerDHY23}, our interface is a set of names with types, rather than a sequent that singles out the providing name from the names being used.  The distinction becomes necessary in an intuitionistic linear logic setting, whereas our work is grounded in classical linear logic.
The contrast between the intuitionistic and classical setting manifests itself in other aspects of our development, too.
For example, in prior intuitionistic IFC session type systems~\cite{DerakhshanLICS2021,report/BalzerDHY23}, the offering channel always caps the secrecy of the channels in the context and the running secrecy.
In our setting, however, the running secrecy of a process and the secrecy levels of its context are not necessarily related.
In particular, waiting for a channel to close may increase the running secrecy of a process.
Once the channel is closed, it disappears from the context, leaving the running secrecy entirely unrelated to the secrecies of the remaining channels.

Like our language, Derakhshan et al.~\cite{DerakhshanLICS2021}'s language lacks any recursion construct.
On the other hand, the possibility of deadlocks introduces a side channel similar to non-termination, which is present in the work by Balzer et al.~\cite{report/BalzerDHY23} that supports general recursive session types and thus possibly looping processes.
Here, we decided to consider non-recursive processes to focus on side channels through deadlocks only and not through non-termination.  In future work we would like to consider scaling our work to support general recursive session types as well.  We envision employing an observation index similar to Balzer et al.~\cite{report/BalzerDHY23} to stratify the logical relation in the number of observable messages exchanged over the interface.  The co-presence of both non-termination and deadlocks will need careful consideration.

Our idea of an observable interface is reminiscent of the free channels with visible communications in prior work by Atkey~\cite{AtkeyESOP2017}.
Atkey establishes observational equivalence for Wadler's Classical Processes (CP) by defining a denotational semantics for CP and a logical-relations argument.
However, the logical relation in Atkey's work does not relate two processes of a certain type but rather identifies the possible observations for each type in terms of the input/output behavior of its connectives.

\subparagraph*{Cyclic process networks.}

Traditionally, (typed) $\pi$-calculi permit cyclic process networks, and as such do not guarantee deadlock freedom.
However, since the discovery of Curry-Howard correspondences between linear logic and session types~\cite{CairesCONCUR2010,journal/jfp/Wadler14}, the majority of works on session types restrict their network shapes to trees, such that deadlock freedom is guaranteed.
The line of work including APCP (to which our process language is highly related)~\cite{conf/concur/Kobayashi06,conf/lics/Padovani14,conf/fossacs/DardhaG18,conf/ice/vdHeuvelP21,journal/scico/vdHeuvelP22} considers restrictions to session type systems that allow cyclic process networks without deadlocks.

\subparagraph*{IFC type systems for multiparty session types.}

Capecchi et al.~\cite{CapecchiCONCUR2010,CapecchiARTICLE2014} explore secure information flow with controlled forms of declassification for multiparty sessions and prove a noninterference result via a bisimulation. 
Our work differs in being flow sensitive and using a binary rather than multiparty session type paradigm.
Our use of a logical relation to show noninterference, and our foundations in linear logic also set us apart.
Follow-up works by Castellani et al.~\cite{CastellaniARTICLE2016,Ciancaglini2016} also study run-time monitoring techniques to ensure secure information flow control in multiparty sessions.

\subparagraph*{IFC type systems for process calculi.}

Several approaches have been explored for designing IFC type systems that prevent the leakage of information through message passing in process calculi~\cite{HondaESOP2000,HondaYoshidaPOPL2002,CrafaARTICLE2002,CrafaTGC2005,CrafaFMSE2006,HondaYoshidaTOPLAS2007,Crafa2007,HENNESSYRIELY2002,HENNESSY20053,KobayashiARTICLE2005,ZDANCEWIC2003,POTTIER2002}.
Some of these approaches include associating a security label with types or channels~\cite{HondaYoshidaTOPLAS2007}, associating a security label with actions~\cite{HondaYoshidaPOPL2002}, associating read and write policies with channels~\cite{HENNESSYRIELY2002,HENNESSY20053}, and associating a security label with processes and capabilities with expressions~\cite{CrafaARTICLE2002}. 
% Similarly, several definitions of noninterference have been proposed, e.g., barbed-congruence, P-congruence, may-testing and must-testing, per-models, and trace equivalence.
% 
Our approach differs from previous work in having a dynamic running secrecy that makes our system flow sensitive, using session types instead of process calculi, and the design of our novel logical relations for establishing noninterference.

% Of interest is Kobayashi's information flow analysis \cite{KobayashiARTICLE2005},
% which also proves progress-sensitive noninterference,
% where message transmission can be unsuccessful due to deadlocks.
% To guarantee successful transmission, the author annotates channel types with usage information \cite{KobayashiARTICLE2022}.

\subparagraph*{Logical relations for stateful languages.}

Kripke logical relations have been used to reason about stateful programs~\cite{PittsStarkHOOTS1998}.
The relation is indexed by a possible world that serves as a semantic model for the heap.
It establishes an invariant on the heap and ensures that the invariant is preserved for all future worlds.
When combined with step indexing~\cite{AppelMcAllesterTOPLAS2001,AhmedESOP2006}, Kripke logical relations can address circularity that arises in higher-order stores~\cite{AhmedPOPL2009,DreyerICFP2010,DreyerJFP2012}.
Our logical relation is similar to Kripke logical relations in being developed for a stateful language; names, like locations, are subject to concurrent mutation.
However, our session types are rooted in linear logic and thus internalize Kripke's logical worlds into the type system.

% Our logical relations, like KLRs, is situated in a stateful setting because names, like locations,
% are subject to concurrent mutation.
% However, our logical relation is rooted in linear logic, which guarantees race freedom, stratifies the store (the configuration of process), and prescribe state transitions.

% \cite{OddershedePOPL2021} recently contributed a logical relation for a higher-order functional language with higher-order store, recursive types, existential types, impredicative type polymorphism, and label polymorphism.
% They use a semantic approach to proving noninterference,  allowing integration of syntactically well-typed and ill-typed components, if the latter are shown to be semantically sound.
% Besides the difference in language, the authors consider termination-\emph{insensitive} noninterference,
% whereas we consider progress-\emph{sensitive} noninterference.

\section{Conclusions}
\label{s:concl}

We have presented a new session type system with information flow control (IFC) for an asynchronous $\pi$-calculus, and a notion of noninterference by means of a logical relation between typed processes.
Our development flexibly supports realistic cyclic process networks that may deadlock.
As such, our main result is that IFC well typedness implies deadlock-sensitive noninterference (DSNI).

In future work, we plan to study the interplay between IFC / DSNI and several interesting features of message-passing concurrency, such as recursion and non-determinism.
As commented on in the previous section, we expect the notion of an observation index~\cite{report/BalzerDHY23} to be applicable to circular process networks with recursive processes, although the interplay between potential leakage through non-termination and deadlocks will need careful consideration.
Support of non-determinism may be more challenging, especially when combining recursion with non-deterministic choice, which without further restriction permits loop guards at mixed confidentiality level.  We will consider focusing on more curtailed but logically motivated notions of non-determinism, such as coexponentials~\cite{QianICFP2021}, $\mathrm{HCP}^-_{\mathrm{ND}}$~\cite{KokkeARTICLE2020}, and linear non-determinism~\cite{VanDenHeuvelAPLAS2023}.
We are also interested in exploring ``name-sensitive'' noninterference: the choice of names in outputs is a possible source of information leakage outside the scope of this paper.

\bibliography{ref.bib}

\appendices

\section{Appendix}
\label{s:appendix}

\begin{definition}[Identity Expanded Forwarders]
    \label{d:fwd}
    \emph{Forwarders} are defined by the following typing rule:
    \[
        \begin{bussproof}[typ-fwd]
            \bussAssume{
                \fIFC{\Omega \Vdash d \lleq c}
            }
            \bussUn{
                \fType{\fIFC{\Omega} \vdash \fIFC{\fProc{\pFwd x y} \at d} :: \fProc{x}:A\fIFC{[c]}, \fProc{y}:\dual{A}\fIFC{[c]}}
            }
        \end{bussproof}
    \]
    This rule expands into the usual rules from \Cref{f:typeSys} by induction on the structure of $\fType{A}$ (we omit dual cases):
    \begin{align*}
        \fType{A} &= \fType{1}
        &
        \fType{\fIFC{\Omega}} &\fType{{}\vdash \fIFC{\fProc{\pWait y() ; \pClose x[]} \at d} :: \fProc{x}:1\fIFC{[c]} , \fProc{y}:\bot\fIFC{[c]}}
        \\
        \fType{A} &= \fType{B \tensor C}
        &
        \fType{\fIFC{\Omega}} &\fType{
            \begin{array}[t]{@{}l@{}}
                {}\vdash \fIFC{\fProc{\pRecv y(v,y_1) ; \nu{uw} \nu{zx_1} ( \pSend x[u,z] \| \pFwd v w \| \pFwd {x_1} {y_1} )} \at d}
                \\
                {}:: \fProc{x}:B \tensor C\fIFC{[c]} , \fProc{y}:\dual{B} \parr \dual{C}\fIFC{[c]}
            \end{array}
        }
        \\
        \fType{A} &= \fType{\oplus \{ i : A_i \}_{i \in I}}
        &
        \fType{\fIFC{\Omega}} &\fType{
            \begin{array}[t]{@{}l@{}}
                {}\vdash \fIFC{\fProc{\pBra y(y_1)>\{ i : \nu{zx_1} ( \pSel x[z]<i \| \pFwd {x_1} {y_1} ) \}_{i \in I}} \at d}
                \\
                {}:: \fProc{x}:\oplus \{ i : A_i \}_{i \in I}\fIFC{[c]} , \fProc{y}:\with \{ i : \dual{A_i} \}_{i \in I}\fIFC{[c]}
            \end{array}
        }
    \end{align*}
\end{definition}
%Derivability follows by induction on the structure of $A$, by expanding the forwarder into actual communications and forwarders on ``smaller'' types.
%In anticipation of our IFC, in each case we rely on the fact that $\fIFC{\Omega \Vdash d \lleq c} \implies \fIFC{\Omega \Vdash d \lcup c \lleq c}$.
%We give half of the cases for $\fType{A}$, the other half are dual and are analogous by swapping $\fProc{x}$ and $\fProc{y}$.
%We omit typing derivations.

\lSubstitution*

\begin{proof}
    By induction on the derivation of the assumption.
    Since $\fProc{y}$ is free in $\fProc{P}$, $\fProc{y}$ could only have been introduced and not been altered/used thereafter.
    Hence, we can inductively substitute $\fProc{x}$ for $\fProc{y}$ everywhere in this derivation without affecting anything else.
    The thesis follows.
\end{proof}

\tSubjCong*

\begin{proof}
    \label{proof:t:subjCong}
    By induction on the derivation of $\fProc{P} \sc \fProc{Q}$.
    The base cases correspond to the seven rules in \Cref{f:procSCRedd} (top), and we will detail each case below.
    The inductive cases correspond to closure under arbitrary process contexts in \Cref{d:procSyntax}; these cases follow from the IH straightforwardly.

    In each base case, we need to show both directions, but the reasoning usually works in both directions trivially.
    We apply inversion on the typing of $\fProc{P}$ to derive the typing of $\fProc{Q}$.

    \begin{itemize}
        \item
            Rule~\ruleLabel{sc-alpha}: $\fProc{P} \alpheq \fProc{Q} \implies \fProc{P} \sc \fProc{Q}$.
            Assume $\fProc{P}$ and $\fProc{Q}$ are $\alpha$-equivalent, i.e., they differ in bound names only.
            It is easy to see that the typing derivation of $\fProc{P}$ can be applied exactly to $\fProc{Q}$ by adopting the changes in bound names.

        \item
            Rule~\ruleLabel{sc-par-nil}: $\fProc{P \| 0} \sc \fProc{P}$.
            \begin{mathpar}
                \begin{bussproof}
                    \bussAssume{
                        \fIFC{\Omega \Vdash d \lleq d \lcap d}
                    }
                    \bussAssume{
                        \fType{\fIFC{\Omega} \vdash \fIFC{\fProc{P} \at d} :: \Gamma}
                    }
                    \bussAx[\ruleLabel{typ-inact}]{
                        \fType{\fIFC{\Omega} \vdash \fIFC{\fProc{0} \at d} :: \emptyset}
                    }
                    \bussTern[\ruleLabel{typ-par}]{
                        \fType{\fIFC{\Omega} \vdash \fIFC{\fProc{P \| 0} \at d} :: \Gamma}
                    }
                \end{bussproof}
                \and
                \Leftrightarrow
                \and
                \fType{\fIFC{\Omega} \vdash \fIFC{\fProc{P} \at d} :: \Gamma}
            \end{mathpar}
            The requirement that we need to prove is $\fIFC{\Omega \Vdash d \lleq d \lcap d}$; this holds trivially.

        \item
            Rule~\ruleLabel{sc-par-symm}: $\fProc{P \| Q} \sc \fProc{Q \| P}$.
            \begin{mathpar}
                \begin{bussproof}
                    \bussAssume{
                        \fIFC{\Omega \Vdash d \lleq d'_1 \lcap d'_2}
                    }
                    \bussAssume{
                        \fType{\fIFC{\Omega} \vdash \fIFC{\fProc{P} \at d'_1} :: \Gamma}
                    }
                    \bussAssume{
                        \fType{\fIFC{\Omega} \vdash \fIFC{\fProc{Q} \at d'_2} :: \Delta}
                    }
                    \bussTern[\ruleLabel{typ-par}]{
                        \fType{\fIFC{\Omega} \vdash \fIFC{\fProc{P \| Q} \at d} :: \Gamma , \Delta}
                    }
                \end{bussproof}
                \and
                \Leftrightarrow
                \and
                \begin{bussproof}
                    \bussAssume{
                        \fIFC{\Omega \Vdash d \lleq d'_2 \lcap d'_1}
                    }
                    \bussAssume{
                        \fType{\fIFC{\Omega} \vdash \fIFC{\fProc{Q} \at d'_2} :: \Delta}
                    }
                    \bussAssume{
                        \fType{\fIFC{\Omega} \vdash \fIFC{\fProc{P} \at d'_1} :: \Gamma}
                    }
                    \bussTern[\ruleLabel{typ-par}]{
                        \fType{\fIFC{\Omega} \vdash \fIFC{\fProc{P \| Q} \at d} :: \Gamma , \Delta}
                    }
                \end{bussproof}
            \end{mathpar}
            Clearly, $\fIFC{\Omega \Vdash d \lleq d'_1 \lcap d'_2}$ if and only if $\fIFC{\Omega \Vdash d \lleq d'_2 \lcap d'_1}$.

        \item
            Rule~\ruleLabel{sc-par-assoc}: $\fProc{(P \| Q) \| R} \sc \fProc{P \| (Q \| R)}$.
            Given
            \begin{align}
                & \fIFC{\Omega \Vdash d \lleq d'_1 \lcap d'_2},
                \label{eq:app:subjCong:parAssoc1}
                \\
                & \fIFC{\Omega \Vdash d'_1 \lleq d''_1 \lcap d''_2},
                \label{eq:app:subjCong:parAssoc2}
                \\
                & \fIFC{\Omega \Vdash d \lleq d''_1 \lcap d'_3},
                \label{eq:app:subjCong:parAssoc3}
            \end{align}
            we have
            \begin{mathpar}
                \begin{bussproof}
                    \bussAssume{
                        \eqref{eq:app:subjCong:parAssoc1}
                    }
                    \bussAssume{
                        \eqref{eq:app:subjCong:parAssoc2}    
                    }
                    \bussAssume{
                        \fType{\fIFC{\Omega} \vdash \fIFC{\fProc{P} \at d''_1} :: \Gamma}
                    }
                    \bussAssume{
                        \fType{\fIFC{\Omega} \vdash \fIFC{\fProc{Q} \at d''_2} :: \Delta}
                    }
                    \bussTern[\ruleLabel{typ-par}]{
                        \fType{\fIFC{\Omega} \vdash \fIFC{\fProc{P \| Q} \at d'_1} :: \Gamma , \Delta}
                    }
                    \bussAssume{
                        \fType{\fIFC{\Omega} \vdash \fIFC{\fProc{R} \at d'_2} :: \Lambda}
                    }
                    \bussTern[\ruleLabel{typ-par}]{
                        \fType{\fIFC{\Omega} \vdash \fIFC{\fProc{(P \| Q) \| R} \at d} :: \Gamma , \Delta , \Lambda}
                    }
                \end{bussproof}
                \and
                \Leftrightarrow
                \and
                \begin{bussproof}
                    \bussAssume{
                        \eqref{eq:app:subjCong:parAssoc3}
                    }
                    \bussAssume{
                        \fType{\fIFC{\Omega} \vdash \fIFC{\fProc{P} \at d''_1} :: \Gamma}
                    }
                    \bussAssume{
                        \fIFC{\Omega \Vdash d'_3 \lleq d''_2 \lcap d'_2}
                    }
                    \bussAssume{
                        \fType{\fIFC{\Omega} \vdash \fIFC{\fProc{Q} \at d''_2} :: \Delta}
                    }
                    \bussAssume{
                        \fType{\fIFC{\Omega} \vdash \fIFC{\fProc{R} \at d'_2} :: \Lambda}
                    }
                    \bussTern[\ruleLabel{typ-par}]{
                        \fType{\fIFC{\Omega} \vdash \fIFC{\fProc{Q \| R} \at d'_3} :: \Delta , \Lambda}
                    }
                    \bussTern[\ruleLabel{typ-par}]{
                        \fType{\fIFC{\Omega} \vdash \fIFC{\fProc{P \| (Q \| R)} \at d} :: \Gamma , \Delta , \Lambda}
                    }
                \end{bussproof}
            \end{mathpar}
            We need to prove that the security conditions hold, in both directions.
            \begin{itemize}
                \item[($\Rightarrow$)]
                    Assume there exists $\fIFC{d'_1}$ such that $\fIFC{\Omega \Vdash d \lleq d'_1 \lcap d'_2}$ and $\fIFC{\Omega \Vdash d'_1 \lleq d''_1 \lcap d''_2}$.
                    We need to prove that there exists $\fIFC{d'_3}$ such that (i)~$\fIFC{\Omega \Vdash d \lleq d''_1 \lcap d'_3}$ and (ii)~$\fIFC{\Omega \Vdash d'_3 \lleq d''_2 \lcap d'_2}$.
                    By assumption and definition, $\fIFC{\Omega \Vdash d \lleq d'_1,d'_2}$, so $\fIFC{\Omega \Vdash d \lleq d''_1,d''_2}$.
                    Let $\fIFC{d'_3} = \fIFC{d}$.
                    Then $\fIFC{\Omega \Vdash d \lleq d'_3}$, so $\fIFC{\Omega \Vdash d \lleq d''_1 \lcap d'_3}$, proving~(i).
                    Moreover, $\fIFC{\Omega \Vdash d'_3} = \fIFC{d \lleq d''_2 \lcap d'_2}$, proving~(ii).

                \item[($\Leftarrow$)]
                    Assume there exists $\fIFC{d'_3}$ such that $\fIFC{\Omega \Vdash d \lleq d''_1 \lcap d'_3}$ and $\fIFC{\Omega \Vdash d'_3 \lleq d''_2 \lcap d'_2}$.
                    We need to prove that there exists $\fIFC{d'_1}$ such that (i)~$\fIFC{\Omega \Vdash d \lleq d'_1 \lcap d'_2}$ and (ii)~$\fIFC{\Omega \Vdash d'_1 \lleq d''_1 \lcap d''_2}$.
                    By assumption and definition, $\fIFC{\Omega \Vdash d \lleq d''_1,d'_3}$, so $\fIFC{\Omega \Vdash d \lleq d''_2,d'_2}$.
                    Let $\fIFC{d'_1} = \fIFC{d}$.
                    Then $\fIFC{\Omega \Vdash d \lleq d'_1}$, so $\fIFC{\Omega \Vdash d \lleq d'_1 \lcap d'_2}$, proving~(i).
                    Moreover, $\fIFC{\Omega \Vdash d'_1} = \fIFC{d \lleq d''_1 \lcap d''_2}$, proving~(ii).
            \end{itemize}

        \item
            Rule~\ruleLabel{sc-res-symm}: $\fProc{\nu{xy} P} \sc \fProc{\nu{yx} P}$.
            This case holds because the type system implicitly permits permutation of the typing contexts.
            \begin{mathpar}
                \begin{bussproof}
                    \bussAssume{
                        \fType{\fIFC{\Omega} \vdash \fIFC{\fProc{P} \at d} :: \Gamma , \fProc{x}:A\fProc{[c]} , \fProc{y}:\dual{A}\fIFC{[c]}}
                    }
                    \bussUn[\ruleLabel{typ-res}]{
                        \fType{\fIFC{\Omega} \vdash \fIFC{\fProc{\nu{xy} P} \at d} :: \Gamma}
                    }
                \end{bussproof}
                \and
                \Leftrightarrow
                \and
                \begin{bussproof}
                    \bussAssume{
                        \fType{\fIFC{\Omega} \vdash \fIFC{\fProc{P} \at d} :: \Gamma , \fProc{x}:A\fIFC{[c]} , \fProc{y}:\dual{A}\fIFC{[c]}}
                    }
                    \bussUn[\ruleLabel{typ-res}]{
                        \fType{\fIFC{\Omega} \vdash \fIFC{\fProc{\nu{yx} P} \at d} :: \Gamma}
                    }
                \end{bussproof}
            \end{mathpar}

        \item
            Rule~\ruleLabel{sc-res-assoc}: $\fProc{\nu{xy} \nu{zw} P} \sc \fProc{\nu{zw} \nu{xy} P}$.
            \begin{mathpar}
                \begin{bussproof}
                    \bussAssume{
                        \fType{\fIFC{\Omega} \vdash \fIFC{\fProc{P} \at d} :: \Gamma , \fProc{x}:A\fIFC{[c]} , \fProc{y}:\dual{A}\fIFC{[c]} , \fProc{z}:B\fIFC{[c]} , \fProc{w}:\dual{B}\fIFC{[c]}}
                    }
                    \bussUn[\ruleLabel{typ-res}]{
                        \fType{\fIFC{\Omega} \vdash \fIFC{\fProc{\nu{zw} P} \at d} :: \Gamma , \fProc{x}:A\fIFC{[c]} , \fProc{y}:\dual{A}\fIFC{[c]}}
                    }
                    \bussUn[\ruleLabel{typ-res}]{
                        \fType{\fIFC{\Omega} \vdash \fIFC{\fProc{\nu{xy} \nu{zw} P} \at d} :: \Gamma}
                    }
                \end{bussproof}
                \and
                \Leftrightarrow
                \and
                \begin{bussproof}
                    \bussAssume{
                        \fType{\fIFC{\Omega} \vdash \fIFC{\fProc{P} \at d} :: \Gamma , \fProc{x}:A\fIFC{[c]} , \fProc{y}:\dual{A}\fIFC{[c]} , \fProc{z}:B\fIFC{[c]} , \fProc{w}:\dual{B}\fIFC{[c]}}
                    }
                    \bussUn[\ruleLabel{typ-res}]{
                        \fType{\fIFC{\Omega} \vdash \fIFC{\fProc{\nu{xy} P} \at d} :: \Gamma , \fProc{z}:B\fIFC{[c]} , \fProc{w}:\dual{B}\fIFC{[c]}}
                    }
                    \bussUn[\ruleLabel{typ-res}]{
                        \fType{\fIFC{\Omega} \vdash \fIFC{\fProc{\nu{zw} \nu{xy} P} \at d} :: \Gamma}
                    }
                \end{bussproof}
            \end{mathpar}

        \item
            Rule~\ruleLabel{sc-res-comm}: $\fProc{x},\fProc{y} \notin \fn(\fProc{Q}) \implies \fProc{\nu{xy} (P \| Q)} \sc \fProc{\nu{xy} P \| Q}$.
            Assuming $\fProc{x},\fProc{y} \notin \fn(\fProc{Q})$, by \Cref{l:notFnImpliesNotInDom}, $\fProc{x},\fProc{y} \notin \dom(\fProc{\Delta})$.
            \begin{mathpar}
                \begin{bussproof}
                    \bussAssume{
                        \fIFC{\Omega \Vdash d \lleq d'_1 \lcap d'_2}
                    }
                    \bussAssume{
                        \fType{\fIFC{\Omega} \vdash \fIFC{\fProc{P} \at d'_1} :: \Gamma , \fProc{x}:A\fIFC{[c]} , \fProc{y}:\dual{A}\fIFC{[c]}}
                    }
                    \bussAssume{
                        \fType{\fIFC{\Omega} \vdash \fIFC{\fProc{Q} \at d'_2} :: \Delta}
                    }
                    \bussTern[\ruleLabel{typ-par}]{
                        \fType{\fIFC{\Omega} \vdash \fIFC{\fProc{P \| Q} \at d} :: \Gamma , \Delta , \fProc{x}:A\fIFC{[c]} , \fProc{y}:\dual{A}\fIFC{[c]}}
                    }
                    \bussUn[\ruleLabel{typ-res}]{
                        \fType{\fIFC{\Omega} \vdash \fIFC{\fProc{\nu{xy} (P \| Q)} \at d} :: \Gamma , \Delta}
                    }
                \end{bussproof}
                \and
                \Leftrightarrow
                \and
                \begin{bussproof}
                    \bussAssume{
                        \fIFC{\Omega \Vdash d \lleq d'_1 \lcap d'_2}
                    }
                    \bussAssume{
                        \fType{\fIFC{\Omega} \vdash \fIFC{\fProc{P} \at d'_1} :: \Gamma , \fProc{x}:A\fIFC{[c]} , \fProc{y}:\dual{A}\fIFC{[c]}}
                    }
                    \bussUn[\ruleLabel{typ-res}]{
                        \fType{\fIFC{\Omega} \vdash \fIFC{\fProc{\nu{xy} P} \at d'_1} :: \Gamma}
                    }
                    \bussAssume{
                        \fType{\fIFC{\Omega} \vdash \fIFC{\fProc{Q} \at d'_2} :: \Delta}
                    }
                    \bussTern[\ruleLabel{typ-par}]{
                        \fType{\fIFC{\Omega} \vdash \fIFC{\fProc{\nu{xy} P \| Q} \at d} :: \Gamma , \Delta}
                    }
                \end{bussproof}
            \end{mathpar}
            \qedhere
    \end{itemize}
\end{proof}

\tSubjRed*

\begin{proof}
    \label{proof:t:subjRed}
    By induction on the derivation of $\fProc{P} \redd \fProc{Q}$.
    The cases correspond to the reduction rules in \Cref{f:procSCRedd} (bottom).
    In each case, we apply inversion on the typing of $\fProc{P}$ to derive the typing of $\fProc{Q}$.
    \begin{itemize}

        \item
            Rule~\ruleLabel{red-close-wait}: $\fProc{\nu{xy} ( \pClose x[] \| \pWait y() ; P )} \redd \fProc{P}$.
            Given
            \begin{align}
                & \fIFC{\Omega \Vdash d \lleq d'_1 \lcap d'_2},
                \label{eq:app:subjRed:CloseWait1}
            \end{align}
            we have
            \begin{mathpar}
                \begin{bussproof}
                    \bussAssume{
                        \eqref{eq:app:subjRed:CloseWait1}
                    }
                    \bussAssume{
                        \fIFC{\Omega \Vdash d'_1 \lleq c}
                    }
                    \bussUn[\ruleLabel{typ-close}]{
                        \fType{\fIFC{\Omega} \vdash \fIFC{\fProc{\pClose x[]} \at d'_1} :: \fProc{x}:1\fIFC{[c]}}
                    }
                    \bussAssume{
                        \fIFC{\Omega \Vdash d''_2 = d'_2 \lcup c}
                    }
                    \bussAssume{
                        \fType{\fIFC{\Omega} \vdash \fIFC{\fProc{P} \at d''_2} :: \Gamma}
                    }
                    \bussBin[\ruleLabel{typ-wait}]{
                        \fType{\fIFC{\Omega} \vdash \fIFC{\fProc{\pWait y() ; P} \at d'_2} :: \Gamma , \fProc{y}:\bot\fIFC{[c]}}
                    }
                    \bussTern[\ruleLabel{typ-par}]{
                        \fType{\fIFC{\Omega} \vdash \fIFC{\fProc{\pClose x[] \| \pWait y() ; P} \at d} :: \Gamma , \fProc{x}:1\fIFC{[c]} , \fProc{y}:\bot\fIFC{[c]}}
                    }
                    \bussUn[\ruleLabel{typ-res}]{
                        \fType{\fIFC{\Omega} \vdash \fIFC{\fProc{\nu{xy} ( \pClose x[] \| \pWait y() ; P )} \at d} :: \Gamma}
                    }
                \end{bussproof}
                \and
                \Rightarrow
                \and
                \fType{\fIFC{\Omega} \vdash \fIFC{\fProc{P} \at d''_2} :: \Gamma}
            \end{mathpar}
            By assumption and by definition, $\fIFC{\Omega \Vdash d''_2 \lgeq d'_2}$.
            Also, by definition, $\fIFC{\Omega \Vdash d'_2 \lgeq d'_1 \lcap d'_2}$, so, by assumption, $\fIFC{\Omega \Vdash d'_2 \lgeq d}$.
            Hence, $\fIFC{\Omega \Vdash d''_2 \lgeq d}$.

        \item
            Rule~\ruleLabel{res-sel-bra}: $j \in I \implies \nu{xy} (\pSel x[b]<j \| \pBra y(w)>\{i:P_i\}_{i \in I}) \redd P_j \pSubst{b/w}$.
            Given
            \begin{align}
                & \fIFC{\Omega \Vdash d \lleq d'_1 \lcap d'_2},
                \label{eq:app:subjRed:selBra1}
                \\
                & \begin{bussproof}
                    \bussAssume{
                        \fIFC{\Omega \Vdash d'_1 \lleq c}
                    }
                    \bussAssume{
                        \fProc{j} \in \fProc{I}
                    }
                    \bussBin[\ruleLabel{typ-sel}]{
                        \fType{\fIFC{\Omega} \vdash \fIFC{\fProc{\pSel x[b]<j} \at d'_1} :: \fProc{x}:\oplus \{ i : A_i\}_{i \in I}\fIFC{[c]} , \fProc{b}:\dual{A_j}\fIFC{[c]}}
                    }
                \end{bussproof},
                \label{eq:app:subjRed:selBra2}
                \\
                & \fIFC{\Omega \Vdash d''_2 = d'_2 \lcup c}
                \label{eq:app:subjRed:selBra3}
            \end{align}
            \begin{mathpar}
                \begin{bussproof}
                    \bussAssume{
                        \eqref{eq:app:subjRed:selBra1}
                    }
                    \bussAssume{
                        \eqref{eq:app:subjRed:selBra2}
                    }
                    \bussAssume{
                        \eqref{eq:app:subjRed:selBra3}
                    }
                    \bussAssume{
                        \forall \fProc{i} \in \fProc{I}.~ \fType{\fIFC{\Omega} \vdash \fIFC{\fProc{P_i} \at d''_2} :: \Gamma , \fProc{z}:\dual{A_i}\fIFC{[c]}}
                    }
                    \bussBin[\ruleLabel{typ-bra}]{
                        \fType{\fIFC{\Omega} \vdash \fIFC{\fProc{\pBra y(z)>\{i:P_i\}_{i \in I}} \at d'_2} :: \Gamma , \fProc{y}:\with \{ i : \dual{A_i} \}_{i \in I}\fIFC{[c]}}
                    }
                    \bussTern[\ruleLabel{typ-par}]{
                        \fType{
                            \fIFC{\Omega} 
                            \begin{array}[t]{@{}l@{}}
                                {} \vdash \fIFC{\fProc{\pSel x[b]<j \| \pBra y(z)>\{i:P_i\}_{i \in I}} \at d} 
                                \\
                                {} :: \Gamma , \fProc{x}:\oplus \{ i : A_i\}_{i \in I}\fIFC{[c]} , \fProc{y}:\with \{ i : \dual{A_i} \}_{i \in I}\fIFC{[c]} , \fProc{b}:\dual{A_j}\fIFC{[c]}
                            \end{array}
                        }
                    }
                    \bussUn[\ruleLabel{typ-res}]{
                        \fType{\fIFC{\Omega} \vdash \fIFC{\fProc{\nu{xy} ( \pSel x[b]<j \| \pBra y(z)>\{i:P_i\}_{i \in I} )} \at d} :: \Gamma , \fProc{b}:\dual{A_j}\fIFC{[c]}}
                    }
                \end{bussproof}
                \and
                \Rightarrow
                \and
                \begin{bussproof}
                    \bussAssume{
                        \text{\Cref{l:substitution}}
                    }
                    \bussUn{
                        \fType{\fIFC{\Omega} \vdash \fIFC{\fProc{P_j \pSubst{b/z}} \at d''_2} :: \Gamma , \fProc{b}:\dual{A_j}\fIFC{[c]}}
                    }
                \end{bussproof}
            \end{mathpar}
            Following the same reasoning as in the previous two cases, $\fIFC{\Omega \Vdash d''_2 \lgeq d}$.
        \item
            Rule~\ruleLabel{red-send-recv}: $\fProc{\nu{xy} ( \pSend x[a,b] \| \pRecv y(z,w) ; P )} \redd \fProc{P \pSubst{a/z,b/w}}$.
            Given
            \begin{align}
                & \fIFC{\Omega \Vdash d \lleq d'_1 \lcap d'_2},
                \label{eq:app:subjRed:sendRecv1}
                \\
                & \begin{bussproof}
                    \bussAssume{
                        \fIFC{\Omega \Vdash d'_1 \lleq c}
                    }
                    \bussUn[\ruleLabel{typ-send}]{
                        \fType{\fIFC{\Omega} \vdash \fIFC{\fProc{\pSend x[a,b]} \at d'_1} :: \fProc{x}:A \tensor B\fIFC{[c]} , \fProc{a}:\dual{A}\fIFC{[c]} , \fProc{b}:\dual{B}\fIFC{[c]}}
                    }
                \end{bussproof},
                \label{eq:app:subjRed:sendRecv2}
                \\
                & \fIFC{\Omega \Vdash d''_2 = d'_2 \lcup c},
                \label{eq:app:subjRed:sendRecv3}
            \end{align}
            we have
            \begin{mathpar}
                \begin{bussproof}
                    \bussAssume{
                        \eqref{eq:app:subjRed:sendRecv1}
                    }
                    \bussAssume{
                        \eqref{eq:app:subjRed:sendRecv2}
                    }
                    \bussAssume{
                        \eqref{eq:app:subjRed:sendRecv3}
                    }
                    \bussAssume{
                        \fType{\fIFC{\Omega} \vdash \fIFC{\fProc{P} \at d''_2} :: \Gamma , \fProc{z}:\dual{A}\fIFC{[c]} , \fProc{w}:\dual{B}\fIFC{[c]}}
                    }
                    \bussBin[\ruleLabel{typ-recv}]{
                        \fType{\fIFC{\Omega} \vdash \fIFC{\fProc{\pRecv y(z,w) ; P} \at d'_2} :: \Gamma , \fProc{y}:\dual{A} \parr \dual{B}\fIFC{[c]}}
                    }
                    \bussTern[\ruleLabel{typ-par}]{
                        \fType{\fIFC{\Omega} \vdash \fIFC{\fProc{\pSend x[a,b] \| \pRecv y(z,w) ; P} \at d} :: \Gamma , \fProc{x}:A \tensor B\fIFC{[c]} , \fProc{y}:\dual{A} \parr \dual{B}\fIFC{[c]} , \fProc{a}:\dual{A}\fIFC{[c]} , \fProc{b}:\dual{B}\fIFC{[c]}}
                    }
                    \bussUn[\ruleLabel{typ-res}]{
                        \fType{\fIFC{\Omega} \vdash \fIFC{\fProc{\nu{xy} ( \pSend x[a,b] \| \pRecv y(z,w) ; P )} \at d} :: \Gamma , \fProc{a}:\dual{A}\fIFC{[c]} , \fProc{b}:\dual{B}\fIFC{[c]}}
                    }
                \end{bussproof}
                \and
                \Rightarrow
                \and
                \begin{bussproof}
                    \bussAssume{
                        \text{\Cref{l:substitution} twice}
                    }
                    \bussUn{
                        \fType{\fIFC{\Omega} \vdash \fIFC{\fProc{P \pSubst{a/z,b/w}} \at d''_2} :: \Gamma , \fProc{a}:\dual{A}\fIFC{[c]} , \fProc{b}:\dual{B}\fIFC{[c]}}
                    }
                \end{bussproof}
            \end{mathpar}
            Following the same reasoning as in the previous case, $\fIFC{\Omega \Vdash d''_2 \lgeq d}$.

        \item
            Rule~\ruleLabel{res-sc}: $\fProc{P} \sc \fProc{P'} \wedge \fProc{P'} \redd \fProc{Q'} \wedge \fProc{Q'} \sc \fProc{Q} \implies \fProc{P} \redd \fProc{Q}$.
            Since $\fProc{P} \sc \fProc{P'}$, by \Cref{t:subjCong}, $\fType{\fIFC{\Omega} \vdash \fIFC{\fProc{P'} \at d} :: \Gamma}$.
            Then, by the IH, $\fType{\fIFC{\Omega} \vdash \fIFC{\fProc{Q'} \at d'} :: \Gamma}$ for some $\fIFC{d'}$ such that $\fIFC{\Omega \Vdash d \lleq d'}$.
            Hence, since $\fProc{Q'} \sc \fProc{Q}$, by \Cref{t:subjCong}, $\fType{\fIFC{\Omega} \vdash \fIFC{\fProc{Q} \at d'} :: \Gamma}$.

        \item
            Rule~\ruleLabel{red-par}: $\fProc{P} \redd \fProc{P'} \implies \fProc{P \| Q} \redd \fProc{P' \| Q}$.
            \[
                \begin{bussproof}
                    \bussAssume{
                        \fIFC{\Omega \Vdash d \lgeq d'_1 \lcap d'_2}
                    }
                    \bussAssume{
                        \fType{\fIFC{\Omega} \vdash \fIFC{\fProc{P} \at d'_1} :: \Gamma}
                    }
                    \bussAssume{
                        \fType{\fIFC{\Omega} \vdash \fIFC{\fProc{Q} \at d'_2} :: \Delta}
                    }
                    \bussTern[\ruleLabel{typ-par}]{
                        \fType{\fIFC{\Omega} \vdash \fIFC{\fProc{P \| Q} \at d} :: \Gamma , \Delta}
                    }
                \end{bussproof}
            \]
            Since $\fProc{P} \redd \fProc{P'}$, by the IH, $\fType{\fIFC{\Omega} \vdash \fIFC{\fProc{P'} \at d''_1} :: \Gamma}$ for some $\fIFC{d''_1}$ such that $\fIFC{\Omega \Vdash d'_1 \lleq d''_1}$.
            Let $\fIFC{d'} = \fIFC{(d''_1 \lcap d'_2) \lcup d}$.
            Then $\fIFC{\Omega \Vdash d \lleq d'}$.
            \begin{mathpar}
                \Rightarrow
                \and
                \begin{bussproof}
                    \bussAssume{
                        \fIFC{\Omega \Vdash d' \lgeq d''_1 \lcap d'_2}
                    }
                    \bussAssume{
                        \fType{\fIFC{\Omega} \vdash \fIFC{\fProc{P'} \at d''_1} :: \Gamma}
                    }
                    \bussAssume{
                        \fType{\fIFC{\Omega} \vdash \fIFC{\fProc{Q} \at d'_2} :: \Delta}
                    }
                    \bussTern[\ruleLabel{typ-par}]{
                        \fType{\fIFC{\Omega} \vdash \fIFC{\fProc{P' \| Q} \at d'} :: \Gamma , \Delta}
                    }
                \end{bussproof}
            \end{mathpar}

        \item
            Rule~\ruleLabel{red-res}: $\fProc{P} \redd \fProc{P'} \implies \fProc{\nu{xy} P} \redd \fProc{\nu{xy} P'}$.
            \[
                \begin{bussproof}
                    \bussAssume{
                        \fType{\fIFC{\Omega} \vdash \fIFC{\fProc{P} \at d} :: \Gamma , \fProc{x}:A\fIFC{[c]} , \fProc{y}:\dual{A}\fIFC{[c]}}
                    }
                    \bussUn[\ruleLabel{typ-res}]{
                        \fType{\fIFC{\Omega} \vdash \fIFC{\fProc{\nu{xy} P} \at d} :: \Gamma}
                    }
                \end{bussproof}
            \]
            Since $\fProc{P} \redd \fProc{P'}$, by the IH, $\fType{\fIFC{\Omega} \vdash \fIFC{\fProc{P'} \at d'} :: \Gamma , \fProc{x}:A\fIFC{[c]} , \fProc{y}:\dual{A}\fIFC{[c]}}$ for some $\fIFC{d'}$ such that $\fIFC{\Omega \Vdash d \lleq d'}$.
            \[
                \begin{bussproof}
                    \bussAssume{
                        \fType{\fIFC{\Omega} \vdash \fIFC{\fProc{P'} \at d'} :: \Gamma , \fProc{x}:A\fIFC{[c]} , \fProc{y}:\dual{A}\fIFC{[c]}}
                    }
                    \bussUn[\ruleLabel{typ-res}]{
                        \fType{\fIFC{\Omega} \vdash \fIFC{\fProc{\nu{xy} P'} \at d'} :: \Gamma}
                    }
                \end{bussproof}
                \tag*{\qedhere}
            \]
    \end{itemize}
\end{proof}

\lRedCtx*

\begin{proof}
    \label{proof:l:redCtx}
    By induction on the derivation of $\fProc{P} \redd \fProc{P'}$.
    \begin{itemize}

        \item
            \textbf{(Rule~\ruleLabel{red-close-wait})}
            Then Item~\ref{i:redCtx:closeWait} holds with $\fCtx{E} = \fCtx{\hole}$.

        \item
            \textbf{(Rule~\ruleLabel{red-send-recv})}
            Then Item~\ref{i:redCtx:sendRecv} holds with $\fCtx{E} = \fCtx{\hole}$.

        \item
            \textbf{(Rule~\ruleLabel{red-sel-bra})}
            Then Item~\ref{i:redCtx:selBra} holds with $\fCtx{E} = \fCtx{\hole}$.

        \item
            \textbf{(Rule~\ruleLabel{red-sc})}
            We have $\fProc{P} \sc \fProc{Q} \redd \fProc{Q'} \equiv \fProc{P'}$.
            By the IH, $\fProc{Q} \sc \fCtx{E[\fProc{R}]}$ and $\fProc{Q'} \sc \fCtx{E[\fProc{R'}]}$, where $\fProc{R},\fProc{R'}$ are as in one of the shapes given by the thesis.
            The thesis then also holds for $\fProc{P},\fProc{P'}$ by definition of structural congruence.

        \item
            \textbf{(Rule~\ruleLabel{red-par})}
            We have $\fProc{P} = \fProc{Q \| R}$, $\fProc{Q} \redd \fProc{Q'}$, and $\fProc{P'} = \fProc{Q' \| R}$.
            By the IH, $\fProc{Q} \sc \fCtx{F[\fProc{S}]}$ and $\fProc{Q'} \sc \fCtx{F[\fProc{S'}]}$, where $\fProc{S},\fProc{S'}$ are as in one of the shapes given by the thesis.
            The thesis then holds for $\fProc{P},\fProc{P'}$ with $\fCtx{E} := \fCtx{F \| R}$.

        \item
            \textbf{(Rule~\ruleLabel{red-res})}
            We have $\fProc{P} = \fProc{\nu{xy} Q}$, $\fProc{Q} \redd \fProc{Q'}$, and $\fProc{P'} = \fProc{\nu{xy} Q'}$.
            By the IH, $\fProc{Q} \sc \fCtx{F[\fProc{S}]}$ and $\fProc{Q'} \sc \fCtx{F[\fProc{S'}]}$, where $\fProc{S},\fProc{S'}$ are as in one of the shapes given by the thesis.
            The thesis then holds for $\fProc{P},\fProc{P'}$ with $\fCtx{E} := \fCtx{\nu{xy} F}$.
            \qedhere

    \end{itemize}
\end{proof}

\luRedd*

\begin{proof}
    \label{proof:l:uRedd}
    Suppose, toward a contradiction, that none of these three items hold.
    Then the only possibility is that \Cref{l:redCtx} applies, on names of which at least one is in $\fType{\Gamma}$.
    W.l.o.g., assume we have $\fCtx{E[\fProc{P}]} \sc \fCtx{F[\fProc{\nu{xy} ( \pClose x[] \| \pWait y() ; Q )}]} \redd \fCtx{F[\fProc{Q}]} \sc \fCtx{E'[\fProc{P'}]}$.
    At least one of $\fProc{x},\fProc{y}$ appears in $\fType{\Gamma}$, so at least one of the close on $\fProc{x}$ and the wait on $\fProc{y}$ appears in $\fProc{P}$; w.l.o.g., assume the wait on $\fProc{y}$ appears in $\fProc{P}$ and the restriction and close on $\fProc{x}$ appear in $\fCtx{E}$.
    We then have that $\fCtx{E} \sc \fCtx{\nu{xy} ( \pClose x[] \| G )}$, $\fCtx{E'} = \fCtx{G}$, $\fProc{P} \sc \fCtx{H[\fProc{\pWait y() ; Q}]}$, and $\fProc{P'} \sc \fCtx{H[\fProc{Q}]}$.
    By the assumption, $\fType{\fIFC{\Omega} \vdash \fIFC{\fProc{P} \at d} :: \Gamma'}$ where $\fType{\Gamma} = \fType{\Gamma' \proj \fIFC{\xi}}$.
    However, $\fType{\fIFC{\Omega} \vdash \fIFC{\fProc{P'} \at d} :: \Gamma' \setminus \fProc{y}}$, and thus $\fType{(\Gamma' \setminus \fProc{y}) \proj \fIFC{\xi}} = \fType{\Gamma \setminus \fProc{y}}$.
    Clearly, then $(\fCtx{E'},\fProc{P'}) \notin \netw{\fIFC{\Omega};\fIFC{\xi}}(\fType{\Gamma})$.
    This violates the definition of $\uRedd{\fIFC{\Omega};\fIFC{\xi};\fType{\Gamma}}$ (\Cref{d:uRedd}): a contradiction.
\end{proof}

\lcatchUp*

\begin{proof}
    \label{proof:l:catchUp}
    By the assumption, we know the typing of all processes involved (cf.\ \Cref{d:projNetw}).
    Moreover, there are normal forms $\fProc{Q_1},\fProc{Q^\star_2}$ of $\fProc{P_1},\fProc{P_2}$ respectively, such that $\fProc{Q_1 \proj \fIFC{\xi}} \sc \fProc{Q^\star_2 \proj \fIFC{\xi}}$ (cf.\ \Cref{d:obsEq}).
    This means that $\fProc{Q_1 \proj \fIFC{\xi}}$ and $\fProc{Q^\star_2 \proj \fIFC{\xi}}$ are equivalent, up to the ordering of binders and nodes, and $\alpha$-renaming.
    For a smoother proof, we thus obtain $\fProc{Q_2}$ from $\fProc{Q^\star_2}$ by rearranging binders and nodes, and $\alpha$-renaming, such that $\fProc{Q^\star_2} \sc \fProc{Q_2}$ (i.e., $\fProc{Q_2}$ is still a normal form of $\fProc{P_2}$) and $\fProc{Q_1 \proj \fIFC{\xi}} = \fProc{Q_2 \proj \fIFC{\xi}}$.
    By \Cref{d:nodesNf}, it suffices to consider only the normal forms $\fProc{Q_1},\fProc{Q_2}$ instead of their counterparts $\fProc{P_1},\fProc{P_2}$: we show that $\fCtx{E_1},\fProc{Q_1} \uRedd{\fIFC{\Omega};\fIFC{\xi};\fType{\Gamma}} \fCtx{E'_1},\fProc{Q'_1}$ where $\fProc{Q'_1}$ in normal form implies $\fCtx{E_2},\fProc{Q_2} \uReddQ{\fIFC{\Omega};\fIFC{\xi};\fType{\Gamma}} \fCtx{E_2},\fProc{Q'_2}$ where $\fProc{Q'_2}$ in normal form and $\fProc{Q'_1 \proj \fIFC{\xi}} \sc \fProc{Q'_2 \proj \fIFC{\xi}}$, from which the thesis follows by definition.

    By \Cref{l:uRedd}, we can distinguish three cases from which $\fCtx{E_1},\fProc{Q_1} \uRedd{\fIFC{\Omega};\fIFC{\xi};\fType{\Gamma}} \fCtx{E'_1},\fProc{Q'_1}$ follows.
    \begin{itemize}

        \item
            \textbf{(Internal in context: $\fCtx{E_1} \redd \fCtx{E'_1}$ and $\fProc{Q_1} = \fProc{Q'_1}$)}
            Then the thesis holds directly with $\fProc{Q'_2} := \fProc{Q_2}$.

        \item
            \textbf{(Internal in process: $\fProc{Q_1} \redd \fProc{Q'_1}$ and $\fCtx{E_1} = \fCtx{E'_1}$)}
            By \Cref{l:redCtx}, $\fProc{Q_1}$'s reduction is due to one of three possible synchronizations inside some evaluation context.
            Note that \Cref{l:redCtx} may give us processes that are $\alpha$-variant to $\fProc{Q_1}$ and $\fProc{Q'_1}$; in the following we implicitly apply further $\alpha$-renaming to match the names in $\fProc{Q_1}$ and $\fProc{Q'_1}$.
            We consider all three possibilities separately.
            \begin{itemize}

                \item
                    \textbf{(Close-Wait)}
                    We have $\fProc{Q_1} \sc \fCtx{F_1[\fProc{\nu{xy} ( \pClose x[] \| \pWait y() ; R )}]} \redd \fCtx{F_1[ \fProc{R} ]} \sc \fProc{Q'_1}$.
                    Let $\fIFC{e_{\fProc{x}}} := \fIFC{\quasi(\fType{\fIFC{\Omega} \vdash \fIFC{\fProc{\pClose x[]} \at d_{\fProc{x}}} :: \fProc{x}:1\fIFC{[c]}})} = \fIFC{d_{\fProc{x}} \lcup c}$ and $\fIFC{e_{\fProc{y}}} := \fIFC{\quasi(\fType{\fIFC{\Omega} \vdash \fIFC{\fProc{\pWait y() ; R} \at d_{\fProc{y}}} :: \Gamma_{\fProc{y}} , \fProc{y}:\bot\fIFC{[c]}})} = \fIFC{d_{\fProc{y}} \lcup c}$.

                    We distinguish cases on whether or not $\fProc{\pClose x[]} \in \N(\fProc{Q_1})$ (cf.\ \Cref{d:relNode}).
                    \begin{itemize}

                        \item
                            \textbf{($\fProc{\pClose x[]} \in \N(\fProc{Q_1})$)}
                            Since $\fProc{x}$ is bound in $\fProc{Q_1}$, the inclusion of $\fProc{\pClose x[]}$ in $\N(\fProc{Q_1})$ must be because it was added through the inductive clause of \Cref{d:relNode}.
                            We must then have $\{\fProc{x},\fProc{y}\} \in \B(\fProc{Q_1})$ and $\fProc{\pWait y() ; R} \in \N(\fProc{Q_1})$.
                            Then also $\{\fProc{x},\fProc{y}\} \in \B(\fProc{Q_2})$ and $\{\fProc{\pClose x[]} , \fProc{\pWait y() ; R}\} \in \N(\fProc{Q_2})$.
                            Hence, $\fProc{Q_2} \sc \fCtx{F_2[ \fProc{\nu{xy} ( \pClose x[] \| \pWait y() ; R )} ]} \redd \fCtx{F_2[ \fProc{R} ]} := \fProc{Q'_2}$.

                            To conclude the proof, we need to show that $\fProc{Q'_1 \proj \fIFC{\xi}} \sc \fProc{Q'_2 \proj \fIFC{\xi}}$.
                            Let $\fProc{Q''_1} := \fCtx{F_1[\fProc{R}]}$; it suffices to show that $\fProc{Q''_1 \proj \fIFC{\xi}} \sc \fProc{Q'_2 \proj \fIFC{\xi}}$, but now both relevant forms agree on bound names.
                            We have to show that $\B(\fProc{Q''_1}) = \B(\fProc{Q'_2})$ and $\N(\fProc{Q''_1}) = \N(\fProc{Q'_2})$.
                            Both directions of both set inclusions are analogous, so we only detail $b \in \B(\fProc{Q''_1}) \implies b \in \B(\fProc{Q'_2})$ and $n \in \N(\fProc{Q''_1}) \implies n \in \N(\fProc{Q'_2})$.
                            The analysis follows by the inductive definition of $\N(\fProc{Q''_1})$ and $\B(\fProc{Q''_1})$ (cf.\ \Cref{d:relNode}).

                            In the base case, we have some $n \in \N(\fProc{Q''_1})$ added through the interface.
                            The analysis depends on whether $n \in \nodes(\fCtx{F_1})$ or $n \in \nodes(\fProc{R})$.
                            \begin{itemize}

                                \item
                                    If $n \in \nodes(\fCtx{F_1})$, then $n \in \N(\fProc{Q_1})$.
                                    Hence, $n \in \N(\fProc{Q_2})$, so $n \in \nodes(\fCtx{F_2})$.
                                    It follows that $n \in \N(\fProc{Q'_2})$.

                                \item
                                    If $n \in \nodes(\fProc{R})$, it is immediate that $n \in \N(\fProc{Q'_2})$.

                            \end{itemize}

                            In the inductive case, we have some $n \in \N(\fProc{Q''_1})$ added through a binder $b \in \B(\fProc{Q''_1})$ connected to some $n' \in \N(\fProc{Q''_1})$ added in a previous step.
                            The analysis depends on whether $n \in \nodes(\fCtx{F_1})$ or $n \in \nodes(\fProc{R})$.
                            \begin{itemize}

                                \item
                                    If $n \in \nodes(\fCtx{F_1})$, it must be that $b \in \binders(\fCtx{F_1})$.
                                    The analysis depends on whether $n' \in \nodes(\fCtx{F_1})$ or $n' \in \nodes(\fProc{R})$.

                                    If $n' \in \nodes(\fCtx{F_1})$, then $n' \in \N(\fProc{Q_1})$.
                                    Hence, $n' \in \N(\fProc{Q_2})$.
                                    Then also $n \in \N(\fProc{Q_1})$ and $b \in \B(\fProc{Q_1})$, so $n \in \N(\fProc{Q_2})$ and $b \in \B(\fProc{Q_2})$.
                                    This means that $n \in \nodes(\fCtx{F_2})$ and $b \in \binders(\fCtx{F_2})$.
                                    Hence, $n \in \N(\fProc{Q'_2})$ and $b \in \B(\fProc{Q'_2})$.

                                    If $n' \in \nodes(\fProc{R})$, then $\fProc{\pWait x() ; R} \in N(\fProc{Q_1})$ was added through some $b' \in \B(\fProc{Q_1}) \cap \binders(\fCtx{F_1})$.
                                    Hence, $b' \in \B(\fProc{Q_2}) \cap \binders(\fCtx{F_2})$, adding $\fProc{\pWait x() ; R} \in \N(\fProc{Q_2})$.
                                    Therefore, $n' \in \N(\fProc{Q'_2})$.
                                    In conclusion, $n \in \N(\fProc{Q'_2})$ and $b \in \B(\fProc{Q'_2})$.

                                \item
                                    If $n \in \nodes(\fProc{R})$, the analysis depends on whether $b \in \binders(\fCtx{F_1})$ or $b \in \binders(\fProc{R})$.

                                    If $b \in \binders(\fCtx{F_1})$, then $b \in \B(\fProc{Q_1})$ to add $\fProc{\pWait x() ; R} \in \N(\fProc{Q_1})$ via $n' \in \N(\fProc{Q_1}) \setminus \nodes(\fProc{R})$.
                                    Then $b \in \B(\fProc{Q_2})$ and $n' \in \N(\fProc{Q_2}) \setminus \nodes(\fProc{R})$, so $n' \in \nodes(\fCtx{F_2})$.
                                    Hence, $n' \in \N(\fProc{Q'_2})$, and so $n \in \N(\fProc{Q'_2})$ and $b \in \B(\fProc{Q'_2})$.

                                    If $b \in \binders(\fProc{R})$, it follows immediately that $n \in \N(\fProc{Q'_2})$ and $b \in \B(\fProc{Q'_2})$.

                            \end{itemize}

                        \item
                            \textbf{($\fProc{\pClose x[]} \notin \N(\fProc{Q_1})$)}
                            This may be due to either of two reasons.
                            The first possibility is that $\fIFC{e_{\fProc{x}} \not\lleq \xi}$: this would directly violate the requirements for relevancy in \Cref{d:relNode}.
                            The second possibility is that $\fProc{\pWait y() ; R} \notin \N(\fProc{Q_1})$: otherwise, we can follow the reasoning above to contradict the assumption that $\fProc{\pClose x[]} \notin \N(\fProc{Q_1})$.
                            We show that in the former case, the latter case holds as well; i.e., either way, $\fProc{\pWait y() ; R} \notin \N(\fProc{Q_1})$.

                            Assuming $\fIFC{e_{\fProc{x}}} = \fIFC{d_{\fProc{x}} \lcup c \not\lleq \xi}$, either $\fIFC{d_{\fProc{x}} \not\lleq \xi}$ or $\fIFC{c \not\lleq \xi}$.
                            If $\fIFC{d_{\fProc{x}} \not\lleq \xi}$, since Rule~\ruleLabel{typ-close} demands $\fIFC{d_{\fProc{x}} \lleq c}$, also $\fIFC{c \not\lleq \xi}$.
                            Hence, we always have $\fIFC{c \not\lleq \xi}$.
                            Then $\fIFC{e_{\fProc{y}}} = \fIFC{d_{\fProc{y}} \lcup c \not\lleq \xi}$, so $\fProc{\pWait y() ; R} \notin \N(\fProc{Q_1})$.

                            By Rule~\ruleLabel{typ-wait}, $\fType{\fIFC{\Omega} \vdash \fIFC{\fProc{R} \at e_{\fProc{y}}} :: \Gamma_{\fProc{y}}}$.
                            It follows by Rule~\ruleLabel{typ-par} that none of the nodes in $\fProc{R}$ have quasi-running secrecy $\fIFC{\lleq \xi}$.
                            Hence, $\fProc{R}$ cannot affect relevancy of nodes of $\fProc{Q_1}$, so $\fProc{Q_1 \proj \fIFC{\xi}} = \fProc{Q'_1 \proj \fIFC{\xi}}$.
                            Hence, the thesis holds with $\fProc{Q'_2} := \fProc{Q_2}$.

                    \end{itemize}

                \item
                    \textbf{(Send-Receive)}
                    We have $\fProc{Q_1} \sc \fCtx{F_1[\fProc{\nu{xy} ( \pSend x[a,b] \| \pRecv y(w,z) ; R )}]} \redd \fCtx{F_1[\fProc{R \pSubst{ a/w,b/z }}]} \sc \fProc{Q'_1}$.
                    Let $\fIFC{e_{\fProc{x}}} := \fIFC{\quasi(\fType{\fIFC{\Omega} \vdash \fIFC{\fProc{\pSend x[a,b]} \at d_{\fProc{x}}} :: \fProc{x}:A \tensor B\fIFC{[c]} , \fProc{a}:\dual{A}\fIFC{[c]} , \fProc{b}:\dual{B}\fIFC{[c]}})} = \fIFC{d_{\fProc{x}} \lcup c}$ and $\fIFC{e_{\fProc{y}}} := \fIFC{\quasi(\fType{\fIFC{\Omega} \vdash \fIFC{\fProc{\pRecv y(z,w) ; R} \at d_{\fProc{y}}} :: \Gamma_{\fProc{y}} , \fProc{y}:\dual{A} \parr \dual{B}\fIFC{[c]}})} = \fIFC{d_{\fProc{y}} \lcup c}$.

                    We distinguish cases on whether or not $\fProc{\pSend x[a,b]} \in \N(\fProc{Q_1})$.
                    \begin{itemize}

                        \item
                            \textbf{($\fProc{\pSend x[a,b]} \in \N(\fProc{Q_1})$)}
                            Since $\fProc{x}$ is bound in $\fProc{Q_1}$ and $\fProc{a},\fProc{b} \notin \fcn(\fProc{\pSend x[a,b]})$, this inclusion can only hold through the inductive clause in \Cref{d:relNode}.
                            Hence, $\{\fProc{x},\fProc{y}\} \in \B(\fProc{Q_1})$ and $\fProc{\pRecv y(w,z) ; R} \in \N(\fProc{Q_1})$.
                            The rest follows as above.

                        \item
                            \textbf{($\fProc{\pSend x[a,b]} \notin \N(\fProc{Q_1})$)}
                            Analogous to the case above.

                    \end{itemize}

                \item
                    \textbf{(Select-Branch)}
                    Analogous to the case above.

            \end{itemize}

        \item
            \textbf{(Communication between context and process on names not in $\fType{\Gamma}$)}
            By definition, the secrecy levels of the involved names are incomparable to $\fIFC{\xi}$.
            Therefore, none of the nodes involved are relevant or influence relevancy of any other nodes: $\N(\fProc{Q_1}) = \N(\fProc{Q'_1})$ and $\B(\fProc{Q_1}) = \B(\fProc{Q'_1})$.
            Hence, the thesis holds with $\fProc{Q'_2} := \fProc{Q_2}$.
            \qedhere

    \end{itemize}
\end{proof}

\tFundamental*

\begin{proof}
    \label{proof:t:fundamental}
    Let $\fType{\Gamma} := \fType{\Gamma_1 \proj \fIFC{\xi}} = \fType{\Gamma_2 \proj \fIFC{\xi}}$.
    Take any $\fCtx{E_1},\fCtx{E_2}$ such that $\fType{\fIFC{\Omega} \vdash \fIFC{\fCtx{E_1[\fProc{P_1}]} \at d'_1} :: \emptyset}$ and $\fType{\fIFC{\Omega} \vdash \fIFC{\fCtx{E_2[\fProc{P_2}]} \at d'_2} :: \emptyset}$.
    We need to show that $(\fCtx{E_1[\fProc{P_1}]};\fCtx{E_2[\fProc{P_2}]}) \in \termrel{\fIFC{\Omega}}{\fIFC{\xi}}{\fType{\Gamma}}$, which we do by induction on $\w(\fType{\Gamma})$.

    The first condition is that $(\fCtx{E_1},\fProc{P_1} ; \fCtx{E_2},\fProc{P_2}) \in \netw{\fIFC{\Omega};\fIFC{\xi}}(\fType{\Gamma})$; this holds by assumption.

    Next, take any $\fCtx{E'_1},\fProc{P'_1}$ such that $\fCtx{E_1},\fProc{P_1} \uRedd*{\fIFC{\Omega};\fIFC{\xi};\fType{\Gamma}} \fCtx{E'_1},\fProc{P'_1} \nuRedd{\fIFC{\Omega};\fIFC{\xi};\fType{\Gamma}}$.
    A straightforward induction on the length of these unobservable reductions shows that, by \Cref{d:uRedd} and \Cref{l:catchUp}, there are $\fCtx{E'_2},\fProc{P'_2}$ such that $\fCtx{E_2},\fProc{P_2} \uRedd*{\fIFC{\Omega};\fIFC{\xi};\fType{\Gamma}} \fCtx{E'_2},\fProc{P'_2} \nuRedd{\fIFC{\Omega};\fIFC{\xi};\fType{\Gamma}}$, $(\fCtx{E'_1},\fProc{P'_1} ; \fCtx{E'_2},\fProc{P'_2}) \in \netw{\fIFC{\Omega};\fIFC{\xi}}(\fType{\Gamma})$, and $\fProc{P'_1} \obseq_{\fIFC{\xi}} \fProc{P'_2}$.

    Now, we need to show that, for every $\fProc{x} \in \big( \ain(\fCtx{E'_1},\fProc{P'_1}) \cup \ain(\fCtx{E'_2},\fProc{P'_2}) \big) \cap \dom(\fType{\Gamma})$,
    \[
        (\fCtx{E'_1},\fProc{P'_1} ; \fCtx{E'_2},\fProc{P'_2}) \in \valrel{\fIFC{\Omega}}{\fIFC{\xi}}{\fType{\Gamma}}{\fProc{x}}.
    \]
    Take any such $\fProc{x}$.
    Either $\fProc{x} \in \ain(\fCtx{E'_1},\fProc{P'_1})$ or $\fProc{x} \in \ain(\fCtx{E'_2},\fProc{P'_2})$; w.l.o.g., assume the former.
    The rest of the analysis depends on the type of $\fProc{x}$ in $\fType{\Gamma}$.

    First, we discuss the output-like cases ($\fType{1},\fType{\oplus},\fType{\tensor}$).
    In each case, by well-typedness, $\fProc{x}$ is the subject of an output-like prefix in $\fProc{P'_1}$.
    Since $\fProc{x} \in \ain(\fCtx{E'_1},\fProc{P'_1})$, this prefix is unguarded.
    Since $\fProc{x} \in \dom(\fType{\Gamma}) = \dom(\fType{\Gamma_1 \proj \fIFC{\xi}})$, the node in which the prefix appears is relevant in $\fProc{P'_1}$.
    Therefore, since $\fProc{P'_1} \obseq_{\fIFC{\xi}} \fProc{P'_2}$, there is also a relevant node in $\fProc{P'_2}$ where this prefix appears unguarded.
    \begin{itemize}

        \item
            \textbf{($\fProc{x}$ has type $\fType{1\fIFC{[c]}}$)}
            Then $\fProc{P'_1} \sc \fProc{\pClose x[] \| P''_1}$ and $\fProc{P'_2} \sc \fProc{\pClose x[] \| P''_2}$.
            Clearly, $\fType{\fIFC{\Omega} \vdash \fIFC{\fProc{P''_1} \at d''_1} :: \Gamma_1 \setminus \fProc{x}}$ and $\fType{\fIFC{\Omega} \vdash \fIFC{\fProc{P''_2} \at d''_2} :: \Gamma_2 \setminus \fProc{x}}$, and $\fType{\fIFC{\Omega} \vdash \fIFC{\fCtx{E'_1 \big[ \pClose x[] \| \fProc{P''_1} \big]} \at d'''_1} :: \emptyset}$ and $\fType{\fIFC{\Omega} \vdash \fIFC{\fCtx{E'_2 \big[ \pClose x[] \| \fProc{P''_2} \big]} \at d'''_2} :: \emptyset}$.
            Also clearly, $\fProc{P''_1} \obseq_{\fIFC{\xi}} \fProc{P''_2}$ and $\fType{\Gamma_1 \setminus \fProc{x} \proj \fIFC{\xi}} = \fType{\Gamma_2 \setminus \fProc{x} \proj \fIFC{\xi}} = \fType{\Gamma \setminus \fProc{x}}$.
            Hence, since $\w(\fType{\Gamma \setminus \fProc{x}}) < \w(\fType{\Gamma})$, it follows from the IH that $(\fCtx{E'_1\big[ \pClose x[] \| \fProc{P''_1} \big]} ; \fCtx{E'_2\big[ \pClose x[] \| \fProc{P''_2} \big]}) \in \termrel{\fIFC{\Omega}}{\fIFC{\xi}}{\fType{\Gamma \setminus \fProc{x}}}$.
            This proves that $(\fCtx{E'_1},\fProc{P'_1} ; \fCtx{E'_2},\fProc{P'_2}) \in \valrel{\fIFC{\Omega}}{\fIFC{\xi}}{\fType{\Gamma}}{\fProc{x}}$.

        \item
            \textbf{($\fProc{x}$ has type $\fType{\oplus \{ i : A_i \}\fIFC{[c]}}$)}
            Then there exists $\fProc{j} \in \fProc{I}$ such that $\fProc{\pSel x[b_1]<j} \in \nodes(\fProc{P'_1})$ and $\fProc{\pSel x[b_2]<j} \in \nodes(\fProc{P'_1})$.
            The analysis depends on whether $\fProc{b_1} \in \dom(\fType{\Gamma})$ or not.
            \begin{itemize}

                \item
                    \textbf{($\fProc{b_1} \in \dom(\fType{\Gamma})$)}
                    By well-typedness, $\fProc{b_1} \in \fn(\fProc{P'_1}) \cap \fn(\fProc{P'_2})$.
                    Since $\fProc{P'_1} \obseq_{\fIFC{\xi}} \fProc{P'_2}$, then $\fProc{b_1} = \fProc{b_2}$.
                    Hence, $\fProc{P'_1} \sc \fProc{\pSel x[b_1]<j \| P''_1}$ and $\fProc{P'_2} \sc \fProc{\pSel x[b_2]<j \| P''_2}$.
                    Similar to the case above, and since $\w(\fType{\Gamma \setminus \fProc{x}}) < \w(\fType{\Gamma})$, it follows from the IH that $(\fCtx{E'_1\big[ \pSel x[b_1]<j \| \fProc{P''_1} \big]} ; \fCtx{E'_2\big[ \pSel x[b_2]<j \| \fProc{P''_2} \big]}) \in \termrel{\fIFC{\Omega}}{\fIFC{\xi}}{\fType{\Gamma \setminus \fProc{x}}}$.
                    This proves that $(\fCtx{E'_1},\fProc{P'_1} ; \fCtx{E'_2},\fProc{P'_2}) \in \valrel{\fIFC{\Omega}}{\fIFC{\xi}}{\fType{\Gamma}}{\fProc{x}}$.

                \item
                    \sloppy
                    \textbf{($\fProc{b_1} \notin \dom(\fType{\Gamma})$)}
                    By well-typedness, $\fProc{P'_1} \sc \fProc{\nu{b_1b'} ( \pSel x[b_1]<j \| P''_1 )}$.
                    The selection on $\fProc{x}$ is a relevant node of $\fProc{P'_1}$.
                    Since $\fProc{P'_1} \obseq_{\fIFC{\xi}} \fProc{P'_2}$, it is also a relevant node of $\fProc{P'_2}$.
                    Moreover, $\fProc{b_2} \notin \fn(\fProc{P'_2})$: otherwise, $\fProc{b_2} = \fProc{b_1}$, and then $\fProc{b_1} \in \fn(\fProc{P'_1})$.
                    Hence, $\fProc{P'_2} \sc \fProc{\nu{b_2b'} ( \pSel x[b_2]<j \| P''_2 )}$.
                    Clearly, $\fType{\fIFC{\Omega} \vdash \fIFC{\fProc{P''_1} \at d''_1} :: \Gamma_1 \setminus \fProc{x} , \fProc{b'}}$ and $\fType{\fIFC{\Omega} \vdash \fIFC{\fProc{P''_2} \at d''_2} :: \Gamma_2 \setminus \fProc{x} , \fProc{b'}}$, and $\fType{\fIFC{\Omega} \vdash \fIFC{\fCtx{E'_1\big[\nu{b_1b'} ( \pSel x[b_1]<j \| \fProc{P''_1} )\big]} \at d'''_1} :: \emptyset}$ and $\fType{\fIFC{\Omega} \vdash \fIFC{\fCtx{E'_2\big[\nu{b_2b'} ( \pSel x[b_2]<j \| \fProc{P''_2} )\big]} \at d'''_2} :: \emptyset}$.
                    Again, since $\fProc{P'_1} \obseq_{\fIFC{\xi}} \fProc{P'_2}$, the chain of nodes and binders that are relevant in $\fProc{P'_1}$ through the binder $\fProc{\nu{b_1b'}}$ has an equivalent such chain in $\fProc{P'_2}$ through $\fProc{\nu{b_2b'}}$ and the selection on $\fProc{b_2}$.
                    Hence, the effect on relevant nodes and binders by removing the binder and the selection on $\fProc{x}$ is the same on $\fProc{P''_1}$ as it is on $\fProc{P''_2}$: $\fProc{P''_1} \obseq_{\fIFC{\xi}} \fProc{P''_2}$.
                    Clearly, $\fType{\Gamma_1 \setminus \fProc{x} , \fProc{b'} \proj \fIFC{\xi}} = \fType{\Gamma_2 \setminus \fProc{x} , \fProc{b'}} = \fType{\Gamma \setminus \fProc{x} , \fProc{b'}}$.
                    Also, $\w(\fType{A_j}) < \w(\fType{\oplus \{ A_i \}_{i \in I}})$, so $\w(\fType{\Gamma \setminus \fProc{x} , \fProc{b'}}) < \w(\fType{\Gamma})$.
                    It then follows from the IH that $(\fCtx{E'_1\big[\nu{b_1b'} ( \pSel x[b_1]<j \| \fProc{P''_1} )\big]} ; \fCtx{E'_2\big[\nu{b_2b'} ( \pSel x[b_2]<j \| \fProc{P''_2} )\big]} ) \in \termrel{\fIFC{\Omega}}{\fIFC{\xi}}{\fType{\Gamma \setminus \fProc{x} , \fProc{b'}}}$.
                    This proves that $(\fCtx{E'_1},\fProc{P'_1} ; \fCtx{E'_2},\fProc{P'_2}) \in \valrel{\fIFC{\Omega}}{\fIFC{\xi}}{\fType{\Gamma}}{\fProc{x}}$.

            \end{itemize}

        \item
            \textbf{($\fProc{x}$ has type $\fType{A \tensor B\fIFC{[c]}}$)}
            Analogous to the previous case.

    \end{itemize}

    Next, we discuss the negative cases ($\fType{\bot},\fType{\with},\fType{\parr}$).
    In each case, by well-typedness, $\fProc{x}$ is the subject of an input-like prefix in $\fProc{P'_1}$.
    The context $\fCtx{E'_1}$ binds $\fProc{x}$ to some $\fProc{y}$ by restriction, and $\fCtx{E'_1}$ contains a complementary output-like prefix on $\fProc{y}$.
    Following similar reasoning, the same holds for $\fCtx{E'_2}$.
    Since $\fProc{x} \in \ain(\fCtx{E'_1},\fProc{P'_1})$, this output-like prefix appears unguarded in $\fCtx{E'_1}$.
    To prove the thesis, we assume that this prefix also appears unguarded in $\fCtx{E'_2}$.
    \begin{itemize}

        \item
            \sloppy
            \textbf{($\fProc{x}$ has type $\fType{\bot\fIFC{[c]}}$)}
            We have $\fCtx{E'_1} \sc \fCtx{\nu{yx}( \pClose y[] \| E''_1 )}$ and $\fCtx{E'_2} \sc \fCtx{\nu{yx}( \pClose y[] \| E''_2 )}$.
            Let $\fProc{P''_1} := \fProc{\nu{yx}( \pClose y[] \| P'_1 )}$ and $\fProc{P''_2} := \fProc{\nu{yx}( \pClose y[] \| P'_2 )}$.
            Clearly, $\fType{\fIFC{\Omega} \vdash \fIFC{\fProc{P''_1} \at d''_1} :: \Gamma_1 \setminus \fProc{x}}$ and $\fType{\fIFC{\Omega} \vdash \fIFC{\fProc{P''_2} \at d''_2} :: \Gamma_2 \setminus \fProc{x}}$, and $\fType{\fIFC{\Omega} \vdash \fIFC{\fCtx{E''_1[\fProc{P''_1}]} \at d'''_1} :: \emptyset}$ and $\fType{\fIFC{\Omega} \vdash \fIFC{\fCtx{E''_2[\fProc{P''_2}]} \at d'''_2} :: \emptyset}$.

            Let $\fProc{Q_1},\fProc{Q_2}$ denote the nodes of $\fProc{P'_1},\fProc{P'_2}$, respectively, in which $\fProc{x}$ appears.
            To prove that $\fProc{P''_1} \obseq_{\fIFC{\xi}} \fProc{P''_2}$, it suffices to show that $\fProc{Q_1}$ and any related binders are relevant in $\fProc{P''_1}$ if and only if $\fProc{Q_2}$ and any related binders are relevant in $\fProc{P''_2}$; any connected nodes/binders follow similar reasoning.
            We detail only the left-to-right direction; the other direction is analogous.
            Suppose $\fProc{Q_1}$ is relevant in $\fProc{P''_1}$.
            Then $\fIFC{\quasi(\fProc{Q_1}) \lleq \xi}$, and thus $\fProc{Q_1}$ is also relevant in $\fProc{P'_1}$ through $\fProc{x}$ in the interface.
            Then also $\fProc{Q_2}$ is relevant in $\fProc{P'_2}$, where $\fProc{Q_1} \sc \fProc{Q_2}$ and $\fIFC{\quasi(\fProc{Q_2}) \lleq \xi}$.
            The analysis depends on how $\fProc{Q_1}$ is relevant in $\fProc{P''_1}$: (i)~through the interface, or (ii)~through a restriction with another relevant node.
            In case~(i), it follows straightforwardly that $\fProc{Q_2}$ is also relevant in $\fProc{P'_2}$.
            In case~(ii), the connected node is also relevant in $\fProc{P'_1}$, and hence there is a related node that is also relevant in $\fProc{P'_2}$.
            Since the two processes agree on observable channels, the channel responsible for including $\fProc{Q_1}$ as a relevant node of $\fProc{P''_1}$ is also bound in $\fProc{P''_2}$.
            Then we can conclude that $\fProc{Q_2}$ is a relevant node of $\fProc{P''_2}$.

            Since $\w(\fType{\Gamma \setminus \fProc{x}}) < \w(\fType{\Gamma})$, it then follows from the IH that $(\fCtx{E''_1[\fProc{P''_1}]} ; \fCtx{E''_2[\fProc{P''_2}]}) \in \termrel{\fIFC{\Omega}}{\fIFC{\xi}}{\fType{\Gamma \setminus \fProc{x}}}$.
            This proves that $(\fCtx{E'_1},\fProc{P'_1} ; \fCtx{E'_2},\fProc{P'_2}) \in \valrel{\fIFC{\Omega}}{\fIFC{\xi}}{\fType{\Gamma}}{\fProc{x}}$.

        \item
            \textbf{($\fProc{x}$ has type $\fType{\with \{ i : A_i \}_{i \in I}\fIFC{[c]}}$)}
            We have $\fCtx{E'_1} \sc \fCtx{\nu{bb'} \nu{yx} ( \pSel y[b]<j \| E''_1 )}$ and $\fCtx{E'_2} \sc \fCtx{\nu{bb'} \nu{yx} ( \pSel y[b]<j \| E''_2 )}$, where $\fProc{j} \in \fProc{I}$.
            Let $\fProc{P''_1} := \fProc{\nu{yx} ( \pSel y[b]<j \| P'_1 )}$ and $\fProc{P''_2} := \fProc{\nu{yx} ( \pSel y[b]<j \| P'_2 )}$.
            Clearly, $\fType{\fIFC{\Omega} \vdash \fIFC{\fProc{P''_1} \at d''_1} :: \Gamma_1 \setminus \fProc{x} , \fProc{b}:A_j\fIFC{[c]}}$ and $\fType{\fIFC{\Omega} \vdash \fIFC{\fProc{P''_2} \at d''_2} :: \Gamma_2 \setminus \fProc{x} , \fProc{b}:A_j\fIFC{[c]}}$, and $\fType{\fIFC{\Omega} \vdash \fIFC{\fCtx{\nu{bb'} E''_1[\fProc{P''_1}]} \at d'''_1} :: \emptyset}$ and $\fType{\fIFC{\Omega} \vdash \fIFC{\fCtx{\nu{bb'} E''_2[\fProc{P''_2}]} \at d'''_2} :: \emptyset}$.

            It is clear that $\fProc{b} \notin \fcn(\fProc{P''_1}) \cup \fcn(\fProc{P''_2})$; hence, $\fProc{b}$ plays no role in relevancy in either process.
            Therefore, proving that $\fProc{P''_1} \obseq_{\fIFC{\xi}} \fProc{P''_2}$ is analogous to the previous case.
            Since $\w(\fType{A_j}) < \w(\fType{\with \{ i : A_i \}_{i \in I}})$, we have $\w(\fType{\Gamma \setminus \fProc{x} , \fProc{b}:A_j\fIFC{[c]}}) < \w(\fType{\Gamma})$.
            Therefore, by the IH, $(\fCtx{\nu{bb'} E''_1[\fProc{P''_1}]} ; \fCtx{\nu{bb'} E''_2[\fProc{P''_2}]}) \in \termrel{\fIFC{\Omega}}{\fIFC{\xi}}{\fType{\Gamma \setminus \fProc{x} , \fProc{b}:A_j\fIFC{[c]}}}$.
            This proves that $(\fCtx{\nu{bb'} E''_1},\fProc{P''_1} ; \fCtx{\nu{bb'} E''_2},\fProc{P''_2}) \in \valrel{\fIFC{\Omega}}{\fIFC{\xi}}{\fType{\Gamma}}{\fProc{x}}$.

        \item
            \textbf{($\fProc{x}$ has type $\fType{A \parr B\fIFC{[c]}}$)}
            Analogous to the previous case.

    \end{itemize}

    Finally, we show that $\aon(\fProc{P'_1}) \cap \dom(\fType{\Gamma}) = \aon(\fProc{P'_2}) \cap \dom(\fType{\Gamma})$.
    To prove this set equality, we take any $\fProc{x} \in \aon(\fProc{P'_1}) \cap \dom(\fType{\Gamma})$ and prove that $\fProc{x} \in \aon(\fProc{P'_2}) \cap \dom(\fType{\Gamma})$; the other direction is analogous.
    Clearly, $\fProc{x}$ is the subject of an output-like prefix in $\fProc{P'_1}$.
    Since $\fProc{x} \in \dom(\fType{\Gamma})$, this output-like prefix must appear unguarded in a node in $\fProc{P'_1}$.
    If the quasi-running secrecy of this node is observable, this node is relevant in $\fProc{P'_1}$.
    Since $\fProc{P'_1} \obseq_{\fIFC{\xi}} \fProc{P'_2}$, $\fProc{P'_2}$ must also have a relevant node in which the output-like prefix appears unguarded.
    Otherwise, the node is not relevant in $\fProc{P'_1}$, and hence the node in which the output-like prefix appears in $\fProc{P'_2}$ is also not relevant in $\fProc{P'_2}$.
    Hence, $\fProc{x} \in \aon(\fProc{P'_2}) \cap \dom(\fType{\Gamma})$.
\end{proof}

\end{document}